\documentclass[acmsmall]{acmart}

\setcopyright{none}

\usepackage{mathtools}
\usepackage{qcircuit}
\usepackage{booktabs}   
\usepackage{subcaption} 

\begin{document}

\title[Quantum Hoare logic with classical variables]{Quantum Hoare logic with classical variables}         


\author{Yuan Feng}
\orcid{nnnn-nnnn-nnnn-nnnn}             
\affiliation{
  \position{Professor}
  \department{Centre for Quantum Software and Information}              
  \institution{University of Technology Sydney}            
  \state{NSW}
  \country{Australia}                    
}

\email{Yuan.Feng@uts.edu.au}          

\author{Mingsheng Ying}
\orcid{nnnn-nnnn-nnnn-nnnn}             
\affiliation{
  \position{Professor}
  \department{Centre for Quantum Software and Information}             
  \institution{University of Technology Sydney}           
  \state{NSW}
  \country{Australia}                   
}

\affiliation{
	\position{Professor}
	\department{Institute of Software}             
	\institution{Chinese Academy of Sciences}           
	\city{Beijing}
	\country{China}                   
}

\affiliation{
  \position{Professor}
  \department{Department of Computer Science}             
  \institution{Tsinghua University}           
  \city{Beijing}
  \country{China}                   
}
\email{Mingsheng.Ying@uts.edu.au}         

\begin{abstract}

Hoare logic provides a syntax-oriented method to reason about program correctness, and has been proven effective in the verification of classical and probabilistic programs. Existing proposals for quantum Hoare logic either lack completeness or support only quantum variables, thus limiting their capability in practical use.
In this paper, we propose a quantum Hoare logic for a simple while language which involves both 
classical and quantum variables. Its soundness and relative completeness are proven for both partial and total correctness of quantum programs written in the language. Remarkably, with novel definitions of classical-quantum states and corresponding assertions, the logic system is quite simple and similar to the traditional Hoare logic for classical programs. Furthermore, to simplify reasoning in real applications, auxiliary proof rules are provided which support standard logical operation in the classical part of assertions, and of super-operator application in the quantum part. Finally, a series of practical quantum algorithms, in particular the whole algorithm of Shor's factorisation, are formally verified to show the effectiveness of the logic.

\end{abstract}



\setcopyright{acmcopyright}

\maketitle

\newcommand {\empstr} {\Lambda}

\newcommand {\qcf}[1] {{\sf{#1}}}

\newcommand {\qc}[1] {{\sf{#1}}}
\def\>{\ensuremath{\rangle}}
\def\<{\ensuremath{\langle}}
\def\sl {\ensuremath{\llparenthesis}}
\def\sr{\ensuremath{\rrparenthesis}}
\def\-{\ensuremath{\textrm{-}}}
\def\ott{t}
\def\otu{u}
\def\ots{s}
\def\apply{\mathrel{*\!\!=}}

\def\comm{\ensuremath{\leftrightarrow^*}}
\def\reach{\ensuremath{\rightarrow^*}}

\def\ctp{P}
\def\ctq{Q}

\def\qVar{\ensuremath{\mathit{qVar}}}
\def\Var{\ensuremath{\mathit{Var}}}

\def\fdmu{\Delta}
\def\fdnu{\dnu}
\def\fdomega{\domega}

\def\dmu{\mu}
\def\dnu{\nu}
\def\domega{\omega}
\def\expect{\mathbb{E}}
\def\preexpect{\mathrm{pre}\mathbb{E}}

\def\rassign{:=_{\$}}
\def\fpi{\widehat{\pi}}
\def\h{\ensuremath{\mathcal{H}}}
\def\p{\ensuremath{\mathcal{P}}}
\def\l{\ensuremath{\mathcal{L}}}
\def\g{\ensuremath{\mathcal{G}}}
\def\lh{\ensuremath{\mathcal{L(H)}}}
\def\dh{\ensuremath{\mathcal{D(H})}}
\def\dhv{\ensuremath{\d(\h_V)}}
\def\q{\bold Q}
\def\Q{\ensuremath{\mathbb Q}}
\def\P{\ensuremath{\mathbb P}}
\def\SO{\ensuremath{\mathcal{SO}}}
\def\HP{\ensuremath{\mathcal{HP}}}
\def\hpe{\ensuremath{\mathcal{\e}}}

\def\r{\ensuremath{\mathcal{R}}}
\def\R{\ensuremath{\mathbb{R}}}
\def\m{\ensuremath{\mathcal{M}}}
\def\u{\ensuremath{\mathcal{U}}}
\def\k{\ensuremath{\mathcal{K}}}
\def\K{\ensuremath{\mathfrak{K}}}
\def\S{\ensuremath{\mathfrak{S}}}
\def\s{\ensuremath{\mathcal{S}}}
\def\t{\ensuremath{\mathcal{T}}}
\def\u{\ensuremath{\mathcal{U}}}
\def\U{\ensuremath{\mathfrak{U}}}
\def\L{\ensuremath{\mathfrak{L}}}
\def\x{\ensuremath{\mathcal{X}}}
\def\y{\ensuremath{\mathcal{Y}}}
\def\z{\ensuremath{\mathcal{Z}}}
\def\v{\ensuremath{\mathcal{V}}}

\def\st{\ensuremath{\mathfrak{t}}}
\def\su{\ensuremath{\mathfrak{u}}}
\def\ss{\ensuremath{\mathfrak{s}}}

\def\ra{\ensuremath{\rightarrow}}
\def\a{\ensuremath{\mathcal{A}}}
\def\b{\ensuremath{\mathcal{B}}}
\def\c{\ensuremath{\mathcal{C}}}

\def\e{\ensuremath{\mathcal{E}}}
\def\f{\ensuremath{\mathcal{F}}}
\def\l{\ensuremath{\mathcal{L}}}
\def\X{\mbox{\bf{X}}}
\def\N{\mathbb{N}}
\def\sreal{\mathbb{R}}
\def\Z{\mathbb{Z}}

\def\qzz{\ensuremath{|0\>_q\<0|}}
\def\qoo{\ensuremath{|1\>_q\<1|}}
\def\qzo{\ensuremath{|0\>_q\<1|}}
\def\qoz{\ensuremath{|1\>_q\<0|}}
\def\qii{\ensuremath{|i\>_q\<i|}}
\def\qiz{\ensuremath{|i\>_q\<0|}}
\def\qzi{\ensuremath{|0\>_q\<i|}}

\def\quzz{\ensuremath{|0\>_{\bar{q}}\<0|}}
\def\quoo{\ensuremath{|1\>_{\bar{q}}\<1|}}
\def\quzo{\ensuremath{|0\>_{\bar{q}}\<1|}}
\def\quoz{\ensuremath{|1\>_{\bar{q}}\<0|}}
\def\quii{\ensuremath{|i\>_{\bar{q}}\<i|}}
\def\quiz{\ensuremath{|i\>_{\bar{q}}\<0|}}
\def\quzi{\ensuremath{|0\>_{\bar{q}}\<i|}}

\DeclarePairedDelimiter{\ceil}{\lceil}{\rceil}

\def\d{\ensuremath{\mathcal{D}}}
\def\dh{\ensuremath{\mathcal{D(H)}}}
\def\lh{\ensuremath{\mathcal{L(H)}}}
\def\le{\ensuremath{\sqsubseteq}}
\def\ge{\ensuremath{\sqsupseteq}}
\def\eval{\ensuremath{{\psi}}}
\def\aeq{\ensuremath{{\ \equiv\ }}}
\def\osnt{\ensuremath{\sl \ott, \e\sr}}
\def\snt{\st}
\def\snti{\ensuremath{\sl \ott_i, \e_i\sr}}
\def\osnu{\ensuremath{\sl \otu, \f\sr}}
\def\osns{\ensuremath{\sl s, \g\sr}}
\def\snu{\su}
\def\fdist{\ensuremath{\d ist_\h}}
\def\dist{\ensuremath{Dist}}
\def\wtx{\ensuremath{\widetilde{X}}}

\def\bv{1{v}}
\def\bV{\mathbf{V}}
\def\bf{\mathbf{f}}
\def\bw{\mathbf{w}}
\def\zo{\mathbf{0}}
\def\bX{\mathbf{X}}
\def\bDelta{\mathbf{\Delta}}
\def\bdelta{\boldsymbol{\delta}}
\def\next{\mathcal{X}}
\def\until{\mathcal{U}}

\def\leqI{\ensuremath{\mathcal{SI}(\h)}}
\def\leqIq{\ensuremath{\mathcal{SI}_{\eqsim}(\h)}}
\def\oact{\ensuremath{\alpha}}
\def\oactb{\ensuremath{\beta}}
\def\sact{\ensuremath{\gamma}}
\def\fpi{\ensuremath{\widehat{\pi}}}
\newcommand{\supp}[1]{\ensuremath{\lceil{#1}\rceil}}
\newcommand{\support}[1]{\lceil{#1}\rceil}

\newcommand{\abis}{\stackrel{\lambda}\approx}
\newcommand{\abisa}[1]{\stackrel{#1}\approx}
\newcommand {\qbit} {\mbox{\bf{new}}}

\renewcommand{\theenumi}{(\arabic{enumi})}
\renewcommand{\labelenumi}{\theenumi}
\newcommand{\tr}{{\rm tr}}
\newcommand{\rto}[1]{\stackrel{#1}\rightarrow}
\newcommand{\orto}[1]{\stackrel{#1}\longrightarrow}
\newcommand{\srto}[1]{\stackrel{#1}\longmapsto}
\newcommand{\sRto}[1]{\stackrel{#1}\Longmapsto}

\newcommand{\ass}[3]{\left\{#1\right\}\ #2\ \left\{#3\right\}}
\newcommand{\andor}{\ \&\ }

\newcommand {\true} {\ensuremath{{\mathbf{true}}}}
\newcommand {\false} {\ensuremath{{\mathbf{false}}}}
\newcommand {\abort}{\ensuremath{{\mathbf{abort}}}}
\newcommand {\sskip} {\mathbf{skip}}

\newcommand {\then} {\ensuremath{\mathbf{then}}}
\newcommand {\eelse} {\ensuremath{\mathbf{else}}}
\newcommand {\while} {\ensuremath{\mathbf{while}}}
\newcommand {\ddo} {\ensuremath{\mathbf{do}}}
\newcommand {\pend} {\ensuremath{\mathbf{end}}}
\newcommand {\inv} {\ensuremath{\mathbf{inv}}}

\renewcommand {\measure} {\mathbf{meas}}

\newcommand {\iif} {\mathbf{if}}
\def\mstm{\iif\ b\ \then\ S_1\ \eelse\ S_0\ \pend}
\def\wstm{\while\ b\ \ddo\ S\ \pend}
\newcommand\measstm[3]{\iif\ #1\ \then\ #2\ \eelse\ #3\ \pend}
\newcommand\whilestm[2]{\while\ #1\ \ddo\ #2\ \pend}

\newcommand {\spann} {\mathrm{span}}

\newcommand{\rrto}[1]{\xhookrightarrow{#1}}
\newcommand{\con}[3]{\iif\ {#1}\ \then\ {#2}\ \eelse\ {#3}}

\newcommand{\Rto}[1]{\stackrel{#1}\Longrightarrow}
\newcommand{\nrto}[1]{\stackrel{#1}\nrightarrow}

\newcommand{\Rhto}[1]{\stackrel{\widehat{#1}}\Longrightarrow}
\newcommand{\define}{\ensuremath{\triangleq}}
\newcommand{\rsim}{\simeq}
\newcommand{\obis}{\approx_o}
\newcommand{\sbis}{\ \dot\approx\ } 
\newcommand{\stbis}{\ \dot\sim\ } 
\newcommand{\nssbis}{\ \dot\nsim\ } 

\newcommand{\bis}{\sim}
\newcommand{\rat}{\rightarrowtail}
\newcommand{\wbis}{\approx}
\newcommand{\id}{\mathcal{I}}
\newcommand{\stet}[1]{\{ {#1}  \}  } 
\newcommand{\unw}[1]{\stackrel{{#1}}\sim}
\newcommand{\rma}[1]{\stackrel{{#1}}\approx}

\def\step{\textsf{step}}
\def\obs{\textsf{obs}}
\def\dom{\textsf{dom}}
\def\purge{\textsf{ipurge}}
\def\source{\textsf{sources}}
\def\cnt{\textsf{cnt}}
\def\read{\textsf{read}}
\def\alter{\textsf{alter}}
\def\dirac#1{\delta_{#1}}

\def\tybool{\ensuremath{\mathbf{Boolean}}}
\def\tyint{\ensuremath{\mathbf{Integer}}}
\def\tyqubit{\ensuremath{\mathbf{Qubit}}}
\def\tyqudit{\ensuremath{\mathbf{Qudit}}}
\def\tyqureg{\ensuremath{\mathbf{Qureg}}}
\def\tyunitreg{\ensuremath{\mathbf{Unitreg}}}
\def\type{\ensuremath{\mathit{type}}}

\def\qstate{\Delta}
\def\qassert{\Theta}
\def\qassertp{\Psi}
\def\casserts{\a}
\def\cstate{\sigma}
\def\cstates{\Sigma}
\def\cassert{p}
\def\emptydis{\bot}
\def\qset{Q}
\def\qsetp{R}
\def\qv{{qv}}
\def\Exp{\mathrm{Exp}}

\def\qstates{\s_V}
\def\qasserts{\a_V}
\def\qstatesh#1{\mathcal{S}_{#1}}
\def\qassertsh#1{\mathcal{A}_{#1}}

\def\qstatesp{\mathcal{S}(\h')}

\newcommand\prog{\mathit{Prog}}
\def\ph{\ensuremath{\mathcal{P}(\h)}}
\def\phv{\ensuremath{\mathcal{P}(\h_V)}}

\def\<{\langle}
\def\>{\rangle}
\def\l{\mathcal{L}}
\def\k{\mathcal{K}}
\def\qmc {\color{red}}
\def\dtmc {\color{black}}
\newcommand{\ysim}[1]{\stackrel{#1}\sim}
\def\z{\mathbf{0}}
\newcommand{\TRANDA}[3]{#1\xrightarrow{#2}_{{\sf D}}#3}
\def\pdist{\mathit{pDist}}

\def\C{\mathbb{C}}

\newcommand{\subs}[2]{{#2}/{#1}}

\def \Rm#1{\mbox{\rm #1}}
\def \lsem      {\raise1pt\hbox{\Rm {[\kern-.12em[}}}
\def \rsem      {\raise1pt\hbox{\Rm {]\kern-.12em]}}}
\def \sem#1{\mbox{\lsem$#1$\rsem}}

\newtheorem{remark}{Remark}

\section{Introduction}

Quantum computing and quantum communication provide potential speed-up and enhanced security compared with their classical counterparts~\cite{grover1996fast,shor1994algorithms,harrow2009quantum,bennett1984quantum,bennett1992quantum}. 
However, the quantum features which are responsible for these benefits, such as entanglement between different systems and non-commutativity of quantum operations, also make analysis of quantum algorithms and protocols notoriously difficult~\cite{mayers2001unconditional}. Furthermore, due to the lack of reliable and scalable hardware on which practical quantum algorithms can be executed, traditional techniques such as testing and debugging in classical software engineering will not be readily available in the near future, and formal methods based static analysis of quantum programs seems indispensable.

Among other techniques, Hoare logic provides a syntax-oriented proof system to reason about program correctness~\cite{hoare1969axiomatic}. For classical (non-probabilistic) programs, the correctness is expressed in the Hoare triple form $\{P\} S \{Q\}$ where $S$ is a program, and $P$ and $Q$ are first-order logic formulas called \emph{assertions} that describe the pre- and post-conditions of $S$, respectively. Intuitively, the triple claims that if $S$ is executed at a \emph{state} (evaluation of program variables) satisfying $P$ and it terminates, then $Q$ must hold in the final state. This is called \emph{partial correctness}. If termination is further guaranteed in all states that satisfy $P$, then partial correctness becomes a \emph{total} one. After decades of development, Hoare logic has been successfully applied in analysis of programs with non-determinism, recursion, parallel execution, etc. For a detailed survey, we refer to~\cite{apt2019fifty,apt2010verification}. 

Hoare logic was also extended to programming languages with probabilistic features. As the program states for probabilistic languages are (sub)distributions over evaluations of program variables, the extension naturally follows two different approaches, depending on how assertions of probabilistic states are defined. The first one
takes subsets of distributions as (qualitative) assertions, similar to the non-probabilistic case, and the satisfaction relation between distributions and assertions is then just the ordinary membership~\cite{ramshaw1979formalizing,den2002verifying,chadha2007reasoning,barthe2018assertion}. In contrast, the other approach takes
non-negative functions on evaluations as (quantitative) assertions. Consequently, one is concerned with the \emph{expectation} of a distribution satisfying an assertion~\cite{morgan1996probabilistic,mciver2005abstraction,olmedo2016reasoning,kozen1981semantics,kozen1985probabilistic}. 

In recent years, Hoare logic and relational Hoare logic for quantum programs have been developed, also following two different approaches similar to the probabilistic setting. Note that quantum (mixed) states are described mathematically by density operators in a Hilbert space. 
Assertions in the satisfaction-based logics proposed in~\cite{chadha2006reasoning,Kakutani:2009} extend the probabilistic counterparts in~\cite{den2002verifying,chadha2006reasoningr} with the ability to 
reason about probabilities (or even the complex amplitudes) and expected values of measuring a quantum state.  
The satisfaction-based logics proposed in ~\cite{zhou2019applied,unruh2019quantum,unruh2019quantumr} regard subspaces of the Hilbert space as assertions, and a quantum state $\rho$ satisfies an assertion $P$ iff the support (the image space of linear operators) of $\rho$ is included in $P$.
In contrast, the expectation-based approaches~\cite{ying2012floyd,barthe2019relational,li2019quantum,ying2018reasoning,ying2016foundations,ying2019toward} take positive operators as assertions for quantum states, following the observation of~\cite{d2006quantum}, and the expectation of a quantum state $\rho$ satisfying an assertion $M$ is then defined to be $\tr(M\rho)$.
A comparison of the quantum Hoare logics in~\cite{chadha2006reasoning,ying2012floyd,Kakutani:2009}  was provided in~\cite{rand2019verification}.


The logics proposed in~\cite{chadha2006reasoning,Kakutani:2009} support classical variables in the language. However, whether or not they are complete is still unknown. Completeness of the logic for a purely quantum language in~\cite{unruh2019quantum} has not been established either. 
On the other hand, the quantum Hoare logics in~\cite{ying2012floyd,ying2016foundations,ying2018reasoning,zhou2019applied} are complete, but the programming languages they consider do not natively support classical variables. Although infinite dimensional quantum variables are provided which are able to encode classical data like integers, in practice it is inconvenient (if possible) to specify and reason about properties in infinite dimensional Hilbert spaces. The subspace assertion in~\cite{zhou2019applied,unruh2019quantum} 
makes it easy to describe and determine properties of quantum programs, but the expressive power of the assertions is limited: they only assert if a given quantum state lies completely within a subspace. Consequently,  quantum algorithms which succeed with certain probability cannot be verified in their logics.

\textbf{Contribution of the current paper}: 
Our main contribution is a sound and relatively complete Hoare logic for a simple while-language where both classical and quantum variables are involved. The expressiveness and effectiveness of our logic are demonstrated by formally specifying and verifying Shor's factorisation algorithm~\cite{shor1994algorithms} and its related subroutines such as quantum Fourier transform, phase estimation, and order finding algorithms. To the best of our knowledge, this is the first time quantum Hoare logic is applied on verification of the whole algorithm of Shor's factorisation.  

Our work distinguishes itself from the works on quantum Hoare logic mentioned above in the following aspects:
\begin{enumerate}
	\item
	\emph{Programming language}.  
	The language considered in this paper supports both classical variables with infinite domains (e.g. the set of integers) and quantum variables. In contrast, the programming languages in~\cite{chadha2006reasoning,Kakutani:2009} allow only a finite variant of integer-type (and bounded iteration for ~\cite{chadha2006reasoning}), while only quantum variables are considered in~\cite{ying2012floyd,barthe2019relational,li2019quantum,zhou2019applied,unruh2019quantum,unruh2019quantumr,ying2018reasoning,ying2016foundations,ying2019toward}. 
	
	\item\emph{Classical-quantum states}. We define program states of our quantum language to be mappings from classical evaluations to partial density operators. This notion of \emph{positive-operator valued distribution} is a direct extension of probability distribution in the probabilistic setting, and often simplifies both specification and verification of program correctness, compared with the way adopted in~\cite{chadha2006reasoning} of regarding probability distributions over pairs of classical evaluation and quantum pure state as classical-quantum states. Note also that if only boolean-type classical variables and qubit-type quantum variables are considered, our definition coincides with the one in~\cite{selinger2004towards}.
	
	\item\emph{Classical-quantum assertions}. 
	Accordingly, assertions for the classical-quantum program states are defined to be mappings from classical evaluations to positive  operators, analogous to discrete random variables in the probabilistic case~\cite{morgan1996probabilistic}. This follows the expectation-based approach in~\cite{ying2012floyd,barthe2019relational,li2019quantum,ying2018reasoning,ying2016foundations,ying2019toward}. However, we also require that the preimage of each positive operator under the mapping be characterised by a classical first-order logic formula. Thus our definition of assertions is essentially in a \emph{hybrid} style, combining the satisfaction-based approach for the classical part and the expectation-based one for the quantum part. 
	
	\item\emph{A simpler quantum Hoare logic}.
	Thanks to the novel definition of classical-quantum states and assertions,  our quantum Hoare logic is much simpler and similar to the traditional Hoare logic, compared with those in~\cite{chadha2006reasoning,Kakutani:2009} for classical-quantum languages. Furthermore, since the language we consider includes probabilistic assignments, it provides a sound and relatively complete Hoare logic for probabilistic programs as a by-product.
	
	\item \emph{Auxiliary rules}. In addition to the sound and complete proof system, various auxiliary proof rules are provided to simplify reasoning in real applications. These include the standard disjunction, invariance, and existential quantifier introduction rules for the classical part of the assertions, and super-operator application for the quantum part.
	In particular, the (ProbComp) rule plays an essential role in verification of quantum algorithms which succeed with a certain probability. These rules turn out to be useful, as illustrated by a series of examples including Grover's search algorithm and Shor's factorisation algorithm.
\end{enumerate}

The paper is organised as follows. In the remainder of this section, related work on quantum Hoare logic is further discussed in detail. We review in Sec. 2 some basic notions from linear algebra and quantum mechanics that will be used in this paper. Classical-quantum states and assertions, which serve as the basis for the semantics and correctness of quantum programs, are defined in Sec. 3. The quantum programming language that we are concerned with is introduced in Sec. 4. A structural operational semantics, a denotational semantics, and a weakest (liberal) precondition semantics are also defined there. Sec. 5 is devoted to a Hoare logic for quantum programs written in our language, where proof rules for both partial and total correctness are proposed. These proof systems are shown to be both sound and relatively complete with respect to their corresponding correctness semantics. Auxiliary proof rules are presented in Sec.~\ref{sec:aux} to help reasoning in real applications. In addition to the running example of Grover's algorithm, verification of quantum Fourier transform, phase estimation, order finding, and Shor's algorithm are provided in Sec.~\ref{sec:case} to illustrate the expressiveness of our language as well as the effectiveness of the proposed Hoare logic. Finally, Sec. 8 concludes the paper and points out some directions for future study.

\subsection{Related work}

Although the first quantum programming languages traced back to~\cite{Om98,sanders2000quantum,bettelli2003toward}, Selinger's seminal paper~\cite{selinger2004towards} proposed for the first time a rigorous semantics for a simple quantum language QPL. The syntax of our language is heavily influenced by Selinger's work. We also borrow from him the idea of using partial density operators (i.e., not normalising them at each computational step) to describe quantum states. This convention simplifies both notationally and conceptually the semantics of quantum languages,  especially the description of non-termination. Our language excludes general recursion and procedure call from QPL, but includes $\tyint$ as a classical data type. Consequently, the semantic model in~\cite{selinger2004towards}, which takes finite tuples (indexed by evaluations of $\tybool$ variables in the program) of partial density operators as program states, does not apply directly to our language considered in this paper. Instead, we extend the `tuples of matrices' notion to matrix-valued functions with countable supports to denote classical-quantum states; see Sec.~\ref{sec:cqstates} for details.

An Ensemble Exogenous Quantum Propositional Logic (EEQPL) was proposed in~\cite{chadha2006reasoning} for a simple quantum language with bounded $\tyint$ type and bounded iteration. In contrast with Selinger's approach, program states of the language are probability sub-distributions over pairs of classical evaluation and quantum pure state. EEQPL has the ability of reasoning about  amplitudes of quantum states. This makes it very strong in expressiveness, but also hinders its use in applications such as debugging, as amplitudes of quantum states are not physically accessible through measurements. The soundness and (weak) completeness of EEQPL is proven in a special case where all real and complex values involved range over a finite set. General completeness result has not been reported.
A qualitative Hoare logic called QHL for Selinger's QPL (again, without general recursion and procedure call) was proposed in~\cite{Kakutani:2009}. The assertion language of QHL is an extended first-order logic with the primitives of applying a matrix on a set of qubits and computing the probability that a classical predicate is satisfied by the outcome of a quantum measurement. The proof system of QHL is sound, but no completeness result was established.

The idea of taking hermitian operators as quantum assertions was first proposed in~\cite{d2006quantum}, which paves the way for expectation-based reasoning about quantum programs. The notion of quantum weakest precondition was also proposed in the same paper in a language-independent manner. Based on these notions, a sound and relatively complete Hoare logic was proposed in~\cite{ying2012floyd} for a quantum while language where only quantum variables are involved. 
The operational semantics of our language, as well as the way the denotational one is derived from it, are inspired by~\cite{ying2012floyd}. Some auxiliary proof rules presented in Sec. 6 are motivated by~\cite{ying2019toward}.

The logic in~\cite{ying2012floyd} does not natively support classical variables. Instead, it allows quantum variables to be of (countably) infinite dimension, thus providing a way to encode classical types like $\tyint$ into quantum states. In contrast, our language explicitly includes classical data types, but only allows $\tyqudit$ (associated with a $d$-dimensional Hilbert space, where $d$ is an arbitrary but finite integer) for quantum variables. Including classical variables makes the description and verification of quantum algorithms easier and more natural, while excluding infinite dimensional quantum variables avoids the mathematical difficulties of dealing with infinite dimensional Hilbert spaces. To illustrate this, one may compare the correctness proofs of Grover's search algorithm in~\cite{ying2012floyd} and the current paper.

A restricted version of~\cite{ying2012floyd}, called applied quantum Hoare logic (aQHL), was proposed in~\cite{zhou2019applied} where quantum predicates are restricted to be projections, instead of general hermitian operators, with the purpose of simplifying its use in debugging and testing. To reason about robustness of quantum programs in aQHL, the (qualitative) satisfaction relation of a quantum state $\rho$ with respect to a projection $P$ is extended to an approximate one $\rho\models_\epsilon P$ for a given error bound $\epsilon$. However, this approximate satisfaction is quite different from the quantitative relation of~\cite{ying2012floyd} (and that in the current paper) which is determined by the expectation $\tr(P\rho)$: the former claims that $\rho$ is $\epsilon$-close to some state in $P$, which is not physically checkable by a quantum measurement; while the latter is the expected value of measuring $\rho$ using the projective measurement $\{P, I-P\}$.

The quantum Hoare logic in~\cite{ying2012floyd} has been implemented on Isabelle/HOL~\cite{liu2019formal}. It was also used in~\cite{hung2019quantitative} to reason about robustness of quantum programs against noise during execution, and extended in~\cite{ying2018reasoning} for analysis of parallel quantum programs. 
A quantum Hoare logic with ghost variables is introduced in~\cite{unruh2019quantum}. Interestingly, by introducing the ghost variables, one can express properties such as a quantum variable is unentangled with others. The logic is shown to be sound, but again, no completeness result is provided.

\section{Preliminaries}

This section is devoted to fixing some notations from linear algebra and quantum mechanics that will be used in this paper. For a thorough introduction of relevant backgrounds, we refer to~\cite[Chapter 2]{nielsen2002quantum}.

\subsection{Basic linear algebra}
Let $\h$ be a Hilbert space. In the finite-dimensional case which we are concerned with here, it is merely a complex linear space equipped with an inner product. Consequently, it is isomorphic to $\C^d$ where $d=\dim(\h)$, the dimension of $\h$.
Following the tradition in quantum computing, vectors in $\h$ are denoted in the Dirac form $|\psi\>$. The inner product of $|\psi\>$ and $|\phi\>$ is written $\<\psi|\phi\>$, and they are \emph{orthogonal} if $\<\psi|\phi\> = 0$. The \emph{outer product} of them, denoted $|\psi\>\<\phi|$, is a rank-one linear operator which maps any $|\psi'\>$ in $\h$ to $\<\phi|\psi'\> |\psi\>$.
The \emph{length} of $|\psi\>$ is defined to be $\||\psi\>\| \define\sqrt{\<\psi|\psi\>}$ and it is called \emph{normalised} if $\||\psi\>\|=1$. A set of vectors $B\define\{|i\> : i\in I\}$ in $\h$ is \emph{orthonormal} if each $|i\>$ is normalised and every two of them are orthogonal. Furthermore, if they span the whole space $\h$; that is, any vector in $\h$ can be written as a linear combination of vectors in $B$, then $B$ is called an \emph{orthonormal basis} of $\h$. 

Let $\lh$ be the set of linear operators on $\h$, and $\z_\h$ and $I_\h$ the zero and identity operators respectively. Let $A\in \lh$. The \emph{trace} of $A$ is defined to be $\tr(A) \define \sum_{i\in I} \<i|A|i\>$ for some (or, equivalently, any) orthonormal basis $\{|i\> : i\in I\}$ of $\h$. The \emph{adjoint} of $A$, denoted $A^\dag$, is the unique linear operator in $\lh$ such that $\<\psi|A|\phi\> = \<\phi|A^\dag |\psi\>^*$ for all $|\psi\>, |\phi\>\in \h$. Here for a complex number $z$, $z^*$ denotes its conjugate. Operator $A$ is said to be \emph{normal} if $A^\dag  A = A A^\dag$, \emph{hermitian} if $A^\dag = A$, \emph{unitary} if $A^\dag A = I_\h$, and \emph{positive} if for all $|\psi\>\in \h$, $\<\psi|A|\psi\>\geq 0$. Obviously,  hermitian operators are normal, and both unitary operators and positive ones are hermitian. Any normal operator $A$ can be written into a \emph{spectral decomposition} form $A  = \sum_{i\in I} \lambda_i |i\>\<i|$ where $\{|i\> : i\in I\}$ constitute some orthonormal basis of $\h$. Furthermore, if $A$ is hermitian, then all $\lambda_i$'s are real; if $A$ is unitary, then all $\lambda_i$'s have unit length; if $A$ is positive, then all $\lambda_i$'s are non-negative.  The L\"owner (partial) order $\le_\h$ on the set of hermitian operators on $\h$ is defined by letting $A\le_\h B$ iff $B-A$ is positive. 

Let $\h_1$ and $\h_2$ be two finite dimensional Hilbert spaces, and $\h_1\otimes \h_2$ their tensor product. 
Let $A_i\in \l(\h_i)$. The tensor product of $A_1$ and $A_2$, denoted $A_1\otimes A_2$ is a linear operator in $\l(\h_1\otimes \h_2)$ such that
$(A_1\otimes A_2)|(\psi_1\>\otimes |\psi_2)\> = (A_1|\psi_1\>)\otimes (A_2|\psi_2\>)$ for all $|\psi_i\> \in \h_i$. To simplify notations, we often write $|\psi_1\> |\psi_2\>$ for $|\psi_1\>\otimes |\psi_2\>$.
Given $\h_1$ and $\h_2$, the \emph{partial trace} with respect to $\h_2$, denoted $\tr_{\h_2}$, is a linear mapping from
$\l(\h_1\otimes \h_2)$ to $\l(\h_1)$ such that for any $|\psi_i\>, |\phi_i\> \in \h_i$, $i=1,2$,
$$\tr_{\h_2}(|\psi_1\>\<\phi_1|\otimes |\phi_1\>\<\phi_2|) = 
\<\phi_2|\phi_1\> |\psi_1\>\<\phi_1|.$$
The definition is extended to $\l(\h_1\otimes \h_2)$ by linearity.

A linear operator $\e$ from $\l(\h_1)$ to $\l(\h_2)$ is called a \emph{super-operator}.  It is said to be (1) \emph{positive} if it maps positive operators to positive operators; (2) \emph{completely positive} if all the cylinder extension $\mathcal{I}_\h\otimes \e$ is positive for all finite dimensional Hilbert space $\h$, where $\mathcal{I}_\h$ is the identity super-operator on $\lh$; (3) \emph{trace-preserving} (resp. \emph{trace-nonincreasing}) if 
$\tr(\e(A)) = \tr(A)$ (resp. $\tr(\e(A)) \leq \tr(A)$ for any positive operator $A\in \l(\h_1)$; (4) \emph{unital} (resp. \emph{sub-unital}) if 
$\e(I_{\h_1})= I_{\h_2}$ (resp. $\e(I_{\h_1}) \le_{\h_2} I_{\h_2}$).
From \emph{Kraus representation theorem}~\cite{kraus1983states}, a super-operator $\e$  from $\l(\h_1)$ to $\l(\h_2)$ is completely positive iff there is some set of linear operators, called \emph{Kraus operators}, $\{E_i : i\in I\}$ from $\h_1$ to $\h_2$ such that $\e(A) = \sum_{i\in I} E_i A E_i^\dag$ for all $A\in \l(\h_1)$. 
It is easy to check that the trace and partial trace operations defined above are both completely positive and trace-preserving super-operators. 
Given a completely positive super-operator $\e$ from $\l(\h_1)$ to $\l(\h_2)$ with Kraus operators $\{E_i : i\in I\}$, the adjoint of $\e$, denoted $\e^\dag$, is a completely positive super-operator from $\l(\h_2)$ back to $\l(\h_1)$ with Kraus operators $\{E_i^\dag : i\in I\}$. Then we have $(\e^\dagger)^\dag = \e$, and $\e$ is trace-preserving (resp. trace-nonincreasing) iff $\e^\dag$ is unital (resp. sub-unital). Furthermore, for any $A\in \l(\h_1)$ and $B\in \l(\h_2)$, $\tr(\e(A)\cdot B) = \tr(A\cdot \e^\dag(B))$.

 \subsection{Basic quantum mechanics}

According to von Neumann's formalism of quantum mechanics
\cite{vN55}, any quantum system with finite degrees of freedom is associated with a finite-dimensional Hilbert space $\h$ called its \emph{state space}. When $\dim(\h) = 2$, we call such a system a \emph{qubit}, the analogy of bit in classical computing. A {\it pure state} of the system is described by a normalised vector in $\h$. When the system is in one of an {ensemble} of states $\{|\psi_i\>: i\in I\}$ with respective probabilities $p_i$, we say it is in a \emph{mixed} state, represented by the \emph{density operator} $\sum_{i\in I} p_i|\psi_i\>\<\psi_i|$ on $\h$. Obviously, a density operator is positive and has trace 1. Conversely, by spectral decomposition, any positive operator with unit trace corresponds to some (not necessarily unique) mixed state.

The state space of a composite system (for example, a quantum system
consisting of multiple qubits) is the tensor product of the state spaces
of its components. For a mixed state $\rho$ in $\h_1 \otimes \h_2$,
partial traces of $\rho$ have explicit physical meanings: the
density operators $\tr_{\h_1}(\rho)$ and $\tr_{\h_2}(\rho)$ are exactly
the reduced quantum states of $\rho$ on the second and the first
component systems, respectively. Note that in general, the state of a
composite system cannot be decomposed into tensor product of the
reduced states on its component systems. A well-known example is the
2-qubit state
$|\Psi\>=\frac{1}{\sqrt{2}}(|00\>+|11\>).
$
This kind of state is called {\it entangled state}, and usually is the key to many quantum information processing tasks  such as teleportation
\cite{bennett1993teleporting} and superdense coding \cite{bennett1992communication}.

The \emph{evolution} of a closed quantum system is described by a unitary
operator on its state space: if the states of the system at times
$t_1$ and $t_2$ are $\rho_1$ and $\rho_2$, respectively, then
$\rho_2=U\rho_1U^{\dag}$ for some unitary operator $U$ which
depends only on $t_1$ and $t_2$. In contrast, the general dynamics which can occur in a physical system is
described by a completely positive and trace-preserving super-operator on its state space. 
Note that the unitary transformation $\e_U(\rho)\define U\rho U^\dag$ is
such a super-operator. 

A quantum {\it measurement} $\m$ is described by a
collection $\{M_i : i\in I\}$ of linear operators on $\h$, where $I$ is the set of measurement outcomes. It is required that the
measurement operators satisfy the completeness equation
$\sum_{i\in I}M_i^{\dag}M_i = I_\h$. If the system is in state $\rho$, then the probability
that measurement result $i$ occurs is given by
$p_i=\tr(M_i^{\dag}M_i\rho),$ and the state of the post-measurement system
is $\rho_i = M_i\rho M_i^{\dag}/p_i$ whenever $p_i>0$. 
Note that the super-operator $$\e_\m: \rho\mapsto
\sum_{i\in I} p_i \rho_i = \sum_{i\in I} M_i\rho M_i^\dag$$
which maps the initial state to the final (mixed) one when the measurement outcome is ignored is completely positive and trace-preserving.
A particular case of measurement is {\it projective measurement} which is usually represented by a hermitian operator $M$ in $\lh$ called \emph{observable}.  Let 
\[
M=\sum_{m\in \mathit{spec}(M)}mP_m
\] 
where $\mathit{spec}(M)$ is the set of eigenvalues of $M$, and $P_m$ the projection onto the eigenspace associated with $m$. 
Obviously, the projectors  $\{P_m:m\in
spec(M)\}$ form a quantum measurement. 

In this paper, we are especially concerned with the set 
 \[
 \ph \define \{M\in \lh : \z_\h\le M\le I_\h\}
 \]
of observables whose eigenvalues lie between 0 and 1, where $\le$ is the L\"owner  order on $\lh$. Furthermore, following Selinger's convention~\cite{selinger2004towards}, we regard the set of \emph{partial density operators} 
\[
\dh \define \{\rho\in \lh : \z_\h\le \rho, \tr(\rho)\leq 1\}
\]
as (unnormalised) quantum states. Intuitively, the partial density operator $\rho$ means that the legitimate quantum state $\rho/\tr(\rho)$ is reached with probability $\tr(\rho)$.
As a matter of fact, we note that $\dh\subseteq\ph$.

\section{Classical-quantum states and assertions}

\newsavebox{\tablebox}

{\renewcommand{\arraystretch}{1.5}
\begin{table}[t]
	\begin{lrbox}{\tablebox}
		\centering
		\begin{tabular}{cccc}
			\hline
			Classical & Probabilistic & Quantum & Classical-quantum \\
			\hline \hline
			state & probability (sub)distribution & (partial) density operator &  cq-state   \\ 
			$\cstate \in \cstates$ & $\mu \in\cstates \ra [0,1]$ & $\rho\in \dh$ & $\qstate \in \cstates \ra \dh$  \\
			& countable support  &  & countable support   \\
			\hline
			assertion & (discrete) random variable & observable &  cq-assertion   \\ 
			$\cassert\in \cstates \ra \{0,1\}$ & $f\in \cstates\ra [0, 1]$ & $M\in \ph$ & $\qassert\in \cstates \ra \ph$   \\ 
			& countable image  &  & countable image   \\		
			\hline 
			satisfaction & expectation & expectation &  expectation   \\ 
			$\cstate \models \cassert$ & $\sum_{\cstate\in \supp{\mu}} \mu(\cstate) f(\cstate)$ & $\tr(M\rho)$ & $\sum_{\cstate\in \supp{\qstate}} \tr\left[\qstate(\cstate) \qassert(\cstate)\right]$   \\ 
			\hline
		\end{tabular}
	\end{lrbox}
	\resizebox{0.9\textwidth}{!}{\usebox{\tablebox}}\\
	\vspace{2mm}
	\caption{Comparison of the basic notions in different language paradigms.
	}
	\label{tbl:comparison}
\end{table}
}

In this section, the notions of program states and assertions are introduced for our quantum language where classical variables are involved. To motivate the definition, we first review the corresponding ones in classical (non-probabilistic), probabilistic, and purely quantum programs. A brief summary of the comparison, which extends the one presented in~\cite{d2006quantum}, is depicted in Table~\ref{tbl:comparison}.

Let $\cstates$ be a non-empty set which serves as the state space of classical programs. An assertion $\cassert$ for classical states is (semantically) a mapping from $\cstates$ to $\{0,1\}$ such that a state $\cstate$ satisfies $\cassert$, written $\cstate\models\cassert$, iff $\cassert(\cstate) = 1$. In contrast, a state for probabilistic programs is a probability sub-distribution $\mu$ on $\cstates$ which has countable support\footnote{For simplicity, we only consider here probabilistic programs in which all random variables are taken discrete. The probabilities are not required to sum up to 1 in a probability sub-distribution, for the sake of describing non-termination.}; that is, $\mu(\cstate) >0$ for at most countably infinite many $\cstate\in \cstates$. Accordingly, an assertion for probabilistic states is a discrete random variable $f$ on $\cstates$ with countable image; that is, $f$ takes at most countably infinite many values. Finally, the `degree' of a state satisfying an assertion corresponds naturally to the expected value of a random variable with respect to a probability distribution. In particular, when the assertion is a traditional one, meaning that its image set is $\{0,1\}$, this expectation reduces to the probability of satisfaction.

To motivate the corresponding notions proposed in~\cite{d2006quantum} for purely quantum programs where classical variables are excluded, note that for any partial density operator $\rho$ in $\dh$ and any orthonormal basis $\{|i\> : i\in I\}$ of $\h$, the function $\mu$ with $\mu(i) =\<i|\rho |i\>$ defines a probability sub-distribution over $I$. Thus the set $\dh$ can naturally be taken as the state space for purely quantum programs. Similarly, for any observable $M\in \ph$, $\<i|M|i\>\in [0,1]$ for all $i\in I$. Thus $\ph$ can be regarded as the quantum extension of probabilistic assertions. Finally, the degree of a state $\rho$ satisfying an assertion $M$ is the expected value $\sum_i \mu(i) \<i| M |i\>$, which, when $|i\>$'s are eigenstates of $M$ or $\rho$, is exactly $\tr(M\rho)$. Most remarkably, as $\tr(M\rho)$ is the expected value of outcomes when the projective measurement represented by $M$ is applied on state $\rho$, it can be physically estimated (instead of mathematically calculated) when multiple copies of $\rho$ are available. This physical implementability is especially important in black box testing of quantum programs, where programs can be executed multiple times, but the implementation detail is not available.

For programs where both quantum and classical variables are involved, we have to find a way to combine the notions for probabilistic programs and purely quantum ones.
The following three subsections are devoted to this goal.

\subsection{Classical-quantum states}\label{sec:cqstates}

We assume two basic types for classical variables: $\tybool$ with the corresponding domain $D_{\tybool} \define \{\true, \false\}$ and $\tyint$ with $D_{\tyint} \define\Z$. For each integer $d\geq 1$, we assume a basic quantum type $\tyqudit$ with domain $\h_{\tyqudit}$, which is a $d$-dimensional Hilbert space with an orthonormal basis $\{|0\>, \ldots, |d-1\>\}$. In particular, we denote the quantum type for $d=2$ as $\tyqubit$. Let $\Var$, ranged over by $x,y,\cdots$, and $\qVar$, ranged over by $q, r, \cdots$, be countably infinite sets of classical and quantum variables, respectively. 
Let $\Sigma \define \Var \rightarrow D$ be the (uncountably infinite) set of classical states, where $D\define D_{\tybool}\cup D_{\tyint}$.  We further require that states in $\cstates$ respect the types of classical variables; that is, $\cstate(x) \in D_{\mathit{type}(x)}$ for all $\cstate\in \cstates$ and $x\in \Var$, where $\mathit{type}(x)$ denotes the type of $x$. For any
 finite subset $V$ of $\qVar$, let
 \[\h_V \define \bigotimes_{q\in V} \h_{q},
 \]
 where $\h_{q} \define \h_{\type(q)}$ is the Hilbert space associated with $q$. For simplicity, we let $\h_{\emptyset} \define \C$.
 As we use subscripts to distinguish Hilbert spaces with different (sets of) quantum variables, their order in the tensor product is not essential. In this paper, when we refer to a subset of $\qVar$, it is always assumed to be finite.
 
\begin{definition}\label{def:cqstate}
	Given $V\subseteq \qVar$, a classical-quantum state (cq-state for short) $\qstate$ over $V$ is a function in $\cstates\rightarrow\dhv$ such that  
	\begin{enumerate}
		\item the support of $\qstate$, denoted $\supp{\qstate}$, is countable. That is, $\qstate(\cstate) \neq \z_{\h_V}$ for at most countably infinite many $\cstate\in \cstates$;
		\item $	\tr(\qstate) \define \sum_{\cstate\in \supp{\qstate}}\tr[\qstate(\cstate)] \leq 1$.
	\end{enumerate}
\end{definition}
One may note the similarity of the above definition with probability sub-distributions. Actually, a probability sub-distribution is obtained by assuming that $V = \emptyset$, as in this case $\dhv = [0,1]$. Recall also that in~\cite{selinger2004towards}, the state for a quantum program with $n$ bits $b_1, \ldots, b_n$ and $m$ qubits $q_1,\ldots, q_m$ is given by a $2^n$-tuple $(\rho_0, \ldots, \rho_{2^n-1})$ of partial density matrices, each with dimension $2^m\times 2^m$. Intuitively, each $\rho_i$ denotes the corresponding state of the qubits when the state of the classical bits $b_1\ldots b_n$ constitute the binary representation of $i$. Such a tuple can be described by a cq-state $\qstate$ over $\{q_1,\ldots, q_m\}$ such that $\qstate(\cstate_i) = \rho_i$, $0\leq i<2^n$, where $\Var\define \{b_1, \ldots, b_n\}$, and $\cstate_i(b_k) = i_k$ with $\sum_{k=1}^{n}i_k  2^{n-k}=i$.

Sometimes it is convenient to denote a cq-state $\qstate$ by the explicit form $\bigoplus_{i\in I}\<\cstate_i, \rho_i\>$ where
$\supp{\qstate}= \{\cstate_i : i\in I\}$ 
and $\qstate(\cstate_i)=\rho_i$ for each $i\in I$. 
When $\qstate$ is a simple function such that $\supp \qstate=\{\cstate\}$ for
some $\cstate$ and $\qstate(\cstate)=\rho$, we denote $\qstate$ simply by $\<\cstate, \rho\>$. 
Let  $\{\qstate_i : i\in I\}$ be a countable set of cq-states over $V$ such that  for any $\cstate$,
	$\sum_{i\in I} \qstate_i(\cstate) = \rho_\cstate$ for some $ \rho_\cstate\in \dhv$ and $\sum_{i\in I} \tr(\qstate_i) \leq 1$. Then the summation of them, denoted $\sum_{i\in I} \qstate_i$, is a cq-state $\qstate$ over $V$
	such that for any $\cstate\in \cstates$, $\qstate(\cstate)= \rho_\cstate$. Obviously, $\supp{\qstate} = \bigcup_{i\in I}\supp{\qstate_i}$. 
It is worth noting the difference between $\sum_{i\in I} \<\cstate_i, \rho_i\>$, the summation of some (simple) cq-states, and  $\bigoplus_{i\in I} \<\cstate_i, \rho_i\>$, the explicit form of a single one: in the latter $\cstate_i$'s must be distinct while in the former they may not.

Let $\e$ be a completely positive and trace-nonincreasing super-operator from $\l(\h_V)$ to $\l(\h_W)$. We extend it to $\qstatesh{V}$ in a point-wise way: $\e(\qstate)(\cstate) = \e(\qstate(\cstate))$ for all $\cstate$. Note that $\e(\z_{\h_V}) = \z_{\h_W}$ and $\tr(\e(\rho)) \leq \tr(\rho)$ for all $\rho\in \dhv$. Thus 
$\e(\qstate)$ is a valid cq-state provided that $\qstate$ is. In particular, for any $V\subseteq qv(\qstate)$, the partial trace
$\tr_{\h_{V}}(\qstate)$ is a cq-state which maps any $\cstate\in \cstates$ to $\tr_{\h_{V}}(\qstate(\cstate))$. 
Furthermore, for any $\rho \in\d(\h_W)$ with $W\cap \qv(\qstate) = \emptyset$ and $\tr(\rho) \leq 1$,
$\qstate\otimes \rho$ is a cq-state in $\qstatesh{\qv(\qstate)\cup W}$ which maps $\cstate$ to $\qstate(\cstate)\otimes \rho$. In the special case that $W=\emptyset$, $\rho$ becomes a real number  in $[0,1]$, and we write $\rho\qstate$  for $\qstate\otimes \rho$.

\begin{example}\label{ex:tel}
To better understand the notion of cq-states, let us consider the output of quantum Teleportation algorithm~\cite{bennett1993teleporting}, where Alice would like to teleport an arbitrary state $\rho_q$ to Bob, using a pre-shared Bell state $\frac{1}{\sqrt{2}} (|00\> + |11)_{q_1, q_2}$ between them. Here $q$, $q_1$, and $q_2$ are all $\tyqubit$-type variables, and we use subscripts to indicate the quantum variables on which the states and operators are acting. Suppose $\cstate_i$ and $\rho^i$, $0\leq i\leq 3$, are the classical and final quantum states (of $q_2$), respectively, when the measurement outcome of Alice is $i$. Then the cq-state output by the algorithm can be written as
\[
\qstate\define \bigoplus_{0\leq i\leq 3} \left\<\cstate_i, \frac{1}{4} |i\>_{q, q_1}\<i| \otimes \rho^i_{q_2}\right\>.
\]
\end{example}

Note that a more intuitive way to describe the cq-state $\qstate$ in Example~\ref{ex:tel} is to use a probability distribution of  classical-quantum state pairs: $\{\frac{1}{4}(\cstate_i, |i\>_{q, q_1}\<i| \otimes \rho^i_{q_2}): 0\leq i\leq 3\}$.
We will explain why we decide not to do so in more detail at the end of this subsection after more notations are introduced.

Let $\qstatesh{V}$ be the set of all cq-states over $V$, and $\s$ the set of all cq-states; that is, 
$$\s \define \bigcup_{V\subseteq \qVar} \qstatesh{V}.$$ 
When $\qstate\in \qstatesh{V}$, denote by $qv(\qstate) \define V$ the set of quantum variables in $\qstate$.
We extend the L\"{o}wner order $\le_V$ for $\l(\h_V)$ point-wisely to $\s$ by letting $\qstate\le \qstate'$ iff $qv(\qstate) = qv(\qstate')$ and for all $\cstate\in \Sigma$, $\qstate(\cstate) \le_{qv(\qstate)} \qstate'(\cstate)$. 
Obviously, when both $\qstate$ and $\qstate'$ are probability sub-distributions, i.e., $qv(\qstate) = qv(\qstate')=\emptyset$, then $\qstate\le \qstate'$ iff they are related with the partial order defined in~\cite{morgan1996probabilistic} for probability sub-distributions.

The following lemma shows that $\s$ is an $\omega$-complete partial order (CPO) under $\le$. 

\begin{lemma}\label{lem:cpo}
	For any $V\subseteq \qVar$, $\qstatesh{V}$ is a pointed $\omega$-CPO under $\le$, with the least element being the constant $\z_{\h_V}$ function, denoted $\emptydis_V$. Furthermore, $\s$ as a whole is an $\omega$-CPO under $\le$.
\end{lemma}
\begin{proof}
	The result follows directly from the fact that for any $V\subseteq \qVar$, $\dhv$ is an $\omega$-CPO under the L\"{o}wner order $\le_V$, with $\z_{\h_V}$ being its least element~\cite{selinger2004towards}.
\end{proof}

When $\qstate \le \qstate'$, there exists a unique $\qstate''\in \qstatesh{qv(\qstate)}$, denoted $\qstate' - \qstate$, such that $\qstate'' + \qstate = \qstate'$. 
 For any real numbers $\lambda_i$, $i\in I$, if both $\qstate_+ \define \sum_{\lambda_i >0}\lambda_i  \qstate_i$ and $\qstate_- \define \sum_{\lambda_i <0}(-\lambda_i)  \qstate_i$ are well-defined and $\qstate_- \le  \qstate_+$ ,
then the linear-sum $\sum_{i\in I}\lambda_i  \qstate_i$ is defined to be $\qstate_+ - \qstate_-$. In the rest of this paper, whenever we write $\sum_{i\in I}\lambda_i  \qstate_i$ we always assume that it is well-defined.

		\begin{example}\label{exa:cqstate}
		Let $\cstate_1\neq \cstate_2 \in \cstates$, $q\in \qVar$ with $\type(q) = \tyqubit$, 
		$
		\qstate \define \<\cstate_1, 0.5 |0\>_q\<0|\>  \oplus \<\cstate_2,  0.25I_q\>.
		$
		 and $\qstate' \define \<\cstate_2, |+\>_q\<+|\>$. Then we have the linear-sum
	\[
	\qstate - 0.25 \qstate' =  \<\cstate_1, 0.5|0\>_q\<0|\> \oplus \<\cstate_2,  0.25|-\>_q\<-|\>.
	\]
\end{example}

To conclude this subsection, we would like to say a few words about the design decision we make in Definition~\ref{def:cqstate}. Recall that a partial density operator $\rho$ encodes both the (normalised) quantum state $\rho/\tr(\rho)$ and the probability $\tr(\rho)$ of reaching it. Thus the meaning of a cq-state $\qstate = \bigoplus_{i\in I} \<\cstate_i, \rho_i\>$ is that with probability $\tr(\rho_i)$, the classical and quantum systems are in states $\cstate_i$ and $\rho_i/\tr(\rho_i)$, respectively. This also explains why we have the requirement in Definition~\ref{def:cqstate}(2): the probabilities of all possible state pairs sum up to at most 1. One may ask why we do not directly define cq-states as sub-distributions over such classical-quantum state pairs, just as in~\cite{chadha2006reasoning} (see also the comment below Example~\ref{ex:tel})? To see the reason, note that $\dhv$ is a convex set, and the quantum state $\sum_i \lambda_i \rho_i$ is indistinguishable from the ensemble that lies in $\rho_i$ with probability $\lambda_i\geq 0$, $\sum_i \lambda_i =1$. Thus we would have to introduce some auxiliary rules to equate the probability distribution $\sum_{i} \lambda_i \<\cstate, \rho_i\>$ with the single state $\<\cstate, \sum_{i} \lambda_i \rho_i\>$, if cq-states had been defined as sub-distributions on classical-quantum state pairs. In contrast, in our framework these two cq-states are \emph{equal by definition}, from the linear-sum form introduced above. Finally, note that this difficulty does not appear in probabilistic programs, as the classical state space $\cstates$ is discrete, and there does not exist any algebraic structure in it.

\subsection{Classical-quantum assertions} 
 
Recall that  assertions for classical program states are usually represented as first order logic formulas over $\Var$.
For any classical assertion $\cassert$, denote by $\sem{\cassert} \define \{\cstate\in \Sigma : \cstate\models \cassert\}$ the set of classical states that satisfy $\cassert$. Two assertions $\cassert$ and $\cassert'$ are equivalent, written $\cassert \equiv \cassert'$, iff $\sem{\cassert} = \sem{\cassert'}$.

\begin{definition}
	Given $V\subseteq \qVar$, a classical-quantum assertion (cq-assertion for short) $\qassert$ over $V$ is a function in $\cstates\rightarrow\phv$ such that  
\begin{enumerate}
	\item the image set $\qassert(\Sigma)$ of $\qassert$ is countable;
	\item for each $M\in \qassert(\Sigma)$, the preimage $\qassert^{-1}(M)$ is definable by a classical assertion $\cassert$ in the sense that $\sem{\cassert} =  \qassert^{-1}(M)$.
\end{enumerate}
\end{definition}
Obviously, the above definition is a natural extension of discrete random variables when $\phv$, the set of operators between $\z_{\h_V}$ and $I_{\h_V}$ with respect to the L\"{o}wner order, is regarded as the quantum generalisation of $[0,1]$. The second clause is introduced to guarantee a compact representation of cq-assertions.

For convenience, we do not distinguish $\cassert$ and $\sem{\cassert}$ when denoting a cq-assertion. Consequently, we write $\bigoplus_{i\in I}\<\cassert_i, M_i\>$ instead of $\bigoplus_{i\in I}\<\sem{\cassert_i}, M_i\>$ for a cq-assertion  $\qassert$ whenever
$\qassert(\cstates) = \{M_i : i\in I\}$ 
and $\qassert^{-1}(M_i)=\sem{\cassert_i}$ for each $i\in I$. Note that this representation is not unique: the representative assertion $\cassert_i$ can be replaced by $\cassert'_i$ whenever $\cassert_i\equiv \cassert'_i$.
Furthermore, the summand with zero operator $\z_{\h_V}$ is always omitted.
 In particular, when $\qassert(\cstates) = \{\z_\h, M\}$ or $\{M\}$ for some $M\neq \z_{\h_V}$, we simply denote $\qassert$ by $\<\cassert, M\>$ for some $\cassert$ with $\qassert^{-1}(M)=\sem{\cassert}$.

Note that any observable $M$ in $\ph$ corresponds to some \emph{quantitative property} of quantum states.
	Thus intuitively, a cq-assertion $\bigoplus_{i\in I}\<\cassert_i, M_i\>$ specifies that whenever the classical state satisfies $\cassert_i$, the property $M_i$ is checked on the corresponding quantum state. The average value of the satisfiability will be defined in the next subsection. 

	\begin{example}\label{ex:telassert}
	Back to the Teleportation algorithm in Example~\ref{ex:tel}. The
	cq-assertion
	\[
	 \qassert_1 \define \bigoplus_{0\leq i\leq 3}\<x=i, |i\>_{q,q_1}\<i|\>
	\]
	where $x$ is the classical variable used by Alice to store (and send to Bob) the measurement outcome,
	claims that the states of $q$ and $q_1$ are both in the computational basis, and they together correspond to the measurement outcome of Alice. To be specific, it states that whenever $x=i$, the corresponding quantum state of $q$ and $q_1$ should be $|x_0\>_q|x_1\>_{q_1}$ where $x_0x_1$ is the binary representation of $x$.
	We will make it more rigorous in Example~\ref{ex:telexp}.
	
	Suppose the teleported state $\rho = |\psi\>\<\psi|$ is a pure one. Then the cq-assertion 
	\[
	\qassert_2 \define \<\true, |\psi\>_{q_2}\<\psi|\>,
	\]
	when applied on the output cq-state,
	actually computes the (average) precision of the Teleportation algorithm when $|\psi\>$ is taken as the input. Again, we refer to Example~\ref{ex:telexp} for more details.
\end{example}

Let $\qassertsh{V}$ be the set of all cq-assertions over $V$, and $\a$ the set of all cq-assertions; that is, 
$$\a \define \bigcup_{V\subseteq \qVar} \qassertsh{V}.$$ 
When $\qassert\in \qassertsh{V}$, denote by $qv(\qassert) \define V$ the set of quantum variables in $\qassert$. Again, we extend the L\"{o}wner order $\le_V$ for $\l(\h_V)$ point-wisely to $\a$ by letting $\qassert\le \qassert'$ iff $qv(\qassert) = qv(\qassert')$ and for all $\cstate\in \Sigma$, $\qassert(\cstate) \le_{qv(\qassert)} \qassert'(\cstate)$. 
It is easy to see that 
 $\qasserts$ is also a pointed $\omega$-CPO under $\le$, with the least element being $\emptydis_{V}$. Furthermore, it has the largest element $\top_{V} \define \<\true, I_{\h_V}\>$. If both $\qassert$ and $\qassert'$ are probabilistic assertions, i.e., $qv(\qassert) = qv(\qassert')=\emptyset$, then $\qassert\le \qassert'$ iff they are related with the partial order defined in~\cite{morgan1996probabilistic} for probabilistic assertions.

When $\qassert \le \qassert'$,  we denote by $\qassert' - \qassert$ the unique $\qassert''\in \qassertsh{qv(\qassert)}$ such that $\qassert'' + \qassert = \qassert'$. With these notions, summation and linear-sum of cq-assertions can be defined similarly as for cq-states.

Given a classical assertion $\cassert$, we denote by $\cassert \bowtie \sum_{i} \<\cassert_i, M_i\>$ the cq-assertion $\sum_{i} \<\cassert \bowtie \cassert_i, M_i\>$  (if it is valid) where $\bowtie$ can be any logic connective such as $\wedge$, $\vee$, $\Rightarrow$, $\Leftrightarrow$, etc. As $\sem{\cassert \bowtie \cassert_1} =  \sem{\cassert \bowtie \cassert_2}$ provided that  $\sem{\cassert_1} = \sem{\cassert_2}$, these notations are well-defined. Let $\f$ be a completely positive and sub-unital linear map from $\phv$ to $\mathcal{P}(\h_W)$. We extend it to $\qassertsh{V}$ in a point-wise way. Note that $\f(\z_{\h_V}) = \z_{\h_{W}}$ and $\f(M) \le \f(I_{\h_V}) \le I_{\h_W}$ for all $M\in \phv$. Thus $\f(\qassert)$ is a valid cq-assertion provided that $\qassert$ is. In particular, when $qv(\qassert) \cap W = \emptyset$, $\qassert\otimes I_{\h_W}$ is a cq-assertion which maps any $\cstate\in \cstates$ to $\qassert(\cstate)\otimes I_{\h_W}$. Note that $1$ is the identity operator on $\h_{\emptyset}$. Sometimes we also abuse the notation a bit to write $\cassert$ for $\<\cassert, 1\>$ where $\cassert$ is a classical assertion, and $\lambda$ for $\<\true, \lambda\>$ where $\lambda\in [0,1]$.

The (lifted) L\"{o}wner order provides a natural way to compare cq-assertions over the same set of quantum variables. However, in later discussion of this paper, we sometimes need to compare cq-assertions acting on different quantum variables. To deal with this situation, we introduce a pre-order $\lesssim$ on the whole set $\a$ of cq-assertions.
To be specific, let $V_1, V_2$ be two subsets of $\qVar$, and $\qassert_i\in \a_{V_i}$, $i=1,2$. We say $\qassert_1\lesssim \qassert_2$ whenever $\qassert_1\otimes I_{\h_{V_2\backslash V_1}} \le I_{\h_{V_1\backslash V_2}} \otimes \qassert_2$. Obviously, when restricted on some given set of quantum variables, $\lesssim$ coincides with $\le$. 
Let $\eqsim$ be the kernel of $\lesssim$. Then $\qassert_1 \eqsim \qassert_2$ iff there exists $\qassert$ such that $\qassert_1 = \qassert\otimes I_V$ and $\qassert_2 = \qassert \otimes I_W$ for some $V$ and $W$.

\subsection{Expectation of satisfaction}

With the above notions, we are now ready to define the \emph{expectation} (or \emph{degree}) of a cq-state satisfying a cq-assertion. 

\begin{definition}\label{def:satisfaction}
	Given a cq-state $\qstate$ and a cq-assertion $\qassert$ with $\qv(\qstate) \supseteq \qv(\qassert)$, the expectation of $\qstate$ satisfying $\qassert$ is defined to be 
	\[
	\Exp(\qstate \models \qassert) \define 
	 \sum_{\cstate\in \supp{\qstate}} \tr\left[\left(\qassert(\cstate)\otimes I_{\h_{V}} \right) \cdot \qstate(\cstate)\right]
	 = 
	 \sum_{\cstate\in \supp{\qstate}} \tr\left[\qassert(\cstate) \cdot \tr_{\h_{V}}(\qstate(\cstate))\right]
	\] 
	where $V = \qv(\qstate) \backslash \qv(\qassert)$ and the dot $\cdot$ denotes matrix multiplication.
\end{definition}

Again, when both $\qstate$ and $\qassert$ are probabilistic, i.e., $qv(\qstate) = qv(\qassert)=\emptyset$, then the expectation defined above is exactly the expected value of $\qassert$ over $\qstate$ defined in~\cite{morgan1996probabilistic} for probabilistic programs.

\begin{example}\label{ex:telexp}
	Consider again the Teleportation algorithm. Let
	$
	\qstate$ be defined as in Example~\ref{ex:tel}, and $\qassert_1$ and $\qassert_2$ in Example~\ref{ex:telassert}. Note that $\cstate_i \models (x = j)$ iff $i=j$. Thus 
	\[
	\Exp(\qstate \models \qassert_1) = \sum_{i=0}^3 \frac{1}{4} \tr(|i\>_{q,q_1}\<i|\otimes I_{q_2} \cdot |i\>_{q,q_1}\<i| \otimes \rho^i_{q_2}) = 1,
	\]
	meaning that with probability 1, $x$ equals the value represented by the states of $q$ and $q_1$.
	
	For $\qassert_2$, we compute 
		\[
	\Exp(\qstate \models \qassert_2) = \sum_{i=0}^3 \frac{1}{4} \tr(I_{q,q_1}\otimes |\psi\>_{q_2}\<\psi| \cdot |i\>_{q,q_1}\<i| \otimes \rho^i_{q_2}) =  \frac{1}{4} \sum_{i=0}^3 \<\psi|\rho^i|\psi\>,
	\]
	which denotes the average fidelity between the output states $\rho^i$ and the ideal one $|\psi\>$.
\end{example}

We collect some properties of the $\Exp$ function in the following lemmas.

\begin{lemma}\label{lem:bpdeg}
	For any cq-state $\qstate\in \qstatesh{V}$, cq-assertion $\qassert\in \qassertsh{W}$ with $W\subseteq V$, and classical assertion $\cassert$,
	\begin{enumerate}
		\item $\Exp(\qstate \models \qassert)\in [0,1]$;
		\item $\Exp(\emptydis_V \models \qassert) = \Exp(\qstate \models \emptydis_{W})=0$, $\Exp(\qstate \models \top_{W}) = \tr(\qstate)$;
		\item $\Exp(\qstate \models \qassert) = \sum_i \lambda_i \Exp(\qstate \models \qassert_i)$ if $\qassert = \sum_i \lambda_i \qassert_i$;
		\item $\Exp(\qstate \models \qassert) = \sum_i \lambda_i \Exp(\qstate_i \models \qassert)$ if $\qstate = \sum_i \lambda_i \qstate_i$;
		\item\label{cl:lem3.9} $\Exp(\qstate |_\cassert \models \qassert) = \Exp(\qstate \models \cassert\wedge\qassert)$
		where $\qstate |_{\cassert}$ is the cq-state by restricting $\qstate$ on the set of classical states $\cstate$ with $\cstate\models \cassert$;
		{\item\label{cl:super} $\Exp(\qstate \models \f(\qassertp)) = \Exp(\f^\dag(\qstate) \models \qassertp)$ for any $\qassertp\in \qassertsh{W'}$ and any completely positive and sub-unital super-operator $\f$ from $\h_{W'}$ to $\h_W$.}
	\end{enumerate} 
\end{lemma}
\begin{proof} We only prove Clause~\ref{cl:lem3.9}; the others are easy from  definitions. For simplicity, we assume $V=W$. Then
	\begin{eqnarray*}
		\Exp(\qstate |_\cassert \models \qassert)  &=& \sum_{\cstate\in \supp{\qstate|_\cassert}} \tr\left[ \qassert(\cstate)\cdot \qstate|_\cassert(\cstate)\right]\\
		&=&\sum_{\cstate\in \supp{\qstate}, \cstate\models \cassert } \tr\left[ \qassert(\cstate)\cdot\qstate(\cstate)\right]\\
		&=&\sum_{\cstate\in \supp{\qstate}} \tr\left[ (\cassert\wedge\qassert)(\cstate)\cdot\qstate(\cstate)\right]\\
		&=&\Exp(\qstate \models \cassert\wedge\qassert)
	\end{eqnarray*}
	where the third inequality comes from the fact that for any  $\cstate\in \cstates$, $(\cassert\wedge\qassert)(\cstate) = \qassert(\cstate)$ if $\cstate\models \cassert$, and $\z_{\h_W}$ otherwise.
\end{proof}

\begin{lemma}\label{lem:qassetorder}
		\begin{enumerate}
	\item For any cq-states $\qstate$ and $\qstate'$ in $\qstatesh{V}$,
	\begin{itemize}
		\item if $\qstate \le \qstate'$, then $\Exp(\qstate \models \qassert)\leq \Exp(\qstate' \models \qassert)$ for all $\qassert\in \qassertsh{W}$ with $W\subseteq V$;
		\item conversely, if $\Exp(\qstate \models \qassert)\leq \Exp(\qstate' \models \qassert)$ for all $\qassert\in \qassertsh{V}$, then $\qstate \le \qstate'$.
	\end{itemize} 
	\item For any cq-assertions $\qassert$ and $\qassert'$ with $W=qv(\qassert)\cup qv(\qassert')$,
\begin{itemize}
	\item if $\qassert \lesssim \qassert'$, then $\Exp(\qstate \models \qassert)\leq \Exp(\qstate \models \qassert')$ for all $\qstate\in \qstatesh{V}$ with $W\subseteq V$;
	\item conversely, if $\Exp(\qstate \models \qassert)\leq \Exp(\qstate \models \qassert')$ for all $\qstate\in \qstatesh{W}$, then $\qassert \lesssim \qassert'$.
\end{itemize} 
	\end{enumerate} 
\end{lemma}
\begin{proof}
	We take the converse part of Clause (1) as an example. Suppose $\qstate \not \le \qstate'$. Then there exists a $\cstate\in \cstates$ and $|\psi\> \in \d(\h_V)$ such that $\<\psi|\qstate(\cstate)|\psi\>  > \<\psi|\qstate'(\cstate)|\psi\>$. If we can find a classical assertion $\cassert$ which distinguishes $\cstate$ from other states in $\supp{\qstate'}$. Then obviously the cq-assertion $\<\cassert, |\psi\>_V\<\psi|\>$ serves as a counter-example for the assumption.
	
	Note that the formula $\cassert_* \define\bigwedge_{x\in \Var} (x = \cstate(x))$ uniquely determines $\cstate$. However, it is not a valid classical assertion, as the set $\Var$ is infinite. To convert it to a finite conjunction, let $\epsilon \define \<\psi|\qstate(\cstate)|\psi\> - \<\psi|\qstate'(\cstate)|\psi\>$. For this $\epsilon$, there exists a finite subset $A$ of $\supp{\qstate'}$ such that $\tr(\qstate'|_A) > \tr(\qstate')  - \epsilon$. For any $\cstate'\in A$ with $\cstate'\neq \cstate$, there exists  $x_{\cstate'}\in\Var$ such that  $\cstate'(x_{\cstate'}) \neq \cstate(x_{\cstate'})$. Now let $X \define \{x_{\cstate'} : \cstate'\in A, \cstate'\neq \cstate\}$. Then the classical assertion $\cassert \define\bigwedge_{x\in X} (x = \cstate(x))$ distinguishes $\cstate$ from other states in $A$. Finally, let $\qassert \define \<\cassert, |\psi\>_V\<\psi|\>$.  
	Then $\Exp(\qstate'|_A \models \qassert) = \<\psi|\qstate'(\cstate)|\psi\>$, and $\Exp(\qstate' - \qstate'|_A \models \qassert) \leq \tr(\qstate' - \qstate'|_A) < \epsilon$. Thus
	\[
	\Exp(\qstate \models \qassert) \geq \<\psi|\qstate(\cstate)|\psi\> = \<\psi|\qstate'(\cstate)|\psi\> + \epsilon > \Exp(\qstate' \models \qassert),\] contradicting the assumption.
\end{proof}

\begin{lemma}\label{lem:qasset}
	For any cq-states $\qstate, \qstate_n\in \qstatesh{V}$ and cq-assertions $\qassert, \qassert_{n} \in \qassertsh{W}$ with $W\subseteq V$, $n=1, 2, \cdots$,
	\begin{enumerate}
		\item $\Exp(\bigvee_{n\geq 0}\qstate_n \models \qassert) = \sup_{n\geq 0}\Exp(\qstate_n\models \qassert)$ for increasing sequence $\{\qstate_n\}_n$;		
		\item $\Exp(\bigwedge_{n\geq 0}\qstate_n \models \qassert) = \inf_{n\geq 0}\Exp(\qstate_n\models \qassert)$ for decreasing sequence $\{\qstate_n\}_n$;
		\item $\Exp(\qstate \models \bigvee_{n\geq 0}\qassert_n) = \sup_{n\geq 0}\Exp(\qstate\models \qassert_n)$ for increasing sequence $\{\qassert_n\}_n$;
		\item $\Exp(\qstate \models \bigwedge_{n\geq 0}\qassert_n) = \inf_{n\geq 0}\Exp(\qstate\models \qassert_n)$ for decreasing sequence $\{\qassert_n\}_n$.	
	\end{enumerate} 
\end{lemma}
\begin{proof}
	We prove (1) as an example; the others are similar. Let $\qstate^* \define \bigvee_{n\geq 0}\qstate_n$. First, from the fact that $\qstate_n \le \qstate^*$ for all $n$, $\Exp(\qstate^* \models \qassert) \geq \sup_{n\geq 0}\Exp(\qstate_n\models \qassert)$ by Lemma~\ref{lem:qassetorder}(1). Furthermore, for any $n$,
	\begin{align*}
		\Exp(\qstate^* \models \qassert) -\Exp(\qstate_n\models \qassert) &= \Exp(\qstate^* - \qstate_n\models \qassert)\\
		&\leq \Exp(\qstate^* - \qstate_n\models \top_W)\\
		& = \tr(\qstate^* - \qstate_n) = \tr(\qstate^*) - \tr(\qstate_n),
	\end{align*}
	where the first and second equalities are from Lemma~\ref{lem:bpdeg}, and the first inequality from Lemma~\ref{lem:qassetorder}(2).
	Thus $\Exp(\qstate^* \models \qassert) = \sup_{n\geq 0}\Exp(\qstate_n\models \qassert)$ from the fact that $\tr(\qstate^*) = \sup_n\tr(\qstate_n)$.
\end{proof}

\subsection{Substitution and state update}
Let $e$ and $e' $ be classical expressions\footnote{We assume standard classical expressions (constructed inductively from $\Var$ and a fixed set of function symbols) in this paper; the precise definition of them is omitted.}, and $x$ a classical variable with the same type of $e'$. Denote by $e[\subs{x}{e'}]$ the expression obtained by substituting $x$ in $e$ with $e'$. Such substitution can be extended to cq-assertions as follows.
Given $\qassert \define \bigoplus_{i\in I}\<\cassert_i, M_i\>$, we define
\[
\qassert[\subs{x}{e}] \define \bigoplus_{i\in I}\<\cassert_i[\subs{x}{e}], M_i\>.
\]
The well-definedness comes from the following two observations: (1) whenever $\sem{\cassert_i} \cap \sem{\cassert_j} = \emptyset$, it holds
$\sem{\cassert_i[\subs{x}{e}]} \cap \sem{\cassert_j[\subs{x}{e}]} = \emptyset$; (2) whenever $\cassert\equiv \cassert'$, it holds
 $\cassert[\subs{x}{e}]\equiv \cassert'[\subs{x}{e}]$.
Thus the substitution is independent of the choice of the representative classical assertions $\cassert_i$ in $\qassert$.

For classical state $\cstate$ and $d\in D_{\mathit{type}(x)}$, denote by $\cstate[\subs{x}{d}]$ the updated state which maps $x$ to $d$, and other classical variables $y$ to $\cstate(y)$. Similarly, this updating can be extended to cq-states by defining
\[
\qstate[\subs{x}{e}] \define \sum_{i\in I} \<\cstate_i[\subs{x}{\cstate_i(e)}], \rho_i\>
\]
whenever 
$\qstate = \bigoplus_{i\in I} \<\cstate_i, \rho_i\>$. 
Since substitution does not change the trace of the whole state, $\qstate[\subs{x}{e}]$ is still a valid cq-state. 
Note that unlike cq-assertions,
$\cstate_i[\subs{x}{\cstate_i(e)}]$ and $\cstate_j[\subs{x}{\cstate_j(e)}]$ can be equal even when $\cstate_i\neq \cstate_j$. Thus we have $\sum$ instead of $\bigoplus$ here. 

Note that for any classical state $\cstate$ and assertion $\cassert$, we have the substitution rule: 
$\cstate\models \cassert[\subs{x}{e}]$ iff
$\cstate[\subs{x}{\cstate(e)}]\models \cassert$.
The next lemma shows a similar relation between substitutions for cq-states and cq-assertions.

\begin{lemma}\label{lem:qavalid}
	For any cq-state $\qstate$ and cq-assertion $\qassert$ with $qv(\qstate) \supseteq \qv(\qassert)$, $x\in \Var$, and classical expression $e$ with the same type of $x$,
$$
	\Exp(\qstate \models \qassert[\subs{x}{e}]) = \Exp(\qstate[\subs{x}{e}] \models \qassert).
	$$
\end{lemma}
\begin{proof}
 Let $\qstate = \bigoplus_{i\in I} \<\cstate_i, \rho_i\>$ and $\qassert = \bigoplus_{j\in J}\<\cassert_j, M_j\>$. Then
	\begin{eqnarray*}
		\Exp(\qstate \models \qassert[\subs{x}{e}]) &=& \sum_{i\in I} \sum_{j\in J, \cstate_i\models \cassert_j[\subs{x}{e}]} \tr(M_j \rho_i)\\
		&=&  \sum_{i\in I} \sum_{j\in J, \cstate_i[\subs{x}{\cstate_i(e)}]\models \cassert_j} \tr(M_j \rho_i),
	\end{eqnarray*}
	which is exactly $\Exp(\qstate[\subs{x}{e}] \models \qassert)$. Here we assume that $qv(\qstate) = \qv(\qassert)$; the general case can be proved similarly. 
\end{proof}

\section{A simple classical-quantum language}\label{sec:blang}

This section is devoted to the syntax and various semantics of our core programming language which supports deterministic and probabilistic assignments, quantum measurements, quantum operations, conditionals, and while loops. 

\subsection{Syntax}\label{sec:syntax}

Our classical-quantum language is based on the one proposed in~\cite{selinger2004towards}, extended with $\tyint$ type classical variables and probabilistic assignments, but excluding general recursion and procedure call. The syntax is defined as follows:
\begin{align*} S::= &\ \sskip\ |\ \abort\ |\ x:= e\ |\ x\rassign g\ | \ x:= \measure\ \m[\bar{q}]\ | \ q:=0\ |\
	\bar{q}\apply U\ |\ S_0;S_1\ |\\ &\ \mstm\ |\ \wstm
\end{align*}
where $S,S_0$ and $S_1$ denote classical-quantum programs (cq-programs for short), $x$ a classical variable in $\Var$, $e$ a classical expression with the same type as $x$, $g$ a discrete probability distribution over $D_{\mathit{type}(x)}$, $b$ a $\tybool$-type expression, $q$ is a quantum variable and $\bar{q} \define q_1, \ldots, q_n$ a (ordered) tuple of distinct quantum variables in $\qVar$, $\m$ a measurement and $U$ a unitary operator on
$d_{\bar{q}}$-dimensional Hilbert space where $$d_{\bar{q}}\define \dim(\h_{\bar{q}}) = \prod_{i=1}^{n} \dim(\h_{q_i}).$$
Sometimes we also use $\bar{q}$ to denote the (unordered) set $\{q_1,q_2,\dots,q_n\}$. Let $|\bar{q}|\define n$ be the size of $\bar{q}$.

Let $\prog$ be the set of all cq-programs. For any $S\in \prog$, the quantum variables that appear in $S$ is denoted $\qv(S)$. The set $var(S)$ (resp. $change(S)$) of classical variables that appear in (resp. can be changed by) $S$ are defined in the standard way. Note that the only way to retrieve information from a quantum system is to measure it, a process which may change its state. Thus the notion of read-only quantum variables does not exist in cq-programs. 

In the purely quantum language presented in~\cite{ying2012floyd}, conditional branching is achieved by the program construct
	$\mathbf{measure}\ \m[\bar{q}] : \bar{S}$
	 where $\bar{S}$ is a set of programs which one-to-one correspond to the measurement outcomes of $\m$. Intuitively, the quantum variables in $\bar{q}$ are measured according to $\m$, and different subsequent programs in $\bar{S}$ will be executed depending on the measurement outcomes. The while loop $\while\ \mathcal{N}[\bar{q}]=1\ \ddo\ S\ \pend$ where the outcome set of $\mathcal{N}$ is $\{0,1\}$ is defined similarly. Let $\bar{S} = \{S_i: 1\leq i\leq n\}$ and the outcome set of $\m$ is $\{1, 2, \ldots, n\}$. Then
	these constructs can be expressed in our language in the following equivalent form:
	\begin{align*}
	&x:=\measure\ \m[\bar{q}]; \ \measstm{x=1}{S_1}{(\measstm{x=2}{S_2}{\cdots})}
	\end{align*}
	and
	\begin{align}\label{eq:encode}
	&x:=\measure\ \mathcal{N}[\bar{q}]; \while\ x=1\ \ddo\  S;  x:= \measure\ \mathcal{N}[\bar{q}]; \pend
	\end{align}
	respectively, where $x$ is a fresh classical variable which does not appear in $\bar{S}$ or $S$. An advantage of having classical variables explicitly in the language is that we can avoid introducing infinite-dimensional quantum variables to encode classical data with infinite domains such as \tyint. This will simplify the verification of real-world quantum programs.

To conclude this subsection, 
we introduce some syntactic sugars for our language which make it easy to use in describing quantum algorithms. Let $\bar{q} \define q_1, \ldots, q_n$.
\begin{itemize}
	\item \emph{Initialisation of multiple quantum variables}. Let $\bar{q} := 0$ stand for  $q_1 := 0; \cdots; q_n := 0$.
	\item \emph{Measurement according to the computational basis}. We write $x:= \measure\ \bar{q}$ for  $x:= \measure\ \m_{\mathit{com}}[\bar{q}]$ 
	where $\m_{\mathit{com}} \define \{P_k \define |k\>\<k| : 0\leq k<d_{\bar{q}}\}$ is the projective measurement according to the computational basis of $\h_{\bar{q}}$. We always write $|k\>$ for the product state $|k_1\>\cdots|k_n\>$, where $k = \sum_{i=1}^{n}k_i  d_{q_{i+1}}\ldots d_{q_n}$.
	\item \emph{Application of parametrised unitary operations}. Let $\u \define \{U_i : 1\leq i\leq K\}$, be a finite family of unitary operators on the $d_{\bar{q}}$-dimensional Hilbert space, and $e$ an $\tyint$-typed expression. We write $\bar{q}\apply {\u}(e)$ for the statement which applies $U_i$ on $\bar{q}$ whenever $e$ evaluates to $i$ in the current classical state. Formally, it denotes the following program:
	$$\measstm{e<1 \vee e> K}{\abort}{S_1; S_2; \cdots; S_{K}}$$
	where for each $1\leq i\leq K$,
	$$S_i \define \measstm{e=i}{\bar{q}\apply {U_i}}{\sskip}.$$
	Note that the order of $S_i$'s is actually irrelevant as there is at most one that will be executed.
	
	\item \emph{Application on selected variables in a quantum register}. Let $1\leq k\leq n$, 
	and $e$ and $e_j$'s be $\tyint$-type expressions. The statement $\bar{q}[e_1, \cdots, e_k] \apply \u(e)$, where $\u$ is defined as in the previous clause, applies  $U_i$ on quantum systems $q_{i_1}, \cdots, q_{i_k}$ whenever $e$ evaluates to $i$ and 
	$e_j$ evaluates to (distinct) $i_j$ for $1\leq j\leq k$
	in the current classical state. Formally, it denotes the following program:	
	\[
		\measstm{\exists i.(e_i <1 \vee e_i >n) \vee \exists i,j. (i\neq j \wedge e_i = e_j)}{\abort}{S_1; S_2; \cdots; S_{|R|}}
	\] where each $S_\ell$ is of the form
		\[
	\measstm{e_1=i_1\wedge \cdots \wedge e_k=i_k}	
	{q_{i_1}, \cdots, q_{i_k} \apply \u(e)}{\sskip}
	\] 
	and $(i_1, \cdots, i_k)$ ranges over
		\[
	R \define \{(i_1, \cdots, i_k) : \forall j. 1\leq i_j \leq n \mbox{ and $ i_j$'s are distinct}\}.
	\]
	 Again, the order of $S_\ell$'s is actually irrelevant as there is at most one that will be executed.
		
\end{itemize}

\subsection{Operational and denotational semantics}
A \emph{configuration} is a triple $\<S, \cstate, \rho\>$ where
$S\in \prog \cup \{E\}$, $E$ is a special symbol to denote termination, $\cstate\in \Sigma$, and $\rho\in \d(\h_V)$ for some $V$ subsuming $\qv(S)$. The operational semantics of programs in $\prog$ is defined as the smallest transition relation $\rightarrow$ on configurations given in Table~\ref{tbl:opsemantics}. Note that there is no transition rule for $\abort$, meaning that the statement $\abort$ simply halts the computation with no proper state reached.

The definition is rather standard and intuitive. We would only like to point out that motivated by~\cite{ying2012floyd}, the operational semantics of quantum measurements (and even probabilistic assignments) are described in a non-deterministic way, while the probabilities of different branches are encoded in the quantum part of the configurations. That is why we need to take partial density operators instead of the normalised density operators as the representation of quantum states.

Similar to~\cite{ying2012floyd}, denotational semantics of cq-programs can be derived from the operational one by summing up all the cq-states obtained by terminating computations.
 
{\renewcommand{\arraystretch}{2.5}
\begin{table}[t]
	\begin{lrbox}{\tablebox}
		\centering
		\begin{tabular}{ll}
			$\<\sskip, \cstate, \rho\> \ra \<E, \cstate, \rho\>$ & $\<x:= e, \cstate, \rho\> \ra \<E, \cstate[\subs{x}{\cstate(e)}], \rho\>$ \\
	 $\<q:=0, \cstate, \rho\> \ra \<E, \cstate, \sum_{i=0}^{d_{q}-1}\qzi \rho\qiz\>$ & $\<\bar{q}\apply U, \cstate, \rho\> \ra \<E, \cstate, U_{\bar{q}} \rho U_{\bar{q}}^\dag\>$\\
	 $\displaystyle\frac{d\in D_{\mathit{type}(x)}}{\<x\rassign g, \cstate, \rho\> \ra \<E, \cstate[\subs{x}{d}], g(d)\cdot \rho\>}$ &
			 	$\displaystyle\frac{\m = \{M_i : i\in I\}}{\<x:= \measure\ \m[\bar{q}], \cstate, \rho\> \ra \<E, \cstate[\subs{x}{i}], M_i \rho M_i^\dag\>}$\\
			 $\displaystyle\frac{\<S_0, \cstate, \rho\> \ra \<S', \cstate', \rho'\>}{\<S_0; S_1, \cstate, \rho\> \ra \<S'; S_1, \cstate', \rho'\>}$\ where $E; S_1 \equiv S_1$ & \\
			$\displaystyle\frac{\cstate \models b}{\<\mstm, \cstate, \rho\> \ra  \<S_1, \cstate, \rho\>}$
			& $\displaystyle\frac{\cstate \models \neg b}{\<\mstm, \cstate, \rho\> \ra  \<S_0, \cstate, \rho\>}$\\
			 $\displaystyle\frac{\cstate \models \neg b}{\<\wstm, \cstate, \rho\> \ra  \<E, \cstate, \rho\>}$&
			$\displaystyle\frac{\cstate \models b}{\<\wstm, \cstate, \rho\> \ra  \<S; \wstm, \cstate, \rho\>}$\\
		\end{tabular}
	\end{lrbox}
	\resizebox{\textwidth}{!}{\usebox{\tablebox}}\\
	\vspace{4mm}
	\caption{Operational semantics for cq-programs. 
	}
	\label{tbl:opsemantics}
\end{table}
}

\begin{definition}
	Let $S\in \prog$, and $\<\cstate, \rho\>\in \qstatesh{V}$ with $V\supseteq \qv(S)$.
	\begin{itemize}
		\item
		A \emph{computation} of $S$ starting in $\<\cstate, \rho\>$ is a (finite or infinite) maximal sequence of configurations $\<S_i, \cstate_i, \rho_i\>$, $i\geq 1$, such that
		$$\<S, \cstate, \rho\> \ra \<S_1, \cstate_1, \rho_1\> \ra \<S_2, \cstate_2, \rho_2\> \ra \cdots$$
		and $\rho_i\neq \z_{\h_V}$ for all $i$.
		
		\item A computation of $S$ \emph{terminates} in $\<\cstate', \rho'\>$ if it is finite and the last configuration is $\<E, \cstate', \rho'\>$; otherwise it is \emph{diverging}.
	\end{itemize}
\end{definition}

Let $\ra^n$ be the $n$-th composition of $\ra$, and $\ra^* {\define} \bigcup_{n\geq 0} \ra^n$. Then we have the following lemma.
\begin{lemma}\label{lem:den}
	Let $S\in \prog$, and $\<\cstate, \rho\>\in \qstatesh{V}$ with $V\supseteq \qv(S)$.
	Then
	\begin{enumerate}
		\item the multi-set $\{\<\cstate', \rho'\> : \<S, \cstate, \rho\> \ra^n \<E, \cstate', \rho'\>\}$ is countable for all $n\geq 0$;
		\item the sequence of cq-states $\{\qstate_n : n\geq 0\}$, where $$\qstate_n \define \sum\left\{\<\cstate', \rho'\> : \<S, \cstate, \rho\> \ra^k \<E, \cstate', \rho'\> \mbox{ for some }k\leq n\right\},$$  is increasing with respect to $\le$. Here we assume $\qstate_n$ to be $\bot_{V}$ if the multi-set on the right-hand side is empty. Thus
		\[
		\sum\left\{\<\cstate', \rho'\> : \<S, \cstate, \rho\> \ra^* \<E, \cstate', \rho'\>\right\} = \bigvee_{n\geq 0} \qstate_n.
		\]
	\end{enumerate}
\end{lemma}
\begin{proof}
	The first clause is easy by induction. The second one is directly from the fact that any configuration with the form $\<E, \cstate', \rho'\>$ has no further transition.
\end{proof}

With this lemma, we are able to define the denotational semantics of cq-programs using the operational one. Let $\qstatesh{\supseteq \qv(S)} \define\bigcup_{V\supseteq qv(S)} \qstatesh{V}$.

\begin{definition}\label{def:denotational}
	Let $S\in \prog$. The \emph{denotational semantics} of $S$ is a mapping 
	$$\sem{S} : \qstatesh{\supseteq \qv(S)}\ra  \qstatesh{\supseteq \qv(S)}$$
	such that for any $\<\cstate, \rho\>\in  \qstatesh{V}$ with $V\supseteq qv(S)$,
	\[
	\sem{S}(\cstate, \rho) = \sum\left\{\<\cstate', \rho'\> : \<S, \cstate, \rho\> \ra^* \<E, \cstate', \rho'\>\right\}.
	\]
	Furthermore, let $\sem{S}(\qstate) = \sum_{i\in I}\sem{S}(\cstate_i, \rho_i)$ whenever  $\qstate = \bigoplus_{i\in I}\<\cstate_i, \rho_i\>$. 
\end{definition}

To simplify notation, we always write $(\cstate, \rho)$ for  $(\<\cstate, \rho\>)$ when $\<\cstate, \rho\>$ appears as a parameter of some function. 
The next lemma guarantees the well-definedness of Definition~\ref{def:denotational}.

\begin{lemma}\label{lem:welldef}
	For any $S\in \prog$ and $\qstate\in \qstatesh{V}$ with $V\supseteq \qv(S)$, 
	\begin{enumerate}
		\item $\tr(\sem{S}(\qstate)) \leq \tr(\qstate)$, and so $\sem{S}(\qstate) \in \qstates$;
		\item $\sem{S}(\qstate) = \sum_i \lambda_i \sem{S}(\qstate_i)$ whenever $\qstate = \sum_i \lambda_i\qstate_i$.
	\end{enumerate}
\end{lemma}
\begin{proof}
	Clause (2) is easy. For (1), we prove by induction on $n$ that
	$\tr(\qstate_n) \leq \tr(\rho)$ whenever $\qstate = \<\cstate, \rho\>$ and $\qstate_n$ is defined as in~Lemma~\ref{lem:den}(2).
	Thus the result holds for simple cq-states. The general case follows easily.
\end{proof}

 To illustrate the concepts and techniques introduced in this paper, we take Grover's search algorithm~\cite{grover1996fast} as a running example. More case studies are presented in Sec.~\ref{sec:case}.

\begin{example}[Grover's algorithm]\label{ex:grover}	
	Suppose we are given an (unstructured) database with $N$ items, $D$ of which are of our concern (called solutions) with $0<D<N/2$. For simplicity, we assume $N=2^n$ for some positive integer $n$.
	Let $\theta\in (0, \pi/2)$ such that
	$$\cos \frac\theta 2 = \sqrt{\frac{2^n-D}{2^n}},$$
	and $K$ be the integer in  $(\frac{\pi}{2\theta} - 1, \frac{\pi}{2\theta}]$. Then Grover's search algorithm can be described in our quantum language (with syntactic sugars) as
	\begin{align*}
	\mathit{Grover} \define &\\
	&\bar{q} := 0; \ \bar{q} \apply H^{\otimes n};\ x := 0;\\
	&\while\ x<K\ \ddo\\
	&\qquad  \bar{q} \apply G; \  x := x +1;\\
	& \pend\\
	&y:= \measure\ \bar{q}
	\end{align*} 
	where $\bar{q} = q_1, \ldots, q_n$ and each $q_i$ has $\tyqubit$-type, $H=|+\>\<0| + |-\>\<1|$ is the Hadamard operator with $|+\> \define \frac 1{\sqrt{2}}(|0\>+|1\>)$ and $|-\> \define \frac 1{\sqrt{2}}(|0\>-|1\>)$. $G = (2|\psi\>\<\psi| - I)O$ is the Grover rotation where $|\psi\> = \frac 1{\sqrt{2^n}} \sum_{i=0}^{2^n-1} |i\>$ and $O$ is the Grover oracle which maps $|i\>$ to $-|i\>$ when $i$ is a solution while to $|i\>$ otherwise.
		
	The (terminating) computations of $\mathit{Grover}$ starting in any $\<\cstate, \rho\>\in \qstatesh{\bar{q}}$ are shown as follows.
	\begin{eqnarray*}
		& &\left\<\mathit{Grover}, \cstate, \rho\right\>\\
		&\ra^n & \left\<\bar{q} \apply H^{\otimes n};\ \cdots, \cstate, |0^{\otimes n}\>\<0^{\otimes n}|\right\>\\
		&\ra & \left\<x := 0;\ \cdots, \cstate, |+^{\otimes n}\>\<+^{\otimes n}|\right\>\\
		&\ra & \left\<\while; y:= \measure\ \bar{q}, \cstate[\subs{x}{0}], |+^{\otimes n}\>\<+^{\otimes n}|\right\>\\
		&\ra & \left\<\bar{q} \apply G;\ x := x +1;\ \while; y:= \measure\ \bar{q}, \cstate[\subs{x}{0}], |+^{\otimes n}\>\<+^{\otimes n}|\right\>\\	
		&\ra & \left\<x := x +1;\ \while; y:= \measure\ \bar{q}, \cstate[\subs{x}{0}], G|+^{\otimes n}\>\<+^{\otimes n}|G^\dag\right\>\\	
		&\ra & \left<\while; y:= \measure\ \bar{q}, \cstate[\subs{x}{1}], G|+^{\otimes n}\>\<+^{\otimes n}|G^\dag\right\>\\			
		&\ra & \cdots\ \cdots\\
		&\ra & \left\<\while; y:= \measure\ \bar{q}, \cstate[\subs{x}{K}], G^K|+^{\otimes n}\>\<+^{\otimes n}|G^{K\dag}\right\>\\			
		&\ra & \left\<y:= \measure\ \bar{q}, \cstate[\subs{x}{K}], G^K|+^{\otimes n}\>\<+^{\otimes n}|G^{K\dag}\right\>\\
		&\ra & \left\<E, \cstate[\subs{x}{K}, \subs{y}{i}], |\<i|G^K|+^{\otimes n}\>|^2 \cdot |i\>\<i|\right\>.
	\end{eqnarray*}
	for all $0\leq i<2^n$.
	We write $\while$ for the while loop in the program. Consequently, 
	\begin{equation}\label{eq:grover}
	\sem{\mathit{Grover}}(\cstate, \rho) = \sum_{i=0}^{2^n-1}\left\<\cstate[\subs{x}{K}, \subs{y}{i}], |\<i|G^K|+^{\otimes n}\>|^2 \cdot |i\>\<i|\right>.
	\end{equation}
	
	Let $Sol\subseteq \{0, \cdots, 2^n-1\}$ be the set of solutions, $|Sol| = D$, and
	$$|\alpha\> = \frac 1 {\sqrt{2^n - D}} \sum_{i \not\in Sol} |i\>, \quad
	|\beta\> = \frac 1 {\sqrt{D}} \sum_{i \in Sol} |i\>.$$
	Then we have 
	$|\psi\> = |+^{\otimes n}\> = \cos\frac{\theta}{2} |\alpha\> + \sin\frac{\theta}{2} |\beta\>$, 
		$$G |\alpha\> = \cos \theta  |\alpha\> + \sin \theta |\beta\>, \quad
	G |\beta\> =  -\sin \theta  |\alpha\> + \cos \theta |\beta\>.$$
	That is, the effect of $G$ in the two-dimensional real space spanned by $|\alpha\>$ and $|\beta\>$ is a rotation with angle $\theta$ (note that $|\alpha\>$ and $|\beta\>$ are orthogonal). Thus the success probability of finding a solution by Grover's algorithm, i.e. the probability of $y\in Sol$ after its execution, can be computed as
	\begin{equation}\label{eq:psucc}
	p_{\mathrm{succ}} = \sum_{i\in Sol} \left|\<i|G^K|+^{\otimes n}\>\right|^2 = \sin^2\left(\frac{2K+1}{2}\theta\right).
	\end{equation}
	Recall that $K\in (\frac{\pi}{2\theta} - 1, \frac{\pi}{2\theta}]$. Thus 
	\[
	1-p_{\mathrm{succ}}  \leq  \left|(2K+1)\theta-\pi\right| \leq \theta.
	\]
	In other words, Grover's algorithm succeeds with a probability at least $1-O(\sqrt{D/N})$, and runs in time $O(\sqrt{N/D})$, achieving a quadratic speed-up over the best classical algorithms which run in $O(N/D)$ time.	
\end{example}

The following lemma presents the explicit form for denotational semantics of various program constructs.

\begin{lemma}\label{lem:iout} For any cq-state $\<\cstate, \rho\>$ in $\s_V$ where $V$ contains all quantum variables of the corresponding program,
	\begin{enumerate}
		\item  $\sem{\sskip}(\cstate, \rho) = \<\cstate, \rho\>$;
		\item  $\sem{\abort}(\cstate, \rho) = \bot_{V}$;
		\item  $\sem{x:=e}(\cstate, \rho) = \<\cstate[\subs{x}{\cstate(e)}], \rho\>$;
		\item  $\sem{x\rassign g}(\cstate, \rho) = \sum_{d\in D_{\mathit{type}(x)}} \<\cstate[\subs{x}{d}], g(d) \cdot \rho\>$;
		\item $\sem{x:= \measure\ \m[\bar{q}]}(\cstate, \rho) = \sum_{i\in I}  \<\cstate[\subs{x}{i}], M_i \rho M_i^\dag\>$ where $M_i$'s are applied on $\bar{q}$, and $\m=\{M_i : i\in I\}$;
		\item  $\sem{q:=0}(\cstate, \rho) = \<\cstate, \sum_{i=0}^{d_q-1}\qzi \rho\qiz\>\>$;
		\item $\sem{\bar{q}\apply U}(\cstate, \rho) = \<\cstate, U_{\bar{q}} \rho U_{\bar{q}}^\dag\>$;
		\item $\sem{S_0; S_1}(\cstate, \rho) = \sem{S_1}(\sem{S_0}(\cstate, \rho))$;
		\item $\sem{(S_0; S_1) ; S_2} = \sem{S_0; (S_1 ; S_2)}$;
		\item $\sem{\mstm}(\cstate, \rho) = \sem{S_1}(\cstate, \rho)$ if $\sigma \models b$, and $\sem{S_0}(\cstate, \rho)$ otherwise;
		\item\label{item:loop} $\sem{\while}(\cstate, \rho)= \bigvee_n (\sem{(\while)^n}(\cstate, \rho))$, where
		$\while \define \wstm$, 
		$(\while)^0 \define \abort$, and for any $n\geq 0$,
		$$(\while)^{n+1} \define \measstm{b}{S; (\while)^n}{\sskip}.$$
	\end{enumerate}
\end{lemma}
\begin{proof}
	We only prove \ref{item:loop} as an example; the others are simpler. For any $\<\cstate, \rho\>\in \qstatesh{V}$ with $V\supseteq \qv(\while)$, let $\Pi$ be the set of all terminating computations of $\while$ starting in $\<\cstate, \rho\>$. Furthermore, let $\Pi_0 \define \emptyset$, and for $n\geq 1$ let
	$$\Pi_n \define \{\pi\in \Pi : \#\{i : \mathit{prog}(\pi[i]) =\while\} \leq n\}$$
	be the set of computations in $\Pi$ in which the loop has iterated for no more than $n$ times before termination. Here $\mathit{prog}(\pi[i])$ is the program (the first component) of the $i$-th configuration of $\pi$. Obviously, $\Pi = \bigcup_{n\geq 0} \Pi_n$ and 
	$$\sem{\while}(\cstate, \rho) = \sum_{\pi\in \Pi} \<\cstate_\pi, \rho_\pi\> = \bigvee_{n\geq 0} \sum_{\pi\in \Pi_n} \<\cstate_\pi, \rho_\pi\>$$ 
	where we assume each computation $\pi\in \Pi$ ends with $\<E, \cstate_\pi, \rho_\pi\>$. The result then follows from the fact that 
	\[
	\sem{(\while)^n}(\cstate, \rho) = \sum_{\pi\in \Pi_n} \<\cstate_\pi, \rho_\pi\>
	\]
	which is easy to observe.
\end{proof}

The next lemma gives a recursive description of the semantics of while loops.
\begin{lemma}\label{lem:whilesem}
	Let $\while \define \wstm$. For any $\qstate\in \s_V$ with $V\supseteq  \qv(\while)$ and $n\geq 0$,
	\[
	\sem{(\while)^{n+1}}(\qstate)
	= \qstate |_{\neg b} + \sem{(\while)^n}(\sem{S}(\qstate |_b)) \]
	Consequently,
	\[
	 \sem{\while}(\qstate) = \qstate |_{\neg b} + \sem{\while}(\sem{S}(\qstate |_b)).
	\]
\end{lemma}
\begin{proof}
	Easy from Lemma~\ref{lem:iout}.
\end{proof}

Finally, we can easily compute the operational semantics of the syntactic sugars introduced in Sec.~\ref{sec:syntax}.

\begin{lemma}\label{lem:densugar}
	Let  $\bar{q}\define q_1,\ldots,\q_n$, $1\leq k\leq n$, and $\u \define \{U_i : 1\leq i\leq K\}$.
	For any cq-state $\<\cstate, \rho\>$ in $\s_V$ where $V$ contains all quantum variables of the corresponding program,
\begin{enumerate}
		\item  $\sem{\bar{q}:=0}(\cstate, \rho) = \<\cstate, \sum_{i=0}^{d_{\bar{q}}-1}\quzi \rho\quiz\>\>$;
	\item $\sem{x:= \measure\ \bar{q}}(\cstate, \rho) = \sum_{i=0}^{d_{\bar{q}}-1}  \<\cstate[\subs{x}{i}], \quii \rho \quii\>$;
	\item $\sem{\bar{q}[e_1, \cdots, e_k]\apply {\u}(e)}(\cstate, \rho) = \bot_V$ if $\cstate(e) <1$, $\cstate(e)>K$, there exists $i$ such that $\cstate(e_i) <1$ or  $\cstate(e_i) > n$, or $\cstate(e_j)$'s are not distinct; otherwise it equals $\<\cstate, {U_i} \rho {U_i}^\dag\>$ where $i=\cstate(e)$ and $U_i$ is applied on $q_{\cstate(e_1)}, \cdots, q_{\cstate(e_k)}$.
\end{enumerate}
\end{lemma}
\begin{proof}
	Routine, using Lemma~\ref{lem:iout}.
\end{proof}

\subsection{Correctness formula}

As usual, program correctness is expressed by \emph{correctness formulas} with the form
$$\ass{\qassert}{S}{\qassertp}$$
where $S$ is a cq-program, and $\qassert$ and $\qassertp$ are both cq-assertions.
Note here that we do not put any requirement on the quantum variables which $\qassert$ and $\qassertp$ are acting on. In fact, the sets $\qv(S)$, $\qv(\qassert)$, and $\qv(\qassertp)$ can be all different. 

The following definition is a direct extension of the corresponding one in~\cite{ying2012floyd}, with the new notions of cq-states and assertions.
\begin{definition}
	Let $S$ be a cq-program, and $\qassert$ and $\qassertp$ cq-assertions.
	\begin{enumerate}
		\item We say the correctness formula $\ass{\qassert}{S}{\qassertp}$ is true in the sense of \emph{total correctness}, written $\models_{{\mathit{tot}}} \ass{\qassert}{S}{\qassertp}$, if for any 
		$V\supseteq \qv(S, \qassert, \qassertp)$ and
		$\qstate\in \qstatesh{V}$,
		$$\Exp(\qstate\models \qassert) \leq \Exp(\sem{S}(\qstate) \models \qassertp).$$
		\item We say the correctness formula $\ass{\qassert}{S}{\qassertp}$ is true in the sense of \emph{partial correctness}, written $\models_{\mathit{par}} \ass{\qassert}{S}{\qassertp}$, if for any 
		$V\supseteq \qv(S, \qassert, \qassertp)$ and
		$\qstate\in \qstatesh{V}$,
		$$\Exp(\qstate\models \qassert) \leq \Exp(\sem{S}(\qstate) \models \qassertp)+  \tr(\qstate) - \tr( \sem{S}(\qstate)).$$
	\end{enumerate}
\end{definition}

The next lemma shows that the validity of correctness formulas can be checked on simple cq-states.
\begin{lemma}\label{lem:formulasimple}
		Let $S$ be a cq-program, $\qassert$ and $\qassertp$ be cq-assertions, and $V\define \qv(S, \qassert, \qassertp)$. Then 
	\begin{enumerate}
		\item  $\models_{{\mathit{tot}}} \ass{\qassert}{S}{\qassertp}$ iff for any 
		$\<\cstate, \rho\> \in \qstatesh{V}$ with $\tr(\rho) = 1$,
		$$\Exp(\<\cstate, \rho\>\models \qassert) \leq \Exp(\sem{S}(\cstate, \rho) \models \qassertp).$$
		\item $\models_{\mathit{par}} \ass{\qassert}{S}{\qassertp}$ iff for any 
		$\<\cstate, \rho\> \in \qstatesh{V}$ with $\tr(\rho) =1$,
		$$\Exp(\<\cstate, \rho\>\models \qassert) \leq \Exp(\sem{S}(\cstate, \rho) \models \qassertp)+  \tr(\rho) - \tr( \sem{S}(\cstate, \rho)).$$
	\end{enumerate}
\end{lemma}
\begin{proof}
	Easy from linearity of $\sem{S}$ for any cq-program $S$; see Lemma~\ref{lem:welldef}(2).	
\end{proof}

\begin{example}\label{ex:groverform}
	We have proven in Example~\ref{ex:grover} that
	no matter what the initial (classical and quantum) state is, the output (value of $y$) of Grover's algorithm lies in $\mathit{Sol}$ with probability $p_{\mathit{succ}}$. This correctness can be stated in the following form
		\begin{equation}\label{eq:groverform}
	\models_{{\mathit{tot}}}\ass{ p_{\mathit{succ}}}{\mathit{Grover}}{y\in Sol},
	\end{equation}
	which claims that the postcondition $y\in Sol$ can be established by $\mathit{Grover}$ with probability $p_{\mathit{succ}}$.
Recall that in Eq.(\ref{eq:groverform}), $p_{\mathit{succ}}$ denotes $\<\true, p_{\mathit{succ}}\>$ and $y\in Sol$ denotes $\<y\in Sol, 1\>$. Both the pre- and post-conditions being purely classical means that the initial and final quantum states are irrelevant.

Note that $qv(\mathit{Grover}) = \bar{q}$. For any $\<\cstate, \rho\> \in \qstatesh{\bar{q}}$ with $\tr(\rho) =1$, we have from Eq.(\ref{eq:grover}) that
\begin{align*}
 \Exp(\sem{\mathit{Grover}}(\cstate, \rho) \models y\in Sol) & = \sum_{i\in Sol} |\<i|G^K|+^{\otimes n}\>|^2 \\
 &= p_{\mathit{succ}} =
	\Exp(\<\cstate, \rho\> \models p_{\mathit{succ}}).
\end{align*}
Then Eq.(\ref{eq:groverform}) follows from Lemma~\ref{lem:formulasimple}.
\end{example}

Finally, we show some basic facts about total and partial correctness as follows.

\begin{lemma}\label{lem:corf}
	Let $S$ be a cq-program, $\qassert$ and $\qassertp$ be cq-assertions, and $V\subseteq \qVar$. 
	\begin{enumerate}
		\item If $\models_{{\mathit{tot}}} \ass{\qassert}{S}{\qassertp}$ then $\models_{\mathit{par}} \ass{\qassert}{S}{\qassertp}$;
		\item $\models_{{\mathit{tot}}} \ass{\emptydis_{V}}{S}{\qassertp}$;
		\item  $\models_{\mathit{par}} \ass{\qassert}{S}{\top_{V}}$;
		\item If $\models_{{\mathit{tot}}} \ass{\qassert_i}{S}{\qassertp_i}$ and  $\lambda_i\geq 0$ for $i=1,2$, then 
		\[
		 \models_{{\mathit{tot}}} \ass{ \lambda_1\qassert_1 +  \lambda_2\qassert_2}{S}{ \lambda_1 \qassertp_1 + \lambda_2 \qassertp_2}.
		 \]
		 The result also holds for partial correctness if $\lambda_1 +  \lambda_2 =1$.
	\end{enumerate}
\end{lemma}
\begin{proof}
	(1) follows from the definitions, (2) and (3) from Lemmas~\ref{lem:bpdeg} and~\ref{lem:qassetorder}, and (4) from Lemma~\ref{lem:bpdeg}.
\end{proof}

\subsection{Weakest (liberal) precondition semantics}\label{sec:wps}

Recall that in classical programming theory, the weakest (liberal) precondition of an assertion $\cassert$ with respect to a given program $S$ characterises the largest set of states $\cstate$ which (upon termination) guarantee that the final states $\sem{S}(\cstate)$ satisfy $\cassert$. Consequently, a program can also be regarded as a predicate transformer which maps any postcondition to its weakest (liberal) precondition. In the following, we extend these semantics to our cq-programs. Let $\qassertsh{\supseteq \qv(S)} \define\bigcup_{V\supseteq qv(S)} \qassertsh{V}$.

\begin{definition}\label{def:weakest}
	Let $S\in \prog$. The \emph{weakest precondition semantics} $wp.S$ and \emph{weakest liberal precondition semantics} $wlp.S$ of $S$ are both mappings 
	$$\qassertsh{\supseteq \qv(S)}\ra  \qassertsh{\supseteq \qv(S)}$$
	defined inductively in Table~\ref{tbl:wpsemantics}. To simplify notation, we use $xp$ to denote both $wp$ and $wlp$ whenever it is applicable for both of them.
\end{definition}

We follow the standard notations $wp.S.\qassert$ and $wlp.S.\qassert$ to denote weakest (liberal) preconditions~\cite{dijkstra1976discipline, morgan1996probabilistic,ying2012floyd}.
The well-definedness of Definition~\ref{def:weakest} follows from the observation that $xp.S$ is monotonic on $\h_V$ (with respect to $\le_V$; see Lemma~\ref{lem:wpcorres}\ref{cl:wpmono} below) for any cq-program $S$ and $V\supseteq qv(S)$.
The weakest (liberal) precondition semantics in Table~\ref{tbl:wpsemantics} is a natural extension of the corresponding semantics of both probabilistic~\cite{morgan1996probabilistic} and purely quantum~\cite{ying2012floyd} programs. For example, the weakest precondition for conditional branching is defined in~\cite{morgan1996probabilistic} as
	$$wp.(\iif\ b\ \then \ S_1\ \eelse\ S_0\ \pend).\beta \define b\times wp.S_1.\beta + (\neg b) \times wp.S_0.\beta$$ 
	where on the right-hand side $b$ and $\neg b$ are regarded as $\{0,1\}$-valued functions on the states space, and $\beta$ is a probabilistic assertion (a non-negative random variable; see Table~\ref{tbl:comparison}). This coincides with the corresponding definition in Table~\ref{tbl:wpsemantics}, as $b\wedge \qassertp$ is exactly $b\times \qassertp$ for $\tybool$-type expression $b$ and probabilistic assertion  $\qassertp$.

The following lemma shows a duality relation between the denotational and weakest (liberal) precondition semantics of cq-programs.

{\renewcommand{\arraystretch}{1.5}
	\begin{table}[t]
		\begin{lrbox}{\tablebox}
			\centering
			\begin{tabular}{l}
			\begin{tabular}{ll}
				$xp.\sskip.\qassert = \qassert$ & $xp.(x\rassign g).\qassert =  \displaystyle\sum_{d\in D_{\mathit{type(x)}}}g(d)\cdot \qassert[\subs{x}{d}]$
				\\$xp.(x:=e).\qassert =  \qassert[\subs{x}{e}]$ 
			  &$xp.(x:= \measure\ \m[\bar{q}]).\qassert  = \displaystyle\sum_{i\in I}M_i^\dag \qassert[\subs{x}{i}]M_i$\\
			 $xp.(\bar{q}\apply U).\qassert  = U_{\bar{q}}^\dag \qassert U_{\bar{q}}$ &$xp.(q:=0).\qassert =  \displaystyle\sum_{i=0}^{d_{q}-1} \qiz \qassert\qzi$\\
			$xp.(S_0; S_1).\qassert = xp.S_0.(xp.S_1.\qassert)$\hspace{3.5em}
			&$xp.(\mstm).\qassert = b\wedge xp.S_1.\qassert + \neg b \wedge xp.S_0.\qassert$\\
						\smallskip
			$wlp.\abort.\qassert = \top_{V}$ & $wp.\abort.\qassert = \bot_{V}$\\
						\end{tabular}\\ 
					\begin{tabular}{c}
			$wlp.(\wstm).\qassert = \bigwedge_{n\geq 0} \qassert_n$, where 
			$\qassert_0 \define \top_{V}$, and for any $n\geq 0$,\\
			$\qassert_{n+1} \define \neg b \wedge \qassert + b\wedge wlp.S.\qassert_n.$\\
						$wp.(\wstm).\qassert = \bigvee_{n\geq 0} \qassert_n$, where 
			$\qassert_0 \define \bot_{V}$, and for any $n\geq 0$,\\
			$\qassert_{n+1} \define \neg b \wedge \qassert + b\wedge wp.S.\qassert_n.$
			\end{tabular}			
		\end{tabular}
		\end{lrbox}
		\resizebox{\textwidth}{!}{\usebox{\tablebox}}\\
		\vspace{4mm}
		\caption{Weakest (liberal) precondition semantics for cq-programs, where $xp\in \{wp, wlp\}$. 
		}
		\label{tbl:wpsemantics}
	\end{table}
}

\begin{lemma}\label{lem:wpwlp}
		Let $S$ be a cq-program, $\qstate$ a cq-state, and $\qassert$ a cq-assertion with $qv(\qstate) \supseteq qv(\qassert)\supseteq qv(S)$. Then
	\begin{enumerate}
		\item $qv(wp.S.\qassert) = qv(wlp.S.\qassert) = qv(\qassert)$;
	\item $\Exp(\qstate\models wp.S.\qassert) =  \Exp(\sem{S}(\qstate)\models \qassert)$;
	\item $\Exp(\qstate\models wlp.S.\qassert) =  \Exp(\sem{S}(\qstate)\models \qassert) + \tr(\qstate) -  \tr(\sem{S}(\qstate))$.
	\end{enumerate}
\end{lemma}
\begin{proof}
	We prove this lemma by induction on the structure of $S$. The basis cases are easy from the definition. We only show the following two cases for clause (2) as examples.
	\begin{itemize}
		\item Let $S\define \mstm$. Then 
		\begin{align*}
			\Exp(\qstate\models wp.S.\qassert) &=\Exp(\qstate\models b\wedge wp.S_1.\qassert + \neg b \wedge wp.S_0.\qassert)\\
			& = 	\Exp(\qstate|_b\models wp.S_1.\qassert) + \Exp(\qstate|_{\neg b}\models wp.S_0.\qassert) \\
			& =\Exp(\sem{S_1}(\qstate|_b)\models \qassert) + \Exp(\sem{S_0}(\qstate|_{\neg b})\models \qassert)
		\end{align*}
		which is exactly $\Exp(\sem{S}(\qstate)\models \qassert)$.
		
		\item Let $S\define \whilestm{b}{S'}$. Let $V \define qv(\qassert)$, $\qassert_0 \define \bot_{V}$, and for any $n\geq 0$,
		$\qassert_{n+1} \define \neg b \wedge \qassert + b\wedge wp.S'.\qassert_n.$
		First, we show by induction that for any $n\geq 0$ and $\qstate'\in \qstatesh{V}$,
		$$\Exp(\qstate' \models \qassert_n) = \Exp(\sem{S^n}(\qstate') \models \qassert).$$
		The case of $n=0$ follows from the definition. We further calculate from Lemmas~\ref{lem:bpdeg} and~\ref{lem:whilesem} that
		\begin{align*}
			\Exp(\qstate' \models \qassert_{n+1}) &= \Exp(\qstate'\models \neg b \wedge \qassert) + \Exp(\qstate' \models b\wedge wp.S'.\qassert_n)\\
			&=\Exp(\qstate'|_{\neg b} \models  \qassert) + \Exp(\qstate'|_{b} \models wp.S'.\qassert_n)\\
			&=\Exp(\qstate'|_{\neg b} \models  \qassert) + \Exp(\sem{S'}(\qstate'|_{b}) \models \qassert_n)\\
			&=\Exp(\qstate'|_{\neg b} \models  \qassert) +
			\Exp(\sem{S^n}(\sem{S'}(\qstate'|_{b}) \models \qassert))\\
			&=\Exp(\sem{S^{n+1}}(\qstate') \models \qassert).
		\end{align*}
		Thus from Lemma~\ref{lem:qasset},
		$$\Exp(\qstate \models  wp.S.\qassert) =\Exp(\qstate \models  \bigvee_{n\geq 0} \qassert_n) =  \Exp(\sem{S}(\qstate) \models \qassert).$$
	\end{itemize}
\end{proof}

We can also compute the weakest (liberal) precondition semantics of the syntactic sugars introduced in Sec.~\ref{sec:syntax}.

\begin{lemma}\label{lem:wpsugar}
	Let  $\bar{q}\define q_1,\ldots, q_n$, $1\leq k\leq n$, $\u \define \{U_i : 1\leq i\leq K\}$, and $xp\in \{wp, wlp\}$.
	Let $\qassert$ be a cq-assertion in $\qassertsh{V}$ with $V$ containing all quantum variables of the corresponding cq-program. Then
	\begin{enumerate}
				\item  $xp.(\bar{q}:=0).\qassert  = \sum_{i=0}^{d_{\bar{q}}-1} \quiz \qassert\quzi$;
		\item $xp.(x:= \measure\ \bar{q}).\qassert = \sum_{i=0}^{d_{\bar{q}}-1}\quii\qassert[\subs{x}{i}]\quii$;
		\item Let $S \define \bar{q}[e_1, \cdots, e_k]\apply {\u}(e)$. Then
		\[
		xp.S.\qassert = \sum_{i_1, \cdots, i_k= 1, \mathit{distinct}}^{n} \sum_{i=1}^K \bigwedge_{j=1}^k (e_j = i_j)\wedge (e=i) \wedge \mathcal{U}^i_{i_1, \cdots, i_k}(\qassert)\]
		 where $\mathcal{U}^i_{i_1, \cdots, i_k}(\qassert) \define
		U_i^\dag \qassert {U_i}
		$
		and $U_i$ is applied on ${q}_{i_1}, \cdots, {q}_{i_k}$.
	\end{enumerate}
\end{lemma}
\begin{proof}
	Direct from  Lemmas~\ref{lem:densugar} and~\ref{lem:wpwlp}.
\end{proof}

The following collects some properties of the weakest (liberal) precondition semantics.

\begin{lemma}\label{lem:wpcorres}
	Let $S$ be a cq-program, $\qstate$ a cq-state, and $\qassert$ a cq-assertion with $qv(\qstate) \supseteq qv(\qassert)\supseteq qv(S)$. Let $xp\in \{wp, wlp\}$. Then
	\begin{enumerate}
		\item\label{cl:wpwlp} $wp.S.\qassert + wlp.S.(\top_{qv(\qassert)}-\qassert) = \top_{qv(\qassert)}$;
		\item\label{cl:wpmono}  the function $xp.S$ is monotonic; that is, for all $\qassert_1 \le \qassert_2$,
		$
		xp.S.\qassert_1 \le xp.S.\qassert_2;
		$
		\item\label{cl:wplinear}  the function $wp.S$ is linear; that is, for all $\qassert_1, \qassert_2\in \qasserts$, 
		\[
		wp.S.(\lambda_1\qassert_1 + \lambda_2\qassert_2) = 
		\lambda_1 wp.S.\qassert_1 + \lambda_2 wp.S.\qassert_2;\]
		\item\label{cl:wlplinear}  the function $wlp.S$ is affine-linear; that is, for all $\qassert_1, \qassert_2\in \qasserts$ and $\lambda_1 + \lambda_2 =1$,
		\[
		wlp.S.(\lambda_1\qassert_1 + \lambda_2\qassert_2) = 
		\lambda_1 wlp.S.\qassert_1 + \lambda_2 wlp.S.\qassert_2.\]
				
		\item\label{cl:superoper} if $W\cap qv(\qassert) \subseteq V \subseteq qv(\qassert)$, $(V\cup W)\cap qv(S) = \emptyset$, and $\f_{V\ra W}$ is a completely positive and sub-unital super-operator, then
		$$
		\f_{V\ra W}(wp.S.\qassert)  = wp.S.(\f_{V\ra W}(\qassert))
		$$
		and 
		$$\f_{V\ra W}(wlp.S.\qassert)  \le wlp.S.(\f_{V\ra W}(\qassert)).
		$$
		The equality holds for $wlp$ as well if $\f_{V\ra W}$ is unital;
	\end{enumerate}
\end{lemma}

	\begin{proof}
		We only prove~\ref{cl:superoper} as an example; other cases are simpler. Let $X \define V\cup W\cup qv(\qassert)$. For any $\qstate\in \qstatesh{X}$,
		\begin{align*}
			\Exp(\qstate \models wp.S.\f_{V\ra W}(\qassert)) &= \Exp(\sem{S}(\qstate)\models \f_{V\ra W} (\qassert)) & (Lemma~\ref{lem:wpwlp}(1))\\
			&=\Exp(\f^\dag_{W\ra V}(\sem{S}(\qstate))\models \qassert) & (Lemma~\ref{lem:bpdeg}\ref{cl:super})\\
			&=\Exp(\sem{S}(\f^\dag_{W\ra V}(\qstate))\models \qassert)  & ((V\cup W)\cap qv(S) = \emptyset)\\
			&=\Exp(\f^\dag_{W\ra V}(\qstate)\models wp.S.\qassert)& (Lemma~\ref{lem:wpwlp}(1)) \\		
			&=\Exp(\qstate\models \f_{V\ra W}(wp.S.\qassert)).& (Lemma~\ref{lem:bpdeg}\ref{cl:super})
		\end{align*}
		Thus $wp.S.\f_{V\ra W}(\qassert) = \f_{V\ra W}(wp.S.\qassert)$ from the arbitrariness of $\qstate$.

		For  $wlp$, let $Y\define qv(\qassert)\backslash V$ and $Z\define Y\cup W$. Then from the assumption $W\cap qv(\qassert) \subseteq V$ we have $Y\cap W = \emptyset$. Since $\f$ is sub-unital,
		\begin{align*}
			(I_W - \f_{V\ra W}(I_V)) \otimes (\top_{Y} - wp.S.\top_{Y}) \ge \bot_{X}.
		\end{align*} 
	Note that $\f_{V\ra W}(I_V) \otimes \top_{Y} = \f_{V\ra W}(\top_{qv(\qassert)})$. We have
	\begin{align*}
			\top_Z - \f_{V\ra W}(\top_{qv(\qassert)}) & \ge I_W\otimes  wp.S.\top_Y - \f_{V\ra W}(I_V)\otimes wp.S.\top_{Y}\\
			& = wp.S.\top_Z - \f_{V\ra W}(wp.S.\top_{qv(\qassert)})
\end{align*} 
		where the second equality follows from the assumption that $(V\cup W)\cap qv(S) = \emptyset$, and thus from clause~\ref{cl:wplinear},
			\begin{align*}
			\top_Z - wp.S.(\top_Z -  \f_{V\ra W}(\qassert)  & \ge   \f_{V\ra W}(\top_{qv(\qassert)} - wp.S.(\top_{qv(\qassert)} - \qassert)).
		\end{align*} 
		The result then follows from clause~\ref{cl:wpwlp}. It is easy to check that when $\f$ is unital, the equality actually holds.
	\end{proof}

Note that from Lemmas~\ref{lem:wpwlp} and~\ref{lem:qassetorder}, if $qv(\qassert) = qv(\qassertp) \supseteq qv(S)$, then $\models_{{\mathit{tot}}} \ass{\qassert}{S}{\qassertp}$ iff $\qassert \le wp.S.\qassertp$, and $\models_{\mathit{par}} \ass{\qassert}{S}{\qassertp}$ iff $\qassert \le wlp.S.\qassertp$. To conclude this section, we extend this result (and a similar one for partial correctness) to the general case.
\begin{lemma}\label{lem:lesssimpre}
	Let $S$ be a cq-program, and $\qassert$ and $\qassertp$ are cq-assertions. Then 
	\begin{align*}
			\models_{{\mathit{tot}}} \ass{\qassert}{S}{\qassertp} \quad &\mbox{ iff }\quad
		\qassert \lesssim wp.S.(\qassertp\otimes I_{qv(S)\backslash \qv(\qassertp)})\\
			\models_{\mathit{par}} \ass{\qassert}{S}{\qassertp} \quad &\mbox{ iff }\quad
		\qassert \lesssim wlp.S.(\qassertp\otimes I_{qv(S)\backslash \qv(\qassertp)}).
	\end{align*}
\end{lemma}
\begin{proof}
	Note that $\models_{{\mathit{tot}}} \ass{\qassert}{S}{\qassertp}$ iff $$\models_{{\mathit{tot}}} \ass{\qassert\otimes I_{qv(S, \qassertp) \backslash qv(\qassert)}}{S}{\qassertp\otimes I_{qv(S, \qassert) \backslash qv(\qassertp)}}.$$ Similar result holds for partial correctness as well. Then the lemma follows from Lemma~\ref{lem:wpcorres}\ref{cl:superoper}.
\end{proof}

\section{Hoare logic for cq-programs}

The core of Hoare logic is a proof system consisting of axioms and proof rules which enable syntax-oriented and modular reasoning of program correctness. In this section, we propose a Hoare logic for cq-programs.

\subsection{Partial correctness}

We propose in Table~\ref{tbl:psystem} the proof system for partial correctness of cq-programs, which looks quite similar to the standard Hoare logic, thanks to the novel definition of cq-assertions. Several cases deserve explanation. The side conditions in rules (Init), (Unit), and (Meas) are introduced to guarantee the well-definedness of the corresponding preconditions. They can always be satisfied by introducing `dull' quantum variables (i.e. tensor product with appropriate identity operators): if, say, $q\not\in qv(\qassert)$, then let $\qassert' \define \qassert\otimes I_q$ which is $\eqsim$-equivalent to $\qassert$ and $q\in qv(\qassert')$. An alternative way to deal with the case where $q\not\in qv(\qassert)$ in (Init) or $\bar{q}\cap \qv(\qassert) = \emptyset$ in (Unit) and (Meas) is to use the corresponding auxiliary rules introduced in Sec.~\ref{sec:aux}.

To use rule (If), we first split the precondition into two parts: $\qassert = b\wedge \qassert + \neg b\wedge \qassert$.
In the first one, all the classical states satisfy $b$, thus the first premise $\ass{b\wedge \qassert}{S_1}{\qassertp}$ is employed; in the second part, all classical states satisfy $\neg b$, thus the second premise is employed. 
As shown in Sec.~\ref{sec:wps}, this rule is essentially a quantum extension of  the corresponding rule in the expectation-based probabilistic Hoare logic~\cite{morgan1996probabilistic}. 
In contrast, more sophisticated rules are introduced in satisfaction-based probabilistic Hoare logics~\cite{ramshaw1979formalizing,den2002verifying,chadha2007reasoning,rand2015vphl} to deal with the case where probabilities of the two branches are different. This illustrates a benefit of adopting the expectation-based approach in reasoning about probabilistic and quantum programs: the quantitative assertions can encode probabilities in a natural way, making proof rules simpler than the  satisfaction-based approach.

The cq-assertion $\qassert$ in rule (While) plays a similar role of `loop invariant' as in classical programs. Finally, as the pre- and post-conditions can act on different quantum variables, we need the pre-order $\lesssim$ for rule (Imp) rather than the L\"{o}wner order in~\cite{ying2012floyd} etc.

Note also that in the rules in Table~\ref{tbl:psystem}, substitutions (say, $\qassert[\subs{x}{i}]$ in (Meas)) and Boolean operations (say, $b\wedge \qassert$ in (If)) are applied on the classical part of $\qassert$, while super-operators (say, $U_{\bar{q}}^\dag \qassert U_{\bar{q}}$ in (Unit)) are on the quantum part only.  For example, let $\qassert \define \bigoplus_{k\in K} \<\cassert_k, N_k\>$. Then rule (Meas) actually claims that if $\bar{q}\subseteq qv(\qassert)$,
\[\left\{\sum_{i\in I} \sum_{k\in K}\left\<\cassert_k[\subs{x}{i}], M_i^\dag N_k M_i\right\> \right\}{x:=\measure\ \m[\bar{q}]}\{\qassert\}.
\]
We write $\vdash_{\mathit{par}}\ass{\qassert}{S}{\qassertp}$ if the correctness formula $\ass{\qassert}{S}{\qassertp}$ can be derived using the axioms and rules presented in Table~\ref{tbl:psystem}.

{\renewcommand{\arraystretch}{2.5}
\begin{table}[t]
	\begin{lrbox}{\tablebox}
		\centering
				\begin{tabular}{l}
		\begin{tabular}{lclc}
	(Skip)	& $\ass{\qassert}{\sskip}{\qassert}$ & 
	(Abort) 	& $\ass{\top_{V}}{\abort}{\bot_V}$\\
 (Assn)	&
$\ass{\qassert[\subs{x}{e}]}{x:=e}{\qassert}$ &
(Rassn) &
$\displaystyle\left\{\sum_{d\in D_{\mathit{type}(x)}} g(d)\cdot \qassert[\subs{x}{d}]\right\}{x\rassign g}\{\qassert\}$ \\
(Init)	& $\displaystyle\frac{q\in \qv(\qassert)}{\ass{\sum_{i=0}^{d_q-1} \qiz \qassert\qzi}{q:=0}{\qassert}}$ &
 (Unit)	&
 $\displaystyle\frac{\bar{q}\subseteq \qv(\qassert)}{\ass{U_{\bar{q}}^\dag \qassert U_{\bar{q}}}{\bar{q}\apply U}{\qassert}}$ \\
 (Meas)	&
 $\displaystyle\frac{\bar{q}\subseteq \qv(\qassert), \m = \{M_i : i\in I\}}{\ass{\sum_{i\in I}M_i^\dag\qassert[\subs{x}{i}]M_i}{x:=\measure\ \m[\bar{q}]}{\qassert}}$  &
   (Seq)	&
 $\displaystyle\frac{\ass{\qassert}{S_0}{\qassert'},\ \ass{\qassert'}{S_1}{\qassertp}}{\ass{\qassert}{S_0; S_1}{\qassertp}}$\\
  (If)	&
 $\displaystyle\frac{\ass{b\wedge \qassert}{S_1}{\qassertp},\ \ass{\neg b \wedge\qassert}{S_0}{\qassertp}}{\ass{\qassert}{\mstm}{\qassertp}}$
&
 (While)	& $\displaystyle\frac{\ass{b\wedge \qassert}{S}{\qassert}
}{\ass{\qassert}{\wstm}{\neg b\wedge\qassert}}$ \\
 (Imp)	&
$\displaystyle\frac{\qassert\lesssim \qassert',\ \ass{\qassert'}{S}{\qassertp'},\ \qassertp'\lesssim \qassertp}{\ass{\qassert}{S}{\qassertp}}$&&
		\end{tabular}\\
			\end{tabular}
	\end{lrbox}
	\resizebox{\textwidth}{!}{\usebox{\tablebox}}\\
	\vspace{4mm}
	\caption{Proof system for partial correctness. 
	}
	\label{tbl:psystem}
\end{table}
}

	Recall the proof rule for loop programs in~\cite{ying2012floyd}:
	\begin{equation}\label{eq:ylooprule}\displaystyle\frac{\ass{M}{S}{\e_{\bar{q}}^0(N) + \e_{\bar{q}}^1(M)}
	}{\ass{\e_{\bar{q}}^0(N) + \e_{\bar{q}}^1(M)}{{\whilestm{\m[\bar{q}]=1}{S}}}{N}}
	\end{equation}
	where $\m = \{M_0, M_1\}$ and $\e_{\bar{q}}^i(A) \define M_i^\dag A M_i$, $i=0,1$. Now we show how this rule can be derived in our proof system, when the assertions like $M$ in Eq.(\ref{eq:ylooprule}) are replaced by cq-assertions of the form $\<\true, M\>$. That is, we are going to show
	\begin{equation}\label{eq:assum}
	\vdash_{\mathit{par}} \ass{\<\true, M\>}{S}{\left\<\true, \e_{\bar{q}}^0(N) + \e_{\bar{q}}^1(M)\right\>}
	\end{equation}
	implies
	\begin{equation}\label{eq:result}
	\vdash_{\mathit{par}} \ass{\left\<\true, \e_{\bar{q}}^0(N) + \e_{\bar{q}}^1(M)\right\>}{x:=\measure\ \m[\bar{q}]; \whilestm{x=1}{S; x:=\measure\ \m[\bar{q}]; }}{\<\true, N\>}.
	\end{equation}
	
	First we have
	\begin{align*}
	&\quad\{\<x=1, M\>\} \\
	&\quad\{\<\true, M\>\}&\mathit{(Imp)} \\
	&\quad S; & \\
	&\quad\left\{\left\<\true, \e_{\bar{q}}^0(N) + \e_{\bar{q}}^1(M)\right\>\right\}&\mathit{Eq.(\ref{eq:assum})}\\
	&\quad\left\{\left\<0=1, \e_{\bar{q}}^0(M)\right\> + \left\< 0\neq 1, \e_{\bar{q}}^0(N)\right\> +\left \<1=1, \e_{\bar{q}}^1(M)\right\> + \left\< 1\neq 1, \e_{\bar{q}}^1(N)\right\>\right\}&\mathit{(Imp)}\\
	&\quad x:= \measure\ \m[\bar{q}];\\
	&\quad\{\<x=1, M\> + \< x\neq 1, N\>\}.& \mathit{(Meas)}
	\end{align*}
	Then, using the (While) rule, 
	\[
	\vdash_{\mathit{par}} \ass{ \<x=1, M\> + \<x\neq 1, N\>}{\whilestm{x=1}{S; x:=\measure\ \m[\bar{q}]; }}{\<x\neq 1, N\>}.
	\]
	Finally, the following reasoning 
	\begin{align*}
	&\left\{\left\<\true, \e_{\bar{q}}^0(N) + \e_{\bar{q}}^1(M)\right\>\right\}\\
	&x:=\measure\ \m[\bar{q}]; \\
	&\{\<x=1, M\> + \< x\neq 1, N\>\}& \mathit{(Meas, Imp)}\\
	&\while\ x=1\ \ddo\ S;\ x:= \measure\ \m[\bar{q}];\ \pend\\
	&\{\<x\neq 1, N\>\}\\
	&\{\<\true, N\>\}& \mathit{(Imp)}
	\end{align*}
	gives us the proof of Eq.(\ref{eq:result}) as desired.

Now we show the soundness and (relative) completeness of the proof system in the sense of partial correctness.

\begin{theorem}\label{thm:psc}
	The proof system in Table~\ref{tbl:psystem} is both sound and complete with respect to the partial correctness of cq-programs.
\end{theorem}
\begin{proof}
	Soundness:  We need only to show that each rule in Table~\ref{tbl:psystem} is valid in the sense of partial correctness. Take the rule (While) as an example; the others are simpler. Let
	$\models_{\mathit{par}}  \ass{b\wedge \qassert}{S}{\qassert}$. Without loss of generality, we assume $qv(S) \subseteq qv(\qassert)$. Then $b\wedge \qassert\le wlp.S.\qassert$. We now prove by induction on $n$ that
	$ \qassert \le \qassert_n$
	for any $n\geq 0$, where $\qassert_n$ is defined as in Table~\ref{tbl:wpsemantics} for the $wlp$ semantics of $\wstm$ when the postcondition is $\neg b\wedge \qassert$. The case when $n=0$ is trivial. Then we calculate
	\begin{align*}
		\qassert_{n+1} &= \neg b\wedge (\neg b\wedge\qassert) + b\wedge wlp.S.\qassert_n\\
		&\ge \neg b\wedge \qassert + b\wedge wlp.S.\qassert\\
		&\ge \neg b\wedge \qassert + b\wedge (b\wedge \qassert)=\qassert,
	\end{align*}
	where the first inequality follows from the induction hypothesis and Lemma~\ref{lem:wpcorres}\ref{cl:wpmono}.
	Thus $$\qassert \le wlp.(\wstm).(\neg b\wedge \qassert),$$ and so 
	$$\models_{\mathit{par}} \ass{\qassert}{\wstm}{\neg b\wedge \qassert}$$ as desired.
	
	Completeness: By Lemma~\ref{lem:lesssimpre} and the (Imp) rule, it suffices to show that for any $\qassert$ and $S'$ with $qv(S')\subseteq qv(\qassert)$,
	$$\vdash_{\mathit{par}} \ass{wlp.S'.\qassert}{S'}{\qassert}.$$
	Again, we take the case for loops as an example. 
	Let $\while \define \wstm$ and $\qassertp \define wlp.\while.\qassert$.
	By induction, we have 
	$\vdash_{\mathit{par}} \ass{wlp.S.\qassertp}{S}{\qassertp}.$
	Note that
	$$\qassertp= \neg b\wedge \qassert+ b\wedge wlp.S.\qassertp.$$
	Thus $b\wedge \qassertp = b\wedge wlp.S.\qassertp \le wlp.S.\qassertp$ and 
	so $\vdash_{\mathit{par}} \ass{b\wedge \qassertp}{S}{\qassertp}$ by the (Imp) rule. Now using (While) we have
	$\vdash_{\mathit{par}} \ass{\qassertp}{\while}{\neg b\wedge \qassertp}$
	and the result follows from the fact that $\neg b\wedge \qassertp = \neg b\wedge \qassert \le \qassert$. 
\end{proof}

\subsection{Total correctness}

Ranking functions play a central role in proving total correctness of while loop programs. Recall that in the classical case, a ranking function maps each reachable state in the loop body to an element of a well-ordered set (say, the set $\N$ of non-negative integers), such that the value decreases strictly after each iteration of the loop. Our proof rule for total correctness of while loops also heavily relies on the notion of ranking assertions.

\begin{definition}
	Let $\qassert\in \qasserts$. A decreasing sequence (w.r.t. $\le$) of cq-assertions $\{\qassert_n : n\geq 0\}$ in $\qasserts$ are $\qassert$-\emph{ranking assertions} for $\wstm$ if 
	\begin{enumerate}
		\item $\qassert \le \qassert_0$ and $\bigwedge_n \qassert_n = \emptydis_V$;
		\item for any $n\geq 0$ and $\qstate\in \qstatesh{qv(S)\cup V}$,
		\[
			\Exp(\sem{S}(\qstate|_b) \models \qassert_n ) \leq \Exp(\qstate \models \qassert_{n+1}).
		\]
	\end{enumerate}
\end{definition}
An alternative definition of $\qassert$-ranking assertions, which uses the weakest precondition semantics instead of the denotational one, is to replace the second clause above by $b\wedge wp.S.\qassert_n \le \qassert_{n+1}.$ It is easy to show that these two definitions are equivalent. 

With the notion of ranking assertions, we can state the proof rule for while loops in total correctness as follows:
$$\mathrm{(WhileT)}\quad \displaystyle\frac{
	\begin{tabular}{l}
	$\ass{b\wedge \qassert}{S}{\qassert}$\\
	$\qassert$-ranking assertions exist for $\wstm$
	\end{tabular}
}{\ass{\qassert}{\wstm}{\neg b\wedge \qassert}}$$
The {proof system for total correctness} is then defined as for partial correctness, except that the rule (While) is replaced by (WhileT), and rule (Abort) replaced by
$$\mathrm{(AbortT)}\qquad \ass{\bot_{V}}{\abort}{\bot_{V}}.$$
We write $\vdash_{{\mathit{tot}}}\ass{\qassert}{S}{\qassertp}$ if the correctness formula $\ass{\qassert}{S}{\qassertp}$ can be derived using the proof system for total correctness.

	Recall that in \cite{ying2012floyd}, a notion of bound function is proposed for proving total correctness of purely quantum programs. Let $M\in \ph$ and $\epsilon >0$. A function 
	$$t : \dh \rightarrow \N$$
	is called $(M, \epsilon)$-bound for the loop $\whilestm{\m[\bar{q}]=1}{S}$ where $\m = \{M_0, M_1\}$ if for any $\rho\in \dh$,
	\begin{enumerate}
		\item $t(\sem{S}(\e_{\bar{q}}^1(\rho))) \leq t(\rho)$, 
		\item if $\tr(M\rho) \geq \epsilon$ then  $t(\sem{S}(\e_{\bar{q}}^1(\rho))) < t(\rho).$ 
	\end{enumerate}
	With the bound functions, the proof rule for total correctness of quantum loops in~\cite{ying2012floyd} reads as follows:
	$$\displaystyle\frac{
		\begin{tabular}{l}
		$\ass{M}{S}{\e_{\bar{q}}^0(N)+ \e_{\bar{q}}^1(M)}$\\
		for each $\epsilon > 0$, $t_\epsilon$ is a $(\e_{\bar{q}}^1(M), \epsilon)$-bound function for $\whilestm{\m[\bar{q}]=1}{S}$
		\end{tabular}
	}{\ass{\e_{\bar{q}}^0(N)+ \e_{\bar{q}}^1(M)}{\whilestm{\m[\bar{q}]=1}{S}}{N}}$$
	As our ranking assertions are essentially \emph{linear} functions on $\dh$, they normally have a more compact representation, and hopefully are easier to use in applications than the bound functions in~\cite{ying2012floyd}.

Again, we can prove the soundness and (relative) completeness of the proof system for total correctness.

\begin{theorem}\label{thm:total}
	The proof system for total correctness is both sound and complete with respect to the total correctness of cq-programs.
\end{theorem}

\begin{proof}
	Soundness:  We need only to show that each rule of the proof system is valid in the sense of total correctness. Take rule (WhileT) as an example. 
	Let $\while \define \whilestm{q}{S}$,
	\begin{align}
		&\models_{{\mathit{tot}}} \ass{b\wedge \qassert}{S}{\qassert},\label{eq:1}
	\end{align}
	and $\{\qassert_n : n\geq 0\}$ be a sequence of $\qassert$-ranking assertions for $\while$. Assume without loss of generality $qv(S)\subseteq qv(\qassert)$. We prove by induction on $n$ that
	$$\qassert \le \qassert_n + \qassertp_{n}$$
	for any $n\geq 0$, where 
	$\qassertp_0 \define \emptydis_{qv(\qassert)}$, and for any $n\geq 0$,
	$\qassertp_{n+1} \define \neg b \wedge \qassert + b\wedge wp.S.\qassertp_n$. 
	The case when $n=0$ is from the assumption that $\qassert \le \qassert_0$. 
	For $n\geq 0$, we calculate
	\begin{align*}
		b\wedge \qassert \le wp.S.\qassert \le wp.S.\qassert_n + wp.S.\qassertp_{n}
	\end{align*}
	where the first inequality follows from Eq.(\ref{eq:1}), and the second one from the induction hypothesis and Lemma~\ref{lem:wpcorres}\ref{cl:wplinear}.
	Thus
	\begin{align*}
		\qassert &= b\wedge \qassert + \neg b\wedge \qassert\\
		& \le  b\wedge \qassert_{n+1} +  b\wedge wp.S.\qassertp_{n} + \neg b\wedge \qassert\\
		&\le \qassert_{n+1} + \qassertp_{n+1},
	\end{align*}
	where the first inequality follows from the definition of ranking assertions, and the second one from that of $\qassertp_{n+1}$.
	Thus 
	$$\qassert \le wp.(\wstm).(\neg b\wedge \qassert)$$
	by noting that $wp.(\wstm).(\neg b\wedge \qassert) = \bigvee_n \qassertp_n$ and
	$\bigwedge_n \qassert_n = \emptydis_\h$, and so 
	$$\models_{{\mathit{tot}}} \ass{\qassert}{\wstm}{\neg b\wedge \qassert}$$ as desired.
	
	Completeness: By the (Imp) rule, it suffices to show that for any $\qassert$ and $S'$ with $qv(S')\subseteq qv(\qassert)$,
	$$\vdash_{\mathit{tot}} \ass{wp.S'.\qassert}{S'}{\qassert}.$$
	Again, we take the case for while loops as an example. Let $\while \define \wstm$ and $\qassertp \define wp.\while.\qassert$.
	By induction, we have 
	$\vdash_{\mathit{tot}} \ass{wp.S.\qassertp}{S}{\qassertp}.$
	Note that
	$$\qassertp= \neg b\wedge \qassert+ b\wedge wp.S.\qassertp.$$
	Thus $b\wedge \qassertp = b\wedge wp.S.\qassertp \le wp.S.\qassertp$, and 
	so $\vdash_{\mathit{tot}} \ass{b\wedge \qassertp}{S}{\qassertp}$ by rule (Imp).

	Let $\qassert_0 = wp.\while.\top_{qv(\qassert)}$ and 
	$\qassert_{n+1}  = b\wedge wp.S.\qassert_n$.
	We are going to show that $\{\qassert_n : n\geq 0\}$ are $\qassert$-ranking assertions for $\while$. First, note that 
	$$\qassert_ 1 = b \wedge wp.S.\qassert_0 \le \neg b\wedge \top_{qv(\qassert)} + b \wedge wp.S.\qassert_0 = \qassert_0.$$
	So $\{\qassert_n : n\geq 0\}$ is decreasing by easy induction, using  Lemma~\ref{lem:wpcorres}\ref{cl:wpmono}. Next, as $\qassert \le \top_{qv(\qassert)}$, we have 
	$\qassertp\le \qassert_0$.  
	
	Finally, we prove that $\bigwedge_n \qassert_n =\bot_{qv(\qassert)}$.
	We show by induction on $n$ that for any $n\geq 0$ and $\qstate\in \qstatesh{qv(\qassert, \while)}$,
	\begin{equation}\label{eq:induc}
		\Exp(\qstate\models \qassert_n) =
		\tr(\sem{\while}(\qstate)) - \tr(\sem{\while^n}(\qstate)).
	\end{equation}
	The case when $n=0$ is direct from Lemmas~\ref{lem:bpdeg} and \ref{lem:wpwlp}.
	We further calculate that
	\begin{align*}
		\Exp(\qstate\models \qassert_{n+1}) 
		&=\Exp(\qstate\models b\wedge wp.S.\qassert_{n})\\	&=\Exp(\qstate|_b \models wp.S.\qassert_{n})\\
		&= \Exp(\sem{S}(\qstate|_b)\models \qassert_{n})\\
		&= \tr(\sem{\while}(\sem{S}(\qstate|_b))) - \tr(\sem{\while^n}(\sem{S}(\qstate|_b)))\\
		&=\tr(\sem{\while}(\qstate)) - \tr(\sem{\while^{n+1}}(\qstate)).
	\end{align*}
	Here the second last equality is from induction hypothesis, and the last one from Lemma~\ref{lem:whilesem}.
	Note that the second term of the r.h.s of Eq.(\ref{eq:induc}) converges to the first one when $n$ goes to infinity. Thus
	$\lim_n \Exp(\qstate\models\qassert_n) = 0$, and so $\bigwedge_n \qassert_n = \emptydis_{qv(\qassert)}$ from the arbitrariness of $\qstate$ and Lemma~\ref{lem:qasset}.
	Now using rule (WhileT), we have
	$\vdash_{\mathit{tot}} \ass{\qassertp}{\while}{\neg b\wedge \qassertp}$
	and the result follows from the fact that $\neg b\wedge \qassertp = \neg b\wedge \qassert \le \qassert$. 	
\end{proof}

To conclude this section, let us point out an alternative statement for the (WhileT) rule.
\begin{lemma} \label{lem:whileT}
	Let $\qassert\in \qassertsh{V}$. The loop $\wstm$ has  $\qassert$-ranking assertions iff there is an increasing sequence $\{\qassertp_n : n\geq 0\}$ of cq-assertions in $\qassertsh{V}$ such that
	\begin{enumerate}
		\item $\qassert \le \top_{V} - \qassertp_0$ and $\bigvee_n \qassertp_n = \top_{V}$;
		\item $\models_{\mathit{par}} \ass{b\wedge \qassertp_{n+1}}{S}{\qassertp_n}$.
	\end{enumerate}			
\end{lemma}
\begin{proof}
	Let $\{\qassert_n : n\geq 0\}$ be a sequence of $\qassert$-ranking assertions for $\wstm$. Let $\qassertp_i \define \top_{V} - \qassert_i$, $i\geq 0$. Then clause (1) holds trivially. To prove clause (2), note that
	\begin{align*}
		& b\wedge wp.S.\qassert_n \le \qassert_{n+1} \\
		\mbox{iff}\quad &	wp.S.\qassert_n \le b\wedge \qassert_{n+1} + \neg b \wedge \top_V\\
		\mbox{iff} \quad& b\wedge (\top_{V} - \qassert_{n+1}) \le \top_{V} - wp.S.\qassert_n \\
		\mbox{iff} \quad& b\wedge \qassertp_{n+1} \le wlp.S.\qassertp_n
	\end{align*}
	where the last equivalence is from Lemma~\ref{lem:wpcorres}\ref{cl:wpwlp}.
\end{proof}

With the above lemma, we can restate rule  $\mathrm{(WhileT)}$ as follows:
	\[\mathrm{(WhileT')}\quad \displaystyle\frac{
		\begin{tabular}{l}
		$\ass{b\wedge \qassert}{S}{\qassert}$\\
		$\{\qassert_n : n\geq 0\} \mbox{ increasing}, \qassert \le \top_{V} - \qassert_0, \ \bigvee_n \qassert_n = \top_{V}$\\		$\vdash_{\mathit{par}} \ass{b\wedge \qassert_{n+1}}{S}{\qassert_n}$\\
		\end{tabular}
	}{\ass{\qassert}{\wstm}{\neg b\wedge \qassert}}\]
	Interestingly, proof of partial correctness is also employed in this new rule for total correctness. Note that however, there are infinitely many premises in the rule which might not be convenient for automated reasoning, unless parametrised reasoning is supported somehow.

\section{Auxiliary Rules}\label{sec:aux}

{\renewcommand{\arraystretch}{2.9}
	\begin{table}[t]
		\begin{lrbox}{\tablebox}
			\centering
			\begin{tabular}{l}
			\begin{tabular}{lclc}
				(Top)  &\hspace{-3em} $\ass{\top_{V}}{S}{\top_V}$\hspace{4em}  (Bot)  \hspace{1em} $\ass{\bot_V}{S}{\bot_{V}}$   &
			(Init0)	& $\displaystyle\frac{\bar{q}\cap \qv(\qassert) = \emptyset}{\ass{\qassert}{\bar{q}:=0}{\qassert}}$\\
						(Meas0)	&
			$\displaystyle\frac{\bar{q}\cap \qv(\qassert) = \emptyset, \m = \{M_i : i\in I\}}{\ass{\sum_{i\in I}\qassert[\subs{x}{i}]\otimes M_i^\dag M_i}{x:=\measure\ \m[\bar{q}]}{\qassert}}$ \hspace{3em}
		&
			(Unit0)	&
			$\displaystyle\frac{\bar{q}\cap \qv(\qassert) = \emptyset}{\ass{\qassert}{\bar{q}\apply U}{\qassert}}$
			\end{tabular}	\\	
					\begin{tabular}{lc}
			(Param)& $\left\{\displaystyle\sum_{i_1, \cdots, i_k= 1, \mathit{distinct}}^{|\bar{q}|} \sum_{i=1}^K \bigwedge_{j=1}^k (e_j = i_j)\wedge (e=i) \wedge \mathcal{U}^i_{i_1, \cdots, i_k}(\qassert)\right\}{\bar{q}[e_1, \cdots, e_k]\apply \u(e)}\left\{\qassert\right\}$\\
			& where $\u \define \{U_i : 1\leq i\leq K\}$, $\mathcal{U}^i_{i_1, \cdots, i_k}(\qassert) \define
			U_i^\dag \qassert {U_i}
			$,
			and $U_i$ is applied on ${q}_{i_1}, \cdots, {q}_{i_k}$.
		\end{tabular}\\
		\begin{tabular}{lc}
							(SupOper)& $\displaystyle\frac{
				\ass{\qassert}{S}{\qassertp}, W\cap qv(\qassert) \subseteq V \subseteq qv(\qassert), W\cap qv(\qassertp) \subseteq V \subseteq qv(\qassertp),  (V\cup W) \cap qv(S) = \emptyset}{\ass{\f_{V\ra W}(\qassert)}{S}{\f_{V\ra W}(\qassertp)}}$\\
			& where $\f_{V\ra W}$ is a completely positive and sub-unital super-operator from $\l(\h_V)$ to $\l(\h_W)$. \\
						(SupPos)	& $\displaystyle\frac{
				\ass{\left\<\cassert, \displaystyle\frac{1}{\sqrt{d}}\sum_{i=1}^d|\phi_i\>_V |i\>_{W}\right\>}{S}{\left\<\cassert', \displaystyle\frac{1}{\sqrt{d}}\sum_{i=1}^d|\psi_i\>_V |i\>_W\right\>}, qv(S)\subseteq V
			}{\ass{\left\<\cassert, \displaystyle\sum_{i=1}^d\alpha_i |\phi_i\>_V\right\>}{S}{\left\<\cassert',  \displaystyle\sum_{i=1}^d\alpha_i |\psi_i\>_V\right\>}}$\\
			& where $|i\>$'s,  $|\phi_i\>$'s, and $|\psi_i\>$'s are all sets of  orthonormal states, $\alpha_i \in \C$, and $\sum_{i=1}^d |\alpha_i|^2 =1$.\\
				 (L-Sum) & $\displaystyle\frac{
				\ass{\sum_{i=1}^d \qassert_i
					\otimes |i\>_W\<i|}{S}{\sum_{i=1}^d \qassertp_i
					\otimes |i\>_W\<i|},\ W \cap qv(S, \qassert_i, \qassertp_i) = \emptyset}
			{\ass{\sum_{i=1}^d \lambda_i \qassert_i}{S}{\sum_{i=1}^d \lambda_i \qassertp_i}}$\\
			& where $|i\>$'s are orthonormal states in $\h_W$, 
			$\lambda_i\geq 0$, and $\sum_{i=1}^d \lambda_i \leq 1$.
			\end{tabular}\\
		\begin{tabular}{lclc}
					(Tens)& $\displaystyle\frac{
		\ass{\qassert}{S}{\qassertp},\ \ W \cap qv(S, \qassert, \qassertp) = \emptyset}{\ass{M_W\otimes \qassert}{S}{M_W\otimes \qassertp}}$ &
					(Trace)& $\displaystyle\frac{
	\ass{\qassert}{S}{\qassertp}, V \subseteq qv(\qassert)\cap qv(\qassertp), V\cap qv(S)= \emptyset}{\ass{\frac{1}{\dim(\h_V)}\tr_V( \qassert)}{S}{\frac{1}{\dim(\h_V)}\tr_V(\qassertp)}}$ \\
				(Exist) & $\displaystyle\frac{
			\ass{\<\cassert, M\>}{S}{\qassertp},\  x\not\in var(S) \cup \mathit{free}(\qassertp)}
		{\ass{\<\exists x.\cassert, M\>}{S}{\qassertp}}$ &
		 (Inv) & $\displaystyle\frac{
			\ass{\qassert}{S}{\qassertp}, \mathit{free}(\cassert)\cap \mathit{change}(S) = \emptyset}
		{\ass{\cassert\wedge \qassert}{S}{\cassert\wedge \qassertp}}$\\
		 (Disj) & $\displaystyle\frac{
			\ass{\<\cassert, M\>}{S}{\qassertp},\ \ass{\<\cassert', M\>}{S}{\qassertp}}
		{\ass{\<\cassert\vee \cassert', M\>}{S}{\qassertp}}$ &
				(Sum) & $\displaystyle\frac{
					\ass{\<\cassert, M\>}{S}{\qassertp},\ \ass{\<\cassert', N\>}{S}{\qassertp},\ \cassert' \rightarrow \neg \cassert}{\ass{\<\cassert, M\> + \<\cassert', N\>}{S}{\qassertp}}$\\
				 (Linear)& $\displaystyle\frac{
					\ass{\qassert_i}{S}{\qassertp_i},\ \lambda_i\geq 0}
				{\ass{\sum_i \lambda_i \qassert_i}{S}{\sum_i \lambda_i \qassertp_i}}$ & (ProbComp) & $\displaystyle\frac{
					\ass{\cassert'}{S_1}{\left\<\cassert, |\psi\>_{\bar{q}}\<\psi|\right\>}, \ass{\left\<\cassert, M_{\bar{q}}\right\>}{S_2}{\qassertp}
				}{\ass{\left\<\psi|M|\psi\right\> \cdot \cassert'}{S_1; S_2}{\qassertp}}$			 \\
				\end{tabular}\\
			\begin{tabular}{lc}
			
				(C-WhileT) & $\displaystyle\frac{
				\ass{b\wedge \qassert}{S}{\qassert},\
				\ass{b\wedge\cassert \wedge t=z}{S}{t<z},\ \cassert \rightarrow t\geq 0
			}{\ass{\qassert}{\wstm}{\neg b\wedge \qassert}}$\\
		& where $\mathit{type}(z) = \mathit{type}(t) = \tyint$, $z\not\in var(\cassert, b, t, S)$, $\qassert = \bigoplus_{i\in I} \<\cassert_i, M_i\>$ and $\cassert \define \bigvee_{i\in I} \cassert_i$.
			\end{tabular}		
			\end{tabular}
		\end{lrbox}
		\resizebox{\textwidth}{!}{\usebox{\tablebox}}
		\caption{Auxiliary rules.}
		\label{tbl:auxrules}
	\end{table}
}
We have provided sound and relatively complete proof systems for both partial and total correctness of cq-programs. Thus in principle, these proof rules are sufficient for proving desired properties as long as they can be described faithfully with Hoare triple formulas.
However, in practice, using these rules directly might be complicated. To simplify reasoning, in this section we introduce some auxiliary proof rules which are listed in Table~\ref{tbl:auxrules}. For the sake of convenience, we write $\<\cassert, |\psi\>\>$ for $\<\cassert, |\psi\>\<\psi|\>$, and $\cassert$ for $\cassert \wedge \top$.

The rules (Top) and (Bot) deals with special cq-assertions. Rules (Init0), (Meas0), and (Unit0) simplify the corresponding ones in Table~\ref{tbl:psystem} when the evolved quantum variables do not appear in the postcondition. Extended commands with syntactic sugars are also considered in these rules, as well as in the rule (Param). 

Rule (SupOper) essentially says that any valid operation applied on the quantum variables not involved in $S$ does not affect the correctness of $S$. 
Note that a weaker version of this rule, where $V$ and $W$ are taken equal, was presented in~\cite{ying2019toward}. However, the current version is much more expressive, evidenced by the fact that  (SupPos), (L-Sum), (Tens), and (Trace) are all its special cases. 

Rule (SupPos) deals with superposition of quantum states, and it is useful in proving the correctness of quantum circuits which consist of solely unitary operators. As unitary operators are linear, a natural question is: can we verify such circuits by only checking each pure state from an orthonormal basis? Specifically, let $V=qv(S)$, and 
$\{|\phi_i\>: 1\leq i\leq d\}$ and $\{|\psi_i\>: 1\leq i\leq d\}$ are both orthonormal bases of $\h_V$. If $\vdash\ass{\left\<\cassert, |\phi_i\>_V\right\>}{S}{\left\<\cassert',  |\psi_i\>_V\right\>}$ for all 
$i$, can we deduce
$$\vdash\left\{\left\<\cassert, \sum_{i=1}^d\alpha_i |\phi_i\>_V\right\>\right\}{S}\left\{\left\<\cassert', \sum_{i=1}^d\alpha_i |\psi_i\>_V\right\>\right\}$$ for any superposed states $\sum_{i=1}^d\alpha_i |\phi_i\>$ and $\sum_{i=1}^d\alpha_i |\psi_i\>$?
This is, however, not correct. For example, let $\type(q) = \tyqubit$. Then 
\[
\ass{\left\<\true, |0\>_q\right\>}{q \apply Z}{\left\<\true,  |0\>_q\right\>}\quad  \mbox{ and } \quad\ass{\left\<\true, |1\>_q\right\>}{q \apply Z}{\left\<\true,  |1\>_q\right\>}\] since $Z|1\> = -|1\>$ and $(-|1\>)(-\<1|) = |1\>\<1|$. However, $\ass{\left\<\true, |+\>_q\right\>}{q \apply Z}{\left\<\true,  |+\>_q\right\>}$ is certainly not true. 
The reason is that observables (thus cq-assertions) cannot distinguish quantum states like $|1\>$ and $-|1\>$ which differ only in the global phases. 
To overcome this difficulty, in rule (SupPos) we combine all the states $|\phi_i\>_{V}$ into a single (entangled) one $\frac{1}{\sqrt{d}}\sum_{i=1}^d |\phi_i\>_V|i\>_W$ in a larger Hilbert space $\h_V \otimes  \h_W$ (Intuitively, we use the orthonormal states $|i\>$ in $\h_W$ to index $|\phi_i\>_V$). In this way, the global phases caused by applying $S$ on $|\phi_i\>_V$'s become local and detectable.

Rules (Exist) and (Inv) are merely classical ones where the logic operations are performed on the classical part of the cq-assertions. The three rules (Disj), (Sum), and (Linear) all extend the rule 
$$\displaystyle\frac{
	\ass{\cassert_1}{S}{\cassert_1'}, \ \ass{\cassert_2}{S}{\cassert_2'}}
{\ass{\cassert_1\vee \cassert_2}{S}{\cassert_1'\vee \cassert_2'}}$$
in classical Hoare logic dealing with disjunction of assertions. In the first two rules, the disjunction is applied only on the classical part: rule (Disj) allows disjunction of any classical assertions $p$ and $p'$, but their quantum part must be the same; rule (Sum) allows different quantum parts, but the  classical assertions must be mutually exclusive. For the general case, a weighted sum (for both the pre- and the postconditions) is used in (Linear). 

Rule (ProbComp) reasons about sequential composition of two programs $S_1$ and $S_2$. Note that rule (Seq) in Table~\ref{tbl:psystem} assumes the postcondition of $S_1$ is the same as, or stronger than, when (Imp) is employed, the precondition of $S_2$. In contrast, rule (ProbComp) can handle the case where such an assumption does not hold. As can be seen from the case studies, this rule is very useful in calculating the success probability of quantum algorithms. 

Finally, we present rule (C-WhileT) for the special case when a classical ranking function can be found to guarantee the (finite) termination of cq-programs. As shown in the case studies in Sec.~\ref{sec:case}, this rule is useful in simplifying the analysis of many practical quantum algorithms.

\begin{theorem} \label{thm:aux}
	\begin{enumerate}
		\item 	All the auxiliary rules presented in Table~\ref{tbl:auxrules}, except (Top), are sound with respect to total correctness.
		\item 	If we require $\sum_{i} \lambda_i \leq 1$ in \textrm{(Linear)}, then all the auxiliary rules presented in Table~\ref{tbl:auxrules}, except (ProbComp), (SupPos) and (C-WhileT), are sound with respect to partial correctness. 
	\end{enumerate}
\end{theorem}
\begin{proof}
	The rules (Top) and (Bot) are from Lemma~\ref{lem:corf}.
	We note from (Meas) that whenever $\bar{q}\cap qv(\qassert) = \emptyset$,
	\[
	 \left\{\sum_{i\in I}\qassert[\subs{x}{i}]\otimes M_i^\dag M_i \right\}{x:=\measure\ \m[\bar{q}]}\{\qassert\otimes I_{\bar{q}}\}.
	\]
	Then (Meas0) follows by (Imp). The proofs for 
	(Init0) and (Unit0) are similar. (Param) follows from Lemma~\ref{lem:wpsugar}(4).
	
	(SupOper): From $\models_{\mathit{tot}} \ass{\qassert}{S}{\qassertp}$, 
	we have $\qassert \lesssim wp.S.\qassertp$ by Lemma~\ref{lem:lesssimpre}.  Then $$\models_{\mathit{tot}} \ass{\f_{V\ra W}(\qassert)}{S}{\f_{V\ra W}(\qassertp)}$$ from Lemma~\ref{lem:wpcorres}\ref{cl:superoper}. The case for partial correctness is similar. For (SupPos), we first let $\f$ be defined as
	\[
	\f_{W\ra \emptyset}(\qassert') = \<\psi^*| \qassert' |\psi^*\>
	\]
	for any $\qassert' \in \qassertsh{W}$, where 
	$|\psi^*\> =\sum_{i=1}^d \alpha_i^*|i\>_W$.
	Then we have from  (SupOper) that
	\[\models_{\mathit{tot}} \ass{\left\<\cassert, \displaystyle\frac{1}{\sqrt{d}}\sum_{i=1}^d|\phi_i\>_V |i\>_{W}\right\>}{S}{\left\<\cassert', \displaystyle\frac{1}{\sqrt{d}}\sum_{i=1}^d|\psi_i\>_V |i\>_W\right\>}\] implies
	\[
	\models_{\mathit{tot}} \ass{\left\<\cassert, \displaystyle\frac{1}{\sqrt{d}}\sum_{i=1}^d\alpha_i |\phi_i\>_V\right\>}{S}{\left\<\cassert',  \displaystyle\frac{1}{\sqrt{d}}\sum_{i=1}^d\alpha_i |\psi_i\>_V\right\>}.
	\]
	The desired result follows from (Linear) by multiplying both pre- and post-conditions with $d$. Similarly, (L-Sum) follows from (SupOper) by taking 
	$
	\f_{W\ra \emptyset}(\qassert') = \sum_{i=1}^d \lambda_i \<i| \qassert' |i\>
	$ for any $\qassert' \in \qassertsh{W}$.
	
	(Tens) follows from (SupOper) by taking $V= \emptyset$ and 
	$\f_{V\ra W}(1) = M$. Conversely, in (Trace) we take $W= \emptyset$ and 
	$$\f_{V\ra W}(\qassert) = \frac{1}{\dim(\h_V)}\sum_{i\in I} \<i|_V \qassert |i\>_V$$ where $\{|i\> : i\in I\}$ is an orthonormal basis of $\h_V$.

	The rules (Exist), (Inv), (Disj), (Sum), and (Linear) are all easy from definition. Note that to prove (Linear) for partial correctness, we have to require $\sum_{i} \lambda_i \leq 1$.
	
	(ProbComp): For any $\cstate$ and $\rho$ with $\cstate \models \cassert'$ and $\tr(\rho) =1$, let $\qstate' \define \sem{S_1}(\cstate, \rho)$ and $\qstate'' \define \sem{S_2}(\qstate')$. We first have from $\models_{\mathit{tot}}\ass{\cassert'}{S_1}{\<p, |\psi\>_{\bar{q}}\<\psi|\>}$ that 
	\[
	1=\Exp(\<\cstate, \rho\> \models \cassert') \leq \sum_{\cstate'\in \supp{\qstate'}, \cstate'\models \cassert} \<\psi|\qstate'(\cstate')|\psi\> \leq \tr(\qstate') \leq 1.
	\]
	Thus for any $\cstate'\in \supp{\qstate'}$,
	$\cstate'\models p$ and $\qstate'(\cstate') = c_{\cstate'}|\psi\>\<\psi|$ for some $c_{\cstate'} \geq 0$ with $\sum_{\cstate'\in \supp{\qstate'}} c_{\cstate'} = 1$. Furthermore, by
	$\models_{\mathit{tot}}\ass{\<\cassert, M_{\bar{q}}\>}{S_2}{\qassertp}$, we have  
	$$\sum_{\cstate'\in \supp{\qstate'}} c_{\cstate'}\<\psi| M|\psi\>=\Exp(\qstate'\models \<\cassert, M_{\bar{q}}\>) \leq \Exp(\qstate''\models \qassertp).$$
	The result then follows from the observation that
	\[
	\Exp(\<\cstate, \rho\> \models \<\psi| M|\psi\>\cdot \cassert')= \<\psi| M|\psi\>.
	\]
	
	(C-WhileT): first note that $\models_{\mathit{tot}}\ass{b\wedge\cassert \wedge t=z}{S}{t<z}$ implies for any $\cstate \models b\wedge \cassert\wedge t=z$,
	and any $\cstate'$ in the support of $\sem{S}(\cstate, \rho)$, we have $\cstate'\models t<z$. Then an argument similar to that for classical programs leads to the conclusion that all computations from $\<\wstm, \cstate, \rho\>$ terminates within $\cstate(t)$ steps, provided that $\cstate \models \cassert$. 
\end{proof}

\section{Case studies}\label{sec:case}

To illustrate the effectiveness of the proof systems proposed in the previous sections, we
employ them to verify Grover's search algorithm presented in Example~\ref{ex:grover} and Shor's factorisation algorithm with its subroutines.

\subsection{Grover's search algorithm}
	We have proved in Examples~\ref{ex:grover} and~\ref{ex:groverform}, by employing the denotational semantics and the definition of correctness formulas respectively, that Grover's algorithm succeeds in finding a desired solution with probability $p_{\mathit{succ}}$ shown in  Eq.(\ref{eq:psucc}). We now re-prove this result using the proof rules for total correctness. As stated in Example~\ref{ex:groverform}, the goal is to show
	\begin{equation}\label{eq:grospe}
		\vdash_{\mathit{tot}} \ass{p_{\mathit{succ}}}{\mathit{Grover}}{y\in Sol}.
	\end{equation}
	
	Let $b\define x<K$ and $\qassert  \define	\sum_{k=0}^K \<x=k,  \Psi_{K-k}\>$, where $\Psi_k = |\psi_k\>\<\psi_k|$ and
	$$|\psi_k\> = \cos\left(\frac \pi 2 - k\theta\right) |\alpha\> + \sin\left(\frac \pi 2 - k\theta\right) |\beta\>.$$
	Note that $|\psi_0\> = |\beta\>$ and $G|\psi_k\> = |\psi_{k-1}\>$. Intuitively, $\qassert$ records the quantum states at each iteration. We show that it serves as an invariant of the while loop in Grover's algorithm.
  Observe from (Unit) and (Assn) that
	\begin{align*}
	\vdash_{\mathit{tot}} \left\{\sum_{k=0}^{K-1} \<x=k,  \Psi_{K-k}\>\right\}  \bar{q} \apply G;\ x:= x+1; \left\{\sum_{k=0}^{K-1} \<x=k+1,  \Psi_{K-k-1}\>\right\} 
	\end{align*} 
	Together with the fact 
	\[
	\sum_{k=0}^{K-1} \left\<x=k+1,  \Psi_{K-k-1}\right\> = \sum_{k=1}^{K} \left\<x=k,  \Psi_{K-k}\right\> \le  \qassert,
	\]
	we deduce from rule (Imp) that
	$\vdash_{\mathit{tot}} \ass{b\wedge \qassert}{ \bar{q} \apply G; x:= x+1;}{\qassert}.$
	 Let $z\not \in \{x,y\}$, $t \define K -x$, and $\cassert \define (0\leq x\leq K)$. Then $\cassert  \rightarrow t\geq 0$, and $t$ serves as a classical ranking function.  
	Thus by rule (C-WhileT),
	\begin{equation}\label{eq:grwhile}
	\vdash_{\mathit{tot}} \ass{\qassert}{ \whilestm{b}{\bar{q} \apply G; x:= x+1;}}{\neg b\wedge \qassert}.
	\end{equation}
	Furthermore, we have
	\begin{align*}
	&\left\{\left|\<+|^{\otimes n}|\psi_K\>\right|^2 \right\}\\ 
	&\bar{q} := 0; \ \bar{q} \apply H^{\otimes n};\\
	&\left\{\left\<\true, \Psi_K\right\>\right\}& \mathit{(Init, Unit)}\\ 
	&x := 0;\\
	& \left\{\sum_{k=0}^K \left\<x=k,  \Psi_{K-k}\right\>\right\}& \mathit{(Assn, Imp)}\\
	& \whilestm{b}{\bar{q} \apply G; x:= x+1;}\\
	&\left\{\left\<x=K,  \Psi_{0}\right\>\right\} &Eq.(\ref{eq:grwhile})\\
	&\left\{\left\<\true,  \sum_{i\in Sol} |i\>\<i|\right\>\right\} &\mathit{(Imp)}\\
	&y:= \measure\ \bar{q}\\
	&\{y\in Sol\} &\mathit{(Meas0)}
	\end{align*} 
	Finally, it is easy to show that $|\<+|^{\otimes n}|\psi_K\>|^2 = p_{\mathit{succ}}$, from which Eq.(\ref{eq:grospe}) follows.

\subsection{Quantum Fourier Transform}
In the rest of the paper, all quantum variables are assumed to have $\tyqubit$ type.
Recall that the $n$-qubit quantum Fourier transform (QFT) is a unitary mapping such that for any integer $j$, $0\leq j\leq 2^n-1$, 
$$|j\> \rightarrow |\psi_j\> \define \frac{1}{\sqrt{2^n}} \sum_{k=0}^{2^n-1} e^{2\pi ijk/2^n} |k\> = \bigotimes_{k=n}^1  |+_{0.j_k\cdots j_n}\>$$
where $j_1\ldots j_n$ is the binary representation of $j$, and $|+_{0.j_k\cdots j_n}\> \define (|0\> + e^{2\pi i0.j_k\cdots j_n}|1\>/\sqrt{2}$. In particular, $|+_0\> = |+\>$. 
QFT serves as an important part for Shor's factorisation and many other quantum algorithms.

The QFT algorithm for $n$ qubits can be described in our cq-language (with syntactic sugars) as follows:
\begin{align*}
	\mathit{QFT}(n)\define&\\
	&x := 1;\\
	&\while\ x \leq n\ \ddo\\
	&\quad \bar{q}[x] \apply H;\  y := x + 1;\\
	&\quad \while\ y\leq n\ \ddo\\
	&\qquad \bar{q}[y,x] \apply \mathit{CR}(y-x+1);\ y := y+1;\\
	&\quad \pend\\
	&\quad x := x + 1;\\
	&\pend\\
	&\bar{q} \apply \mathit{SWAP}_n
\end{align*}
where $CR\define \{CR_k : 1\leq k\leq n\}$ and for each $k$, $CR_k$ is the controlled-$R_k$ operator with $$R_k = |0\>\<0| + e^{2\pi i/2^k} |1\>\<1|,$$
and $\mathit{SWAP}_n$ reverses the order of a list of $n$ qubits; that is, $\mathit{SWAP}_n|i_1, \cdots, i_n\>_{\bar{q}} = |i_n,\cdots, i_1\>_{\bar{q}}$ for all $|i_j\>\in \h_{q_j}$.
The correctness of $	\mathit{QFT}(n)$ is stated as follows: for any $\alpha_j\in \C$, $\sum_{j} |\alpha_j|^2 =1$,
\[
\vdash_{\mathit{tot}} \left\{\left\<\true, \sum_j \alpha_j |j\>_{\bar{q}}\right\>\right\} \mathit{QFT}(n) \left\{\left\<\true, \sum_j \alpha_j|\psi_j\>_{\bar{q}}\right\>\right\}.
\]	
With the help of rule (SupPos), it suffices to prove
\[
\vdash_{\mathit{tot}} \left\{\left\<\true, |\alpha\>_{\bar{q}, \bar{q}'}\right\>\right\} \mathit{QFT(n)} \left\{\left\<\true, |\beta\>_{\bar{q}, \bar{q}'}\right\>\right\}
\]
where $|\bar{q}'| = |\bar{q}|$, $|\alpha\> \define \frac{1}{\sqrt{2^n}}\sum_{j=0}^{2^n-1} |j\>|j\>$ is a maximally entangled state in $\h_{\bar{q}} \otimes \h_{\bar{q}'} $, and $|\beta\> \define \frac{1}{\sqrt{2^n}}\sum_{j=0}^{2^n-1} |\psi_j\>|j\>$.

The proof is rather involved. Due to the limit of space, we sketch the main ideas instead.

\begin{enumerate}
	\item Let $\while'$ be the inner loop. We show that 
	$$\qassertp \define \sum_{\ell = 1}^{n} \sum_{m=\ell+1}^{n+1} \left\<x=\ell \wedge y=m,
	\frac{1}{\sqrt{2^n}} \sum_{j=0}^{2^n-1}	\left[\bigotimes_{k=1}^{\ell-1} |+_{0.j_k\cdots j_n}\> \otimes |+_{0.j_l\cdots j_{m-1}}\> \bigotimes_{k=\ell + 1}^n |j_k\>\right]_{\bar{q}} |j\>_{\bar{q}'}\right\>$$ serves as an invariant for $\while'$.
	Furthermore, let $p \define (1\leq x\leq n) \wedge (x+1\leq y\leq n+1)$. Then $t\define n+1-y$ serves as a classical ranking function for $\while'$. 
	Thus we have from (C-WhileT) 
	\[
	\vdash_{\mathit{tot}} \ass{\qassertp}{\while'}{y>n \wedge \qassertp}.
	\]
	\item  Let $\while$ be the outer while-loop, and 
	\[
	\qassert \define \sum_{\ell = 1}^{n+1}  \left\<x=\ell,\
	\frac{1}{\sqrt{2^n}}\sum_{j=0}^{2^n-1}	\left[\bigotimes_{k=1}^{\ell-1} |+_{0.j_k\cdots j_n}\> \otimes \bigotimes_{k=\ell}^n |j_k\>\right]_{\bar{q}} |j\>_{\bar{q}'}\right\>.
	\]
	Then it can be shown that $\qassert$ is an invariant for $\while$.
	Again, it is easy to construct a classical ranking function ($t\define n+1-x$), so
	\begin{equation}\label{eq:qft}
		\vdash_{\mathit{tot}} \ass{\qassert}{\while}{x>n \wedge \qassert}.
	\end{equation}
	
	\item For the whole program, we have
	\begin{align*}
		&\left\{\left\<\true, |\alpha\>_{\bar{q},\bar{q}'}\right\>\right\}\\
		&x := 1;\\
		&\left\{\sum_{\ell = 1}^{n+1}  \left\<x=\ell,\
		\frac{1}{\sqrt{2^n}}\sum_{j=0}^{2^n-1}	\left[\bigotimes_{k=1}^{\ell-1} |+_{0.j_k\cdots j_n}\> \otimes \bigotimes_{k=\ell}^n |j_k\>\right]_{\bar{q}} |j\>_{\bar{q}'}\right\> \right\} &\mathit{(Assn)}\\
		&\while\\
		&\left\{\left\<x=n+1,\
		\frac{1}{\sqrt{2^n}}\sum_{j=0}^{2^n-1}	\left[\bigotimes_{k=1}^{n} |+_{0.j_k\cdots j_n}\>\right]_{\bar{q}} |j\>_{\bar{q}'}\right\> \right\} & \mathit{Eq.(\ref{eq:qft})}\\
		&\bar{q} \apply \mathit{SWAP}_n\\
		&\left\{ \left\<\true, \frac{1}{\sqrt{2^n}}\sum_{j=0}^{2^n-1} |\psi_j\>_{\bar{q}} |j\>_{\bar{q}'}\right\>\right\}. & \mathit{(Unit, Imp)}
	\end{align*}		
\end{enumerate}
\subsection{Phase Estimation}

Given (the controlled version of) a unitary operator $U$ acting on $m$ qubits and one of its eigenstate $|u\>$ with $U|u\> = e^{2\pi i\varphi} |u\>$ for some $\varphi \in [0,1)$. The phase estimation algorithm computes an $n$-bit approximation $\tilde{\varphi}$ of $\varphi$  with success probability at least $1-\epsilon$, where $n$ and $\epsilon$ are two given parameters.
Let $t\define n+ \ceil{\log(2+ \frac{1}{2\epsilon})}$. The algorithm is detailed as follows:
\begin{align*}
	\mathit{PE}\define &                                               \\
	&\bar{r} := 0;  \ \bar{r} \apply \mathit{U_u};\ \bar{q} := 0; \ x := 1;                                       \\
	& \while\ x \leq t\ \ddo                        \\
	& \quad \bar{q}[x] \apply H;  \ y := 0;                                 \\
	& \quad \while\ y< 2^{t-x}\ \ddo                \\
	& \qquad \bar{q}[x],\bar{r} \apply \mathit{CU};\ y := y+1;                              \\
	& \quad \pend                                   \\
	& \quad x := x + 1;                             \\
	& \pend                                         \\
	& \bar{q} \apply \mathit{QFT}(t)^\dag;             \\
	& z := \measure\ \bar{q}
\end{align*}
where $|\bar{q}| = t$, $|\bar{r}| = m$, $\mathit{U_u}$ is a unitary operator to prepare $|u\>$ from $|0\>$, $\mathit{CU}$ is the controlled-$U$ operator, and $\mathit{QFT}(t)^\dag$ is the inverse quantum Fourier transform on $t$ qubits.

The correctness of $\mathit{PE}$ can be stated as
\begin{equation}\label{eq:pecor}
\vdash_{\mathit{tot}}	\ass{p_{\mathit{PE}}}{\mathit{PE}}{|\varphi - z/2^t| < 2^{-n}}
\end{equation}
with $p_{\mathit{PE}}\geq 1-\epsilon$.
Let $\while$ be the outer while-loop and $\while'$ be the inner one.
The proof consists of three phases.
\begin{enumerate}
	\item For the body of $\while'$, we have for any $1\leq k\leq t$ and $0\leq \ell< 2^{t-k}$,
	\begin{align*}
		&\left\{\left\<x=k\wedge y=\ell, |+_{\ell \varphi}\>_{{q}_k}|u\>_{\bar{r}}\right\>\right\}\\
		& \bar{q}[x],r \apply \mathit{CU};\ y := y+1;\\
		&\left\{\left\<x=k\wedge y=\ell +1, |+_{(\ell+1) \varphi}\>_{{q}_k}|u\>_{\bar{r}}\right\>\right\}& \mathit{(Unit, Assn)}
	\end{align*}
	where $|+_a\> \define (|0\> + e^{2\pi i a}|1\>)/\sqrt{2}$ for any $a\in \R$,
	and in particular, $|+_0\> = |+\>$. Furthermore, it is easy to construct a classical ranking function $2^{t-x}-y$. 
	Thus we have from (Linear) and (C-WhileT),
	\begin{equation}\label{eq:peinner}
		\vdash_{\mathit{tot}} \ass{\qassertp}{\while'}{\qassertp \wedge y\geq 2^{t-x}}
	\end{equation}
	where 
	$
	\qassertp \define \sum_{\ell=0}^{2^{t-k}} \left\<x=k\wedge y=\ell, |+_{\ell \varphi}\>_{{q}_k}|u\>_{\bar{r}}\right\>.
	$
	
	\item For the body of $\while$, we have for any $1\leq k\leq t$,
	\begin{align*}
		&\left\{\left\<x=k, |0\>_{{q}_k}|u\>_{\bar{r}}\right\>\right\}\\
		& \bar{q}[x] \apply H;\ y := 0;\\
		&\left\{\qassertp \equiv \sum_{\ell=0}^{2^{t-k}}\left\<x=k \wedge y=\ell, |+_{\ell\varphi}\>_{{q}_k}|u\>_{\bar{r}}\right\>\right\} & \mathit{(Unit, Assn)}\\
		&\while'\\
		&\left\{\qassertp \wedge y\geq 2^{t-x} \equiv \left\<x=k, |+_{2^{t-k}\varphi}\>_{{q}_k}|u\>_{\bar{r}}\right\>\right\} &\mathit{Eq.(\ref{eq:peinner})}\\
		&x := x + 1;\\
		&\left\{\left\<x=k+1, |+_{2^{t-k}\varphi}\>_{{q}_k}|u\>_{\bar{r}}\right\>\right\} & \mathit{(Assn)}
	\end{align*}
	Furthermore, it is easy to construct a classical ranking function $t+1-x$. 
	Thus from (Tens), (Linear), and (C-WhileT) we have
	\begin{equation}\label{eq:peout}
		\vdash_{\mathit{tot}} \ass{\qassert}{\while}{\qassert \wedge x>t} 
	\end{equation}
	where 
	\[
	\qassert \define \sum_{k=1}^{t+1} \left\<x=k, \bigotimes_{j=1}^{k-1}|+_{2^{t-j}\varphi}\>\otimes |0\>^{\otimes (t-k+1)}|u\>_{\bar{r}}\right\>.
	\]
	
	\item For the whole program, we have
	\begin{align*}
		& \left\{\top\right\}\\
		&\bar{r} := 0; \bar{r} \apply \mathit{U_u};\ \bar{q} := 0;\\
		& \left\{\left\<\true, |0\>^{\otimes t}|u\>_{\bar{r}}\right\>\right\} & \mathit{(Init, Unit)}\\
		&x := 1;\\
		& \left\{\sum_{k=1}^{t+1} \left\<x=k, \bigotimes_{j=1}^{k-1}|+_{2^{t-j}\varphi}\>\otimes |0\>^{\otimes (t-k+1)}|u\>_{\bar{r}}\right\>\right\} & \mathit{(Assn)}\\
		&\while\\
		& \left\{\left\<x=t+1, \bigotimes_{j=1}^{t}|+_{2^{t-j}\varphi}\>|u\>_{\bar{r}}\right\>\right\} & \mathit{Eq.(\ref{eq:peout})}\\
		& \left\{\left\<\true, \frac{1}{\sqrt{2^t}}\sum_{k=0}^{2^t-1} e^{2\pi ik \varphi} |k\>\right\>\right\} & \mathit{(Imp)}
	\end{align*}

	Furthermore, let 		
	\begin{equation}\label{eq:defK}
	K\define \left\{0\leq m < 2^t :  \left|\varphi - 2^{-t}m\right| < 2^{-n}\right\}
	\end{equation}
	and for each $m$,
	$|\psi_m\> \define 1/\sqrt{2^t}\sum_{j=0}^{2^t-1} e^{2\pi ijm/2^t} |j\>.$
	Then we have
	\begin{align*}
		& \left\{\left\<\true, \sum_{m\in K}|\psi_m\>_{\bar{q}}\<\psi_m|\right\>\right\}\\
		&\bar{q} \apply \mathit{QFT}(t)^\dag;\\
		& \left\{\left\<\true, \sum_{m\in K}|m\>_{\bar{q}}\<m| \right\>\right\} &\mathit{(Unit)}\\
		&z := \measure\ \bar{q}\\
		& \left\{\left|\varphi - z/2^t\right| < 2^{-n}\right\} & \mathit{(Meas0)}
	\end{align*}
	Finally, by (ProbComp) we have Eq.(\ref{eq:pecor}) with 
	\[
	p_{\mathit{PE}} = \sum_{m\in K}\left|\frac{1}{\sqrt{2^t}}\sum_{k=0}^{2^t-1} e^{2\pi ik \varphi} \left\<\psi_m|k\right\>\right|^2 \geq 1-\epsilon,
	\]
	where the last inequality is from the following lemma.
\end{enumerate}

\begin{lemma}\label{lem:peprob}
	Let $\varphi \in [0,1)$, $\epsilon \in (0,1)$, $n\geq 1$, $t\define n+ \ceil{\log(2+ \frac{1}{2\epsilon})}$, and $K$ be defined in Eq.~\eqref{eq:defK}.
	Then
	\[
	\sum_{m\in K} \left|\frac{1}{2^t} \sum_{j=0}^{2^t-1} \exp\left[2\pi i j (\varphi - 2^{-t}m)\right]\right|^2 \geq 1-\epsilon.
	\]
\end{lemma}
\begin{proof}
	See page 224 of~\cite{nielsen2002quantum}.
\end{proof}

\subsection{Order-finding}\label{sec:order}
Given positive co-prime integers $x$ and $N$, the order of $x$ modulo $N$ is the least positive integer $r$ such that $x^r \equiv_N 1$, where $\equiv_N$ denotes equality modulo $N$.
Let $L\define \ceil{\log(N)}$, $\epsilon \in (0,1)$, and $t\define2L+ 1 + \ceil{\log(2+ \frac{1}{2\epsilon})}$. The order-finding algorithm computes the order $r$ of $x$ by using $O(L^3)$ operations, with success probability at least $(1-\epsilon)/(2\log(N))$. The algorithm goes as follows:
\begin{align*}
	\mathit{OF}(x,N)\define&\\
	&\bar{q} := 0; \ \bar{q} \apply H^{\otimes t};\\	
	&\bar{q}' := 0; \ \bar{q}' \apply U_{+1};\\
	&\bar{q}, \bar{q}' \apply \mathit{CU};\\
	&\bar{q} \apply \mathit{QFT}(t)^\dag;\\
	&z' := \measure\ \bar{q};\\
	&z := f(z'/2^t)
\end{align*}
where $|\bar{q}| = t$, $|\bar{q}'| = L$, $U_{+1}$ is a unitary operator on $\h_{\bar{q}'}$ such that $U_{+1}|0\> = |1\>$, $f(x)$ is the continued fractions algorithm which computes all convergents $m/n$ of the continued fraction for $x$ with 
	$\left|m/n - x\right| < 1/(2n^2)$
	and returns the minimal $n$ if there is any, and $\mathit{CU}$ is the controlled-$U$ operator on $\h_{\bar{q}}\otimes \h_{\bar{q}'}$ such that $\mathit{CU}|j\>_{\bar{q}} |y\>_{\bar{q}'} = |j\>_{\bar{q}} U^j|y\>_{\bar{q}'}$, where for each $0\leq y< 2^L$,
\begin{equation}\label{eq:ofu}
	U|y\> = \left\{\begin{tabular}{ll}
		$|xy \mbox{ mod } N\>$ & if $y<N$\\
		$|y\>$ & otherwise.\\
	\end{tabular}
	\right.
\end{equation}
Note that $CU$ can be implemented using $O(L^3)$ basic quantum gates by employing the technique of modular exponentiation~\cite{shor1997}. For the sake of simplicity, we omit the detailed implementation of $CU$ in the description of $\mathit{OF}(x,N)$.

The correctness of $\mathit{OF}(x,N)$ can be stated as
\begin{equation}\label{eq:OF}
	\vdash_{\mathit{tot}} \{p_{\mathit{OF}} \cdot (\gcd(x,N) =1)\} \mathit{OF}(x,N) \{z=r\}
\end{equation}
for some $p_{\mathit{OF}}\geq (1-\epsilon)/(2\log(N))$. For each $0\leq s < r$, let
$$|u_s\> \define \frac{1}{\sqrt{r}} \sum_{k=0}^{r-1} e^{-2\pi i s k/r} |x^k \mbox{ mod } N\>.$$
Then $|u_s\>$'s are orthonormal, $U|u_s\> = e^{2\pi i s /r}|u_s\>$, and 
$
1/\sqrt{r} \sum_{s=0}^{r-1} |u_s\> = |1\>.
$
We compute
\begin{align*}
	&\{\gcd(x,N) =1\}\\
	&\bar{q} := 0; \ \bar{q} \apply H^{\otimes t};\ \bar{q}' := 0;\ \bar{q}' \apply U_{+1};\\	
	&\left\{\left\<\gcd(x,N) =1,   |+^{\otimes t}\>_{\bar{q}} |1\>_{\bar{q}'}\right\>\right\}& \mathit{(Init, Unit)}\\
	&\left\{\left\<\gcd(x,N) =1, \frac{1}{\sqrt{r2^t}} \sum_{s=0}^{r-1}\sum_{j=0}^{2^t-1}|j\>_{\bar{q}} |u_s\>_{\bar{q}'}\right\>\right\}& \mathit{(Imp)}\\
	&\bar{q}, \bar{q}' \apply \mathit{CU};\\
	&\left\{\left\<\gcd(x,N) =1, \frac{1}{\sqrt{r2^t}} \sum_{s=0}^{r-1}\sum_{j=0}^{2^t-1} e^{2\pi i js/r} |j\>_{\bar{q}} |u_s\>_{\bar{q}'}\right\>\right\}& \mathit{(Unit)}\\
	&\bar{q} \apply \mathit{QFT}(t)^\dag;\\
	&\left\{\left\<\gcd(x,N) =1, \frac{1}{\sqrt{r}2^t} \sum_{s=0}^{r-1} \sum_{k=0}^{2^t-1} \sum_{j=0}^{2^t-1} \exp\left[2\pi i j\left(\frac{s}{r}-\frac{k}{2^t}\right)\right] |k\>_{\bar{q}} |u_s\>_{\bar{q}'}\right\>\right\}& \mathit{(Unit)}
\end{align*}
Furthermore, for any $0\leq s < r$ with $\gcd(s, r) = 1$, let 
\[
	K_s\define \left\{0\leq k < 2^t :  \left|\frac{s}{r}-\frac{k}{2^t}\right| < \frac{1}{2^{2L+1}}\right\}.
\]
Then from~\cite[Theorem 5.1]{nielsen2002quantum} and \cite{hardy1979introduction}, for each $k\in K_s$ the continued fractions algorithm $f$ in $\mathit{OF}(x,N)$ computes $f(k/2^t) = r$.
Thus
\begin{align*}
	&\left\{\left\< \gcd(x,N) =1, \sum_{s: \gcd(s,r) =1} \sum_{k\in K_s} |k\>_{\bar{q}}\<k| \otimes |u_s\>_{\bar{q}'}\<u_s| \right\>\right\}\\
	&\left\{\sum_{s=0}^{r-1}\sum_{k=0}^{2^t-1}\left\<  f(k/2^t)=r, |k\>_{\bar{q}}\<k| \otimes |u_s\>_{\bar{q}'}\<u_s| \right\>\right\} & \mathit{(Imp)}\\
	&z' := \measure\ \bar{q};\\
	&\{f(z'/2^t)=r\} & \mathit{(Meas0, Imp)}\\
	&z := f(z'/2^t)\\
	&\{z=r\} & \mathit{(Assn)}
\end{align*}
Then by (ProbComp), Eq.(\ref{eq:OF}) holds where
\begin{align*}
	p_{\mathit{OF}}  &= \frac{1}{r} \sum_{s: \gcd(s,r) =1} \sum_{k\in K_s}  \left| \frac{1}{2^t}  \sum_{j=0}^{2^t-1} \exp\left[2\pi i j\left(\frac{s}{r}-\frac{k}{2^t}\right)\right]\right|^2\\
	& \geq \frac{1}{r}  \sum_{s: \gcd(s,r) =1}\left(1-\epsilon\right)\geq \frac{1-\epsilon}{2\log(N)}
\end{align*}
where the first inequality is from Lemma~\ref{lem:peprob} and the last one from the fact that there are at least $r/(2\log(r))$ prime numbers less than $r\leq N$.

Note that we can check easily (in $O(L^3)$ time using, say, modular exponentiation) whether or not an output of $\mathit{OF}(x,N)$ is indeed the order of $x$ modulo $N$. 
By repeating the above algorithm $O(L)$ times we can further increase the success probability to $1-\epsilon$. Actually, the $ 1-\epsilon$ success probability can be achieved without introducing the $O(L)$ overhead, by only repeating $\mathit{OF}(x,N)$ a constant number of times and taking the least common multiple of the outputs~\cite{nielsen2002quantum}.

\subsection{Shor's factorisation algorithm}

Given a positive integer $N$ which is composite, the factorisation problem asks to find all the factors of $N$. No classical algorithm can solve this problem in polynomial (in $\ceil{\log(N)}$, the number of bits to encode $N$) time. The difficulty of this problem is at the heart of many widely used cryptographic algorithms such as RSA~\cite{rivest1978method}.

One of the killer apps of quantum computing is Shor's algorithm~\cite{shor1994algorithms}, which solves the factorisation problem (actually, a polynomial-time equivalent one which finds a non-trivial factor of $N$) in $O(\log^3(N))$ time, achieving an exponential speed-up over the best classical algorithms. Shor's algorithm uses the order-finding algorithm as a subroutine in an inline manner, and is depicted in Table~\ref{tbl:shor} (left column), where $\mathit{Unif}(1,N-1)$ is the uniform distribution over $\{1, \cdots, N-1\}$. 

Let $F(y)\define (y\mbox{ is a non-trivial factor of } N) \equiv [1<y<N \wedge (y\ div\ N)],$ and 
$$cmp(N) \define (N>2 \wedge N \mbox{ is odd and composite} \wedge N \neq a^b \mbox{ for any integers $a$ and $b>1$}).$$ Here we assume $N$ to be odd and not of the form $a^b$ for simplicity; otherwise, the  non-trivial factor 2 or $a$ of $N$ can be easily found.
Then the correctness of $\mathit{Shor}(N)$ can be stated as
\begin{equation}\label{eq:shorspec}
\vdash_{\mathit{tot}} \{ p_{\mathit{Shor}} \cdot cmp(N)\}\ \mathit{Shor}(N)\ \{F(y)\}
\end{equation}
for some success probability $p_{\mathit{Shor}}$.

\begin{table}[t]
	\begin{lrbox}{\tablebox}
		\centering
		\begin{tabular}{l|l}
\begin{tabular}{l}
$\mathit{Shor}(N)\define$\\
$x\rassign \mathit{Unif}(1,N-1);$\\
 $\iif\ \gcd(x, N) >1\ \then  $\\
$\quad y:= \gcd(x, N)$\\
 $\eelse$\\
$\quad {\mathit{OF}}(x, N);$\\
$\quad \iif\ (z\mbox{ is even} \wedge x^{z/2}\not\equiv_N -1)\ \then$\\
$\qquad {y_1:= \gcd(x^{z/2}-1, N)};$\\
$\qquad {y_2:= \gcd(x^{z/2}+1, N)}$\\
$\quad \eelse$\\
$\qquad \abort$\\
$\quad \pend$\\
 $\quad\iif\ (y_1\mbox{ non-trivial factor of } N)\ \then$\\
$\qquad y:= y_1$\\
$\quad \eelse$\\
$\qquad \iif\ (y_2\mbox{ non-trivial factor of } N)\ \then$\\
$\qquad \qquad y:= y_2$\\
$ \qquad \eelse$\\ 
$\qquad \qquad \abort$\\ 
$\qquad \pend$\\
$\quad \pend$\\
$\pend$
\end{tabular}
&
\begin{tabular}{lr}
	$\{p_{\mathit{OF}}(1-1/2^{m-1})  \cdot cmp(N) \}$&\\
	$x\rassign \mathit{Unif}(1,N-1);$&\\
	$\{cmp(N)   \wedge  \qassert\}$& main text, (Inv)\\
	$\iif\ \gcd(x, N) >1\ \then\ \{F(\gcd(x, N))\}$&  (Imp)\\
	$\quad y:= \gcd(x, N)\ \{F(y)\}$&\\ 
	$\eelse$&\\
	$\quad \{p_{\mathit{OF}} \cdot (cmp(N) \wedge E(x) \wedge \gcd(x, N) =1)\}$&\\
	$\quad {\mathit{OF}}(x, N);$&\\
	$\quad \{cmp(N)  \wedge E(x)  \wedge z = r_x\}$& Eq.(\ref{eq:OF}), (Inv)\\
	$\quad \iif\ (z\mbox{ is even} \wedge x^{z/2}\not\equiv_N -1)\ \then$&\\
	$\qquad \{F(\gcd(x^{z/2}-1, N))\vee F(\gcd(x^{z/2}+1, N))\}$& Eq.(\ref{eq:shor1})\\
	$\qquad {y_1:= \gcd(x^{z/2}-1, N)};$&\\
	$\qquad {y_2:= \gcd(x^{z/2}+1, N)}$&\\
	$\quad \eelse\ \{\bot\}\ \abort\ \pend$&\\
	$\quad \{F(y_1)\vee F(y_2)\}$&\\
	$\quad\iif\ F(y_1)\ \then$&\\
	$\qquad \{F(y_1)\}\ y:= y_1$&\\
	$\quad \eelse\ \{F(y_2)\}\ \iif\ F(y_2)\ \then$&\\
	$\qquad \qquad\{F(y_2)\}\ y:= y_2$ &\\
	$\qquad \quad \eelse\ \{\bot\}\ \abort\  \pend$&\\
	$\quad \pend\ \{F(y)\}$&\\
	$\pend\ \{F(y)\}$&\\
\end{tabular}
\end{tabular}
\end{lrbox}
\resizebox{\textwidth}{!}{\usebox{\tablebox}}
\vspace{2mm}
\caption{Shor's factorisation algorithm and its proof outline.}
\label{tbl:shor}
\end{table}
As the only non-classical part of Shor's algorithm is the order-finding subroutine, it can be verified by simply employing some theorems from number theory and the result from Sec.~\ref{sec:order}. To be specific, for any $x$ which is co-prime with $N$, let $r_x \define  \mathit{ord}(x,N)$ be the order of $x$ modulo $N$ and $E(x)\define (r_x\mbox{ is even} \wedge x^{r_x/2}\not\equiv_N -1)$.
The following lemma (see~\cite[Theorems 5.2 and 5.3]{nielsen2002quantum}, \cite{ekert1996quantum}) is crucial:

\begin{lemma}\label{lem:shor} Let $N>0$ be a composite integer.
	\begin{enumerate}
\item If $s$ is a non-trivial solution to the equation $s^2 \equiv_N 1$, then at least one of $\gcd(s-1, N)$ and $\gcd(s+1, N)$ is a non-trivial factor of $N$.
\item Let $m$ be the number of prime factors of $N$ and $N$ is odd. If $x$ is chosen uniformly at random from the set $\{1, \ldots, N-1\}$. Then  the conditional probability
\[
\mathrm{Pr}[E(x)\ |\ \gcd(x,N)=1] \geq 1-\frac{1}{2^{m-1}}.
\]
\end{enumerate}
\end{lemma}

Since $r_x$ is the order of $x$ modulo $N$, we have $x^{r_x/2}\not\equiv_N 1$. Thus the first clause of Lemma~\ref{lem:shor} implies
\begin{equation}\label{eq:shor1}
cmp(N) \wedge E(x) \rightarrow F(\gcd(x^{r_x/2}-1, N))\vee F(\gcd(x^{r_x/2}+1, N)).
\end{equation}
Furthermore, let $x$ be chosen uniformly at random from the set $\{1, \ldots, N-1\}$ and $\lambda\define\mathrm{Pr}(\gcd(x,N) >1)$. Then
from the second clause of Lemma~\ref{lem:shor} we have
\begin{align}
&\mathrm{Pr}[\gcd(x,N)>1]  + p_{\mathit{OF}} \cdot  \mathrm{Pr}[\gcd(x,N) = 1 \wedge E(x)] \notag\\
&\qquad\geq \lambda +  p_{\mathit{OF}} \cdot  (1-\lambda)\left(1-\frac{1}{2^{m-1}}\right)  \geq   p_{\mathit{OF}} \cdot  \left(1-\frac{1}{2^{m-1}}\right), \label{eq:shor}
\end{align}
and so
\begin{align*}
&\left\{p_{\mathit{OF}}\left(1-\frac{1}{2^{m-1}}\right) \right\}\\
&\left\{ \sum_{n=1}^{N-1} \left[\left\<\gcd(n,N) > 1, \frac{1}{N-1}\right\> +  \left\<\gcd(n,N) = 1 \wedge E(n),  \frac{p_{\mathit{OF}}}{N-1}\right\>\right]\right\}&\mathit{Eq.(\ref{eq:shor}), (Imp)}\\
&x\rassign \mathit{Unif}(1,N-1)\\
&\left\{(\gcd(x,N) > 1) +  \left\<\gcd(x,N) = 1 \wedge E(x),  p_{\mathit{OF}}\right\>\right\}. &\mathit{(Rassn)}
\end{align*}
The rest of the proof is sketched in the right column of Table~\ref{tbl:shor}, where $$\qassert \define (\gcd(x,N) > 1) +  \left\<\gcd(x,N) = 1 \wedge E(x),  p_{\mathit{OF}}\right\>.$$
Thus we have Eq.~(\ref{eq:shorspec}) with $p_{\mathit{Shor}} =p_{\mathit{OF}}(1-1/2^{m-1}) \geq p_{\mathit{OF}}/2$, by noting that $cmp(N)$ implies $m>1$. Furthermore, as stated in the previous subsection, this probability can be further increased to $1-\epsilon$ without increasing the time complexity of the algorithm.


\section{Conclusion}

We studied in this paper a simple quantum while-language where classical variables are explicitly involved. This language supports deterministic and probabilistic assignments of classical variables; initialisation, unitary transformation, and measurements of quantum variables; conditionals and while loops. Simultaneous initialisation of multiple quantum variables, and application of parametrised unitary operations on selected variables in a quantum register are also supported as syntactic sugars. These features make the description of practical quantum algorithms easy and compact, as shown by various examples.

With novel definition of cq-states and assertions, we defined for our language a small-step structural operational semantics, and based on it, a denotational one. Partial and total correctness of cq-programs were then introduced in the form of Hoare triples. We proposed Hoare-type logic systems for partial and total correctness respectively, and showed their soundness and relative completeness. 
Case studies including Grover's algorithm, quantum Fourier transformation, phase estimation, order finding, and Shor's algorithm illustrate the expressiveness of our language as well as the capability of the Hoare logic.

As future work, we would like to develop a software tool to implement the proof systems proposed in this paper, and use it to analyse more quantum algorithms and protocols from the area of quantum computation and communication. Another direction we are going to pursue is to extend our Hoare logic to classical-quantum languages with general recursion and procedure call. Finally, techniques of constructing invariants and ranking assertions for quantum loops are also interesting and important topics for further investigation.

\begin{acks}             
This work is partially supported by the National Key R\&D Program of China (Grant No: 2018YFA0306 701) and the Australian Research Council (Grant No: DP180100691). Y. F. also acknowledges the support of Center for Quantum Computing, Peng Cheng Laboratory, Shenzhen during his visit.
\end{acks}

\bibliography{ref}

\end{document}